\documentclass[10pt]{article}
\usepackage[utf8]{inputenc}
\usepackage{graphicx}
\usepackage[all]{xy}
\usepackage{color}
\usepackage{latexsym}
\usepackage{amsmath,amstext}
\usepackage{amssymb}
\usepackage{epsfig}
\usepackage{graphics}
\usepackage{euscript}
\usepackage{ifthen}
\usepackage{wrapfig}
\usepackage{amsthm}
\usepackage{multicol}
\usepackage{paralist}
\usepackage{verbatim}
\usepackage[colorlinks=true, urlcolor=blue, linkcolor=blue, citecolor=blue]{hyperref}
\usepackage[margin=1in]{geometry}
\usepackage{tikz,amsmath, amssymb,bm,color}
\usepackage{placeins}
\usepackage[linesnumbered,lined,commentsnumbered,ruled]{algorithm2e}
\usepackage{caption}
\usepackage[noabbrev]{cleveref}
\usepackage{longtable}
\usepackage{lscape}
\usepackage{fancyhdr}
\usepackage{listings}
\usepackage{mathtools}
\usepackage{appendix} 
\usepackage{float}
\usepackage{textcomp}
\usepackage[inline]{enumitem}
\usepackage{tikz}
\usepackage{multicol}
\usepackage{tkz-euclide,subfigure}
\usepackage{subfigure}

\usepackage{chemformula}

\usepackage{rotating}
\usepackage{url}
\makeatletter
\g@addto@macro{\UrlBreaks}{\UrlOrds}
\makeatother

\usetikzlibrary{arrows,decorations.pathmorphing,backgrounds,positioning,fit,petri}

\usepackage[utf8]{inputenc}
\usepackage{csquotes}
\usepackage{xparse}

\usepackage[normalem]{ulem}

\theoremstyle{plain}
\newtheorem{theorem}{Theorem}[section]
\newtheorem{corollary}[theorem]{Corollary}
\newtheorem{lemma}[theorem]{Lemma}
\newtheorem{proposition}[theorem]{Proposition}

\theoremstyle{definition}
\newtheorem{definition}[theorem]{Definition}

\newtheorem{example}[theorem]{Example}
\newtheorem{question}[theorem]{Question}

\numberwithin{equation}{section}

\begin{document}
\title{Distinguishing Level-2 Phylogenetic Networks Using Phylogenetic Invariants}
\author{Muhammad Ardiyansyah}

\maketitle
\begin{abstract}
\noindent
    In phylogenetics, it is important for the phylogenetic network model parameters to be identifiable so that the evolutionary histories of a group of species can be consistently inferred. However, as the complexity of the phylogenetic network models grows, the identifiability of network models becomes increasingly difficult to analyze. As an attempt to analyze the identifiability of network models, we check whether two networks are distinguishable. In this paper, we specifically study the distinguishability of phylogenetic network models associated with level-2 networks. Using an algebraic approach, namely using discrete Fourier transformation, we present some results on the distinguishability of some level-2 networks, which generalize earlier work on the distinguishability of level-1 networks. In particular, we study simple and semisimple level-2 networks. Simple and semisimple level-2 networks can be thought as generalizations of level-1 sunlet and cycle networks, respectively. Moreover, we also compare the varieties associated with semisimple level-2 and cycle networks. 
\end{abstract}

\section{Introduction}\label{introduction}

Phylogenetics is a field in biology that mainly concerns with the evolutionary history of organisms. Using data from existing species, we would like to infer the phylogenetic tree, a tree that best describes the relationship between them. Some standard references about phylogenetics are \cite{felsenstein2004inferring,semple2003phylogenetics}. Given some biological data, the evolutionary process of a set of species can be studied using a mathematical model, called phylogenetic model. The most well-studied phylogenetic models are the nucleotide substitution models, which are Markov models that describe evolutionary changes over time. Recently mathematicians have explored both algebraic and geometric methods in studying phylogenetics. Interested readers can refer to \cite{allman2003phylogenetic,banos2016phylogenetic,MartaJesus07, CavenderFelsenstein,gross2018distinguishing, EvansSpeed, PachterSturmfels,sturmfels2005toric} for a list of publications in phylogenetic using both algebraic and geometric approaches.

Although the notion of a phylogenetic model on a tree is fairly simple, this notion demonstrates many desirable features in applications. A phylogenetic tree can explain well the speciation events, in which one ancestral lineage gives rise to two or more lineages, and the descent with modification events occurring over a period of time, in which genetic differences could be passed on to the next generation. On the other hand, a phylogenetic tree may fail to model hybridization, recombination, or horizontal gene transfer \cite{maddison1997gene,pamilo1988relationships,sneath1975cladistic}. These reticulate mechanisms are thought to play a crucial role in shaping evolutionary histories \cite{doolittle1999phylogenetic}. Phylogenetic networks come into play in phylogenetics to take into account these reticulation events. Instead of using a tree, a phylogenetic network models the evolutionary process on a connected directed acyclic graph (DAG) that contains some edges representing reticulation events. In this paper, we focus on a well-studied class of networks, namely level-$k$ networks \cite{jansson2006inferring}, particularly for $k=2$.

We aim to study networks that take into account reticulation events represented by reticulation edges. These networks belong to the class of partly-directed networks, which are graphs that may contain both undirected and directed edges. There are several publications in the biological literature that consider partly-directed networks. These publications include the descriptions of the evolutionary history of humans in Europe \cite{lazaridis2018evolutionary}, of grapes \cite{myles2011genetic}, and of bears \cite{kumar2017evolutionary}. From the mathematical perspectives, a class of partly-directed networks called semi-directed networks has been introduced in \cite{gross2018distinguishing}. In this class of networks, the only directed edges are the edges representing the reticulation events. In this paper, we are mainly interested to reconstruct the underlying semi-directed network topology of a phylogenetic network as semi-directed networks may provide information about the reticulation events among the given organisms.

Due to the growing importance of networks, the recent literature includes several papers on network inferences. Some of the most important inference methods for phylogenetic trees cover the sequence-based methods \cite{farris1970methods,fitch1971toward,chor2005maximum,felsenstein2004inferring} and the distance-based methods \cite{sokal1958statistical,saitou1987neighbor}. Recently, these methods have been adapted to the general network setting. In \cite{hein1990reconstructing}, the parsimony principles have been taken into account in reconstructing networks. Various works in the past already examined the likelihood estimation method for general networks \cite{jin2006maximum, strimmer2000likelihood,von1993network}. Finally, a distance-based method for a class of networks is proven to be useful in recovering the true phylogeny \cite{bordewich2018recovering, van2020reconstructibility}. In addition to these methods, combinatorial methods can reconstruct a network using the sets of its subnetworks. In \cite{van2009constructing}, the set of triplets, which are the phylogenetic networks on three leaves, is considered. Similarly, the set of quarnets, which are four leaf-networks, is studied in \cite{huebler2019constructing} and a set of subnetworks obtained by deleting some edges of the network is considered in \cite{murakami2019reconstructing}.

In this paper, we will look more closely at the network reconstruction using phylogenetic invariants. Phylogenetic invariants are polynomials associated with a phylogenetic network model. A reconstruction approach that uses phylogenetic invariants for a tree has been studied in \cite{allman2010identifiability,sturmfels2005toric}. Recently, network inference using phylogenetic invariants was considered in \cite{gross2018distinguishing} to distinguish a class of networks called cycle networks. 

To achieve the main goal in reconstructing the true phylogenetic network given some biological data, we need to study the identifiability of phylogenetic network models. Roughly speaking, model identifiability ensures that given some biological data, it is possible to recover the unique network topology and model parameters that produce a probability distribution that best represents our biological data. The identifiability of model parameters in a basic model of evolution is studied in \cite{chang1996full}. Moreover, the identifiability of model parameters on some more complex phylogenetic models, which take into account some specific information about the evolution, has been established \cite{allman2008identifiability,allman2006identifiability}. 

We would like to employ an algebraic approach, namely using the discrete Fourier transformation, to study the identifiability of phylogenetic network model. For these algebraic models, we will study the generic identifiability of the model, which is a slightly weaker notion than the (global) identifiability. In the network setting, recent works in \cite{gross2018distinguishing,gross2020distinguishing} provide generic identifiability results for level-1 triangle-free networks. These results are obtained by first checking that the two given networks are distinguishable.  Roughly speaking, two networks are distinguishable if the intersection of varieties associated with each network model is a proper subvariety of both network varieties. No work has been done in obtaining generic identifiability or distinguishability results for higher-level networks. This paper attempts to provide some distinguishability results of phylogenetic network models associated with level-2 networks. We focus on a class of networks which permits us to use the discrete Fourier transformation to obtain phylogenetic invariants associated with each network model. A network belonging to this class will be called nice (see \Cref{subsection:nice phylogenetic networks} for the precise definition of this term). In addition to this class of nice networks, we limit our study to the class of simple and semisimple networks, which are generalizations of level-1 sunlet and cycle networks. More detailed definitions of simple and semisimple networks will be given in \Cref{section:preliminaries} and \Cref{section:semisimple}, respectively. 

The main results of this paper will be given by the following series of theorems. These results are derived under some conditions on the set of reticulation leaves or branches. More detailed definitions of reticulation leaves and branches will be given in \Cref{section:preliminaries} and \Cref{section:semisimple}, respectively. Roughly speaking, a leaf is said to be a reticulation leaf or is contained in a reticulation branch, if the corresponding species evolves from a reticulate evolutionary process. For the case of simple nice level-2 networks, under some assumptions on the set of reticulation leaves, the main results will be given in \Cref{proposition:distinguishability r(N_1)=r(N_2)=2}, \Cref{proposition:distinguishability r(N_1)>=r(N_2)}, and \Cref{proposition:distinguishability r(N_1)=r(N_2)=1}.  \Cref{proposition:distinguishability r(N_1)=r(N_2)=2} provides distinguishability of two networks, each with exactly two reticulation leaves. \Cref{proposition:distinguishability r(N_1)>=r(N_2)} presents distinguishability of two networks: one has two reticulation leaves and the other has only one reticulation leaf. \Cref{proposition:distinguishability r(N_1)=r(N_2)=1} provides distinguishability of two networks, each with exactly one reticulation leaf. For the case of semisimple nice level-2 networks that are not level-1 networks, under some assumptions on the set of reticulation branches, the main results will be given in  \Cref{proposition:distinguishability semisimple R(N_1)=R(N_2)=2}, \Cref{proposition:distinguishability semisimple R(N_1)>=R(N_2) (1)}, and \Cref{proposition:distinguishability semisimple R(N_1)=R(N_2)=1}. \Cref{proposition:distinguishability semisimple R(N_1)=R(N_2)=2} provides distinguishability of two networks, each with exactly two reticulation branches. \Cref{proposition:distinguishability semisimple R(N_1)>=R(N_2) (1)} presents distinguishability of two networks: one has two reticulation branches and the other has only one reticulation branch. \Cref{proposition:distinguishability semisimple R(N_1)=R(N_2)=1} provides distinguishability of two networks, each with exactly one reticulation branch. Moreover, \Cref{theorem:comparison-level-2-and-level-1} and \Cref{thm:semisimple level-2 and level-1} compares the network varieties associated with nice level-2 and level-1 networks for the class of simple and semisimple networks, respectively.

The outline of the paper is organized as follows. \Cref{section:preliminaries} contains a brief summary of phylogenetic networks which will be used throughout the paper. The definition of level-$k$ and simple networks will be provided in this section. In \Cref{section:phylogenetic network model}, we will describe the phylogenetic models we take into consideration. Moreover, we will present an algebraic approach to distinguish networks using discrete Fourier transform and phylogenetic invariants associated with each network model. We will introduce the class of nice networks which enable us to perform discrete Fourier transform to distinguish networks. \Cref{section:semisimple} introduces the notion of semisimple level-2 networks and discusses some of their properties. This class of semisimple networks is a generalization of the class of simple networks. In \Cref{section:distinguishing nice networks}, we will look more closely at the 4-leaf nice level-2  networks. A few key results about the varieties of 4-leaf nice level-2 networks will be presented. These results will be our basic building blocks in proving the main results of the paper.  \Cref{subsub:Distinguishing simple level-2 networks with at least five leaves} provides some results on the distinguishability of some simple nice networks with at least five leaves using the results in \Cref{section:distinguishing nice networks}. Moreover, \Cref{section:Distinguishing semisimple level-2 networks with at least five leaves and beyond} provides some results on the distinguishability of some semisimple nice networks with at least five leaves, which generalize the results in \Cref{subsub:Distinguishing simple level-2 networks with at least five leaves}. Finally, in the last section, we will discuss some limitations of this paper and some possible directions for future research.

\section{Preliminaries and definitions}\label{section:preliminaries}

In this section, the notion of phylogenetic networks and some of their properties will be discussed. We will also present some graph-theoretic terminologies, which are adapted from \cite{bouvel2020counting,huebler2019constructing,gross2018distinguishing} and will be used throughout the paper. From now on, we assume that the set $X$ is a finite set corresponding to taxa or a set of species. In particular, we will assume that $X=[n]$, which is the set of natural numbers up to $n$ where $n\geq 2$. 

In a graph or network, we call an edge \textit{cut-edge} if by removing this edge, the network becomes disconnected. A cut-edge is said to be \textit{trivial} if one part of the induced  partition of the leaves is a single leaf.  A \textit{root} of a directed graph is a distinguished vertex of a graph with in-degree zero and out-degree two. Furthermore, the set of vertices with in-degree one and out-degree zero is called the \textit{leaves}. A subtree component consisting of two leaves adjacent to a vertex incident to a non-trivial cut-edge is called a \textit{tree cherry}. If a graph contains no vertex whose removal disconnects the graph, then we say that the graph is \textit{biconnected}. A \textit{biconnected component} of a graph $N$ is a maximal biconnected subgraph of $N$.

\begin{definition}[\cite{gross2018distinguishing}, Section 2]\label{def:rooted phylogenetic network}

 A \textit{rooted phylogenetic network} $N$ on $X$ is a rooted connected DAG with no parallel edges such that: $N$ has a single root; the leaf set of $N$ is bijectively labeled by $X$; and the remaining vertices either have in-degree one and out-degree two or in-degree two and out-degree one.
\end{definition}
In this network setting, 
a vertex with in-degree one and out-degree two is called a \textit{tree vertex} and it represents a speciation event. On the other hand,  a vertex with in-degree two and out-degree one vertex is called a \textit{reticulation vertex} and it represents a reticulation event. Additionally, \textit{reticulation edges} are the two edges directed into a reticulation vertex and the remaining edges are called \textit{tree edges}. For illustration purpose of a rooted phylogenetic network, a tree edge is represented by a solid arrow while a reticulation edge is represented by a dashed arrow. 

In a rooted phylogenetic network, we say that a vertex $u$ is a \textit{child} of vertex $v$ or $v$ is a \textit{parent} of $u$ if there exists a directed edge going from $v$ to $u$. In the rest of the paper, the unique parent of a leaf labeled by $x\in X$ will be denoted by $up(x)$.
We say that a vertex $u$ is a \textit{descendant} of vertex $v$ if there exists a directed path from the root to the vertex $u$ containing the vertex $v$. Alternatively, we could also say that the vertex $v$ is an \textit{ancestor} of the vertex $u$ to mean the same thing. In addition, we say that a leaf vertex $u$ is a \textit{reticulation leaf}  if there exists a reticulation vertex $v$ such that $v$ is a parent of $u$. We call a vertex $u$ \textit{purely interior vertex} if all vertices adjacent to $u$ are not leaves. In addition to the tree cherry, a subset $\{x,y\}\subseteq X$ is said to be a \textit{reticulated cherry} if there exists an undirected path $(x,u,v,y)$ such that $u$ is the parent of $x$ and $v$ is the parent of $y$ such that exactly one of $u$ or $v$ is a reticulation vertex.

Similar to \Cref{def:rooted phylogenetic network}, we could define an unrooted phylogenetic network.

\begin{definition}[\cite{fischer2020classes}, Section Methods]
 An \textit{unrooted phylogenetic network} $N$ on $X$ is a connected undirected graph without loops and parallel edges such that each vertex either has degree three, which we call \textit{internal vertices}, or degree one, which we call \textit{leaves}, and its leaf set is bijectively labeled by $X$. 
\end{definition}

Two rooted (unrooted) phylogenetic networks are said to be \textit{isomorphic} if they are isomorphic as rooted (unrooted) directed (undirected) leaf-labeled graphs. We then define a class of phylogenetic network based on the number of reticulation vertices in each biconnected component as follows.

\begin{definition}[\cite{van2009constructing,bouvel2020counting}, Section 2]
A \textit{rooted level-k phylogenetic network} is a rooted phylogenetic network such that each biconnected component has at most $k$ reticulation vertices. Similarly, an \textit{unrooted level-k phylogenetic network} is an unrooted phylogenetic network such that removing at most $k$ edges in each biconnected component and contracting each degree two vertex to one of its neighbors produce an unrooted phylogenetic tree. Furthermore, we say that a network is a \textit{strict level-k phyologenetic network} if it is a level-$k$ network but not a level-$(k-1)$ network. 
\end{definition}

We could not identify the root location in a network because the phylogenetic models we are considering are time reversible (see \Cref{section:phylogenetic network model}). Time reversible model means that in the model associated with the reversed process, the frequency of transitions and transversions of DNA bases remains unchanged. For this reason, we wish to reconstruct the semi-directed network topology of a phylogenetic network. The \textit{semi-directed
network} topology of a rooted phylogenetic network is obtained by first collapsing the root
and then we keep the direction of the reticulation edges but forget the direction of all tree edges. 

\begin{example}\label{ex:level-k example}
In the semi-directed network $N_1$, the set of reticulation vertices is $\{a,b\}$ and the set of reticulation edges is $\{e_1,\dots,e_4\}$ while in $N_2$, the set of reticulation vertices is $\{a,b,c,d\}$ and the set of reticulation edges is $\{e_1,\dots,e_8\}$. In $N_1$, there are two biconnected components, each with exactly one reticulation vertex while in $N_2$, there are two biconnected components, each with exactly two reticulation vertices. Additionally, $N_1$ contains a tree cherry on $\{c,d\}$ but no reticulated cherry while $N_2$ contains a reticulated cherry on $\{e,f\}$ but no tree cherry.

 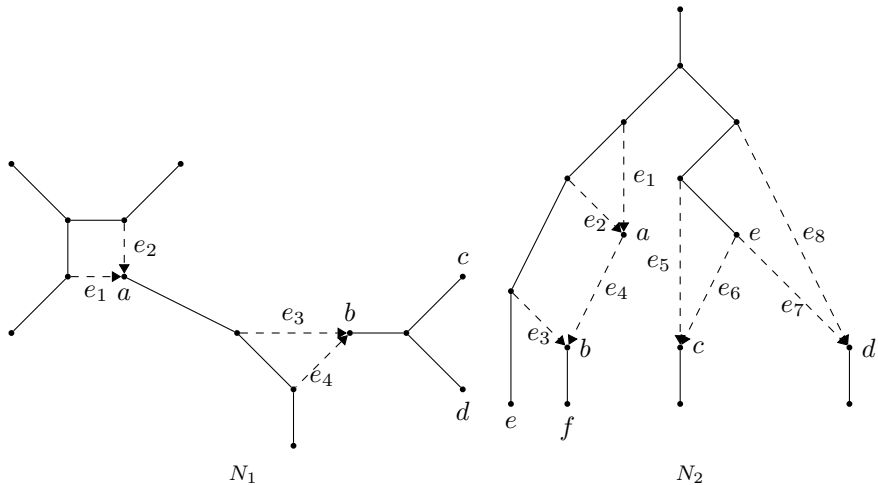
\begin{figure}[H]
    \centering
    \subfigure[$N_1$]{
    \begin{tikzpicture}[scale=0.75]
      \node (a1) at (0,0)[circle , draw, fill=black!50, inner sep=0pt, very thick ]{};
      \node [label=below:$a$] (a2) at (1,0)[circle , draw, fill=black!50, inner sep=0pt, very thick ]{};
      \node (a3) at (-1,-1)[circle , draw, fill=black!50, inner sep=0pt, very thick ]{};
      \node (a4) at (0,1)[circle , draw, fill=black!50, inner sep=0pt, very thick ]{};
      \node (a5) at (1,1)[circle , draw, fill=black!50, inner sep=0pt, very thick ]{};
      \node (a6) at (-1,2)[circle , draw, fill=black!50, inner sep=0pt, very thick ]{};
      \node (a7) at (2,2)[circle , draw, fill=black!50, inner sep=0pt, very thick ]{};
      \node (a8) at (3,-1)[circle , draw, fill=black!50, inner sep=0pt, very thick ]{};
      \node [label=above:$b$] (a9) at (5,-1)[circle , draw, fill=black!50, inner sep=0pt, very thick ]{};
      \node (a10) at (4,-2)[circle , draw, fill=black!50, inner sep=0pt, very thick ]{};
      \node (a11) at (4,-3)[circle , draw, fill=black!50, inner sep=0pt, very thick ]{};
      \node (a12) at (6,-1)[circle , draw, fill=black!50, inner sep=0pt, very thick ]{};
      \node [label=below:$d$](a13) at (7,-2)[circle , draw, fill=black!50, inner sep=0pt, very thick ]{};
      \node [label=above:$c$](a14) at (7,0)[circle , draw, fill=black!50, inner sep=0pt, very thick ]{};
      \draw [->,>=Triangle,dashed](a1)--(a2) node[midway, below]{$e_1$};
      \draw (a1)--(a4);
      \draw [->,>=Triangle,dashed](a5)--(a2) node[midway, right]{$e_2$};
      \draw (a5)--(a4);
      \draw (a1)--(a3);
      \draw (a4)--(a6);
      \draw (a5)--(a7);
      \draw (a8)--(a2);
      \draw [->,>=Triangle,dashed](a8)--(a9)node[midway, above]{$e_3$};
      \draw (a8)--(a10);
      \draw [<-,>=Triangle,dashed](a9)--(a10) node[midway, below]{$e_4$};
      \draw (a10)--(a11);
      \draw (a9)--(a12);
      \draw (a12)--(a13);
      \draw (a12)--(a14);
    \end{tikzpicture}}
    \subfigure[$N_2$]{
    \begin{tikzpicture}[scale=0.75]
      \node (a1) at (0,0)[circle , draw, fill=black!50, inner sep=0pt, very thick ]{};
     \node (a2) at (1,-1)[circle , draw, fill=black!50, inner sep=0pt, very thick ]{};
     \node (a3) at (-1,-1)[circle , draw, fill=black!50, inner sep=0pt, very thick ]{};
     \node (a4) at (0,-2)[circle , draw, fill=black!50, inner sep=0pt, very thick ]{};
     \node (a5) at (-2,-2)[circle , draw, fill=black!50, inner sep=0pt, very thick ]{};
     \node [label=right:$a$](a6) at (-1,-3)[circle , draw, fill=black!50, inner sep=0pt, very thick ]{};
     \node  [label=right:$e$](a7) at (1,-3)[circle , draw, fill=black!50, inner sep=0pt, very thick ]{};
     \node (a8) at (-3,-4)[circle , draw, fill=black!50, inner sep=0pt, very thick ]{};
     \node [label=right:$b$](a9) at (-2,-5)[circle , draw, fill=black!50, inner sep=0pt, very thick ]{};
     \node [label=right:$c$](a10) at (0,-5)[circle , draw, fill=black!50, inner sep=0pt, very thick ]{};
     \node [label=right:$d$](a11) at (3,-5)[circle , draw, fill=black!50, inner sep=0pt, very thick ]{};
     \node [label=below:$e$](a12) at (-3,-6)[circle , draw, fill=black!50, inner sep=0pt, very thick ]{};
     \node [label=below:$f$](a13) at (-2,-6)[circle , draw, fill=black!50, inner sep=0pt, very thick ]{};
     \node (a14) at (0,-6)[circle , draw, fill=black!50, inner sep=0pt, very thick ]{};
     \node (a15) at (3,-6)[circle , draw, fill=black!50, inner sep=0pt, very thick ]{};
     \node (a16) at (0,1)[circle , draw, fill=black!50, inner sep=0pt, very thick ]{};
     \draw (a1)--(a3);
     \draw (a1)--(a2);
     \draw (a3)--(a5);
     \draw[->,>=Triangle,dashed](a3)--(a6) node[midway, right]{$e_1$};
      \draw (a2)--(a4);
       \draw [->,>=Triangle,dashed](a2)--(a11)node[midway, right]{$e_8$};
      \draw[->,>=Triangle,dashed](a5)--(a6) node[midway, below]{$e_2$};
       \draw (a5)--(a8);
    \draw[->,>=Triangle,dashed] (a8)--(a9) node[midway, below]{$e_3$};
    \draw [->,>=Triangle,dashed](a6)--(a9)node[midway, right]{$e_4$};
     \draw[->,>=Triangle,dashed] (a4)--(a10)node[midway, left]{$e_5$};
      \draw (a4)--(a7);
       \draw[->,>=Triangle,dashed] (a7)--(a10)node[midway, right]{$e_6$};
        \draw[->,>=Triangle,dashed] (a7)--(a11)node[midway, below]{$e_7$};
         \draw (a9)--(a13);
          \draw (a8)--(a12);
           \draw (a10)--(a14);
            \draw (a11)--(a15);
            \draw (a16)--(a1);
    \end{tikzpicture}}
    \caption{The semi-directed network $N_1$ is a strict level-1 network with two reticulation vertices while $N_2$ is a strict level-2 network with four reticulation vertices.}
    \label{fig:my_label}
\end{figure}
\end{example}

Following the terminology used in \cite{bouvel2020counting,van2009constructing}, we will first consider a class of level-2 networks that is called \textit{simple level-2 networks}. These simple level-2 networks is a generalization of sunlet networks, which are level-1. Furthermore, this class of networks can be thought to be the basic building blocks of level-2 networks. Namely, each biconnected component of
a level-2 network can be approximated by a simple level-2 network.  To begin with, we introduce level-$k$ generators as follows.

\begin{definition}[\cite{van2009constructing}, Section 3]
 For $k\geq 1$, a \textit{rooted level-$k$ generator} is a biconnected directed acyclic graph with parallel edges that has exactly $k$ reticulation vertices, has a single root, and the remaining vertices are tree vertices. 
 \end{definition}

\begin{definition}[\cite{huber2016transforming}, Lemma 6]\label{lemma:characterization of generator}
An \textit{unrooted level-$l$ generator} $L$ is defined as the following multigraph.
\begin{enumerate}
    \item If $l=0$, then $L$ is a single vertex.
    \item If $l=1,$ then $L$ is the 2-regular multigraph with 2 vertices.
    \item If $l>1$, then $L$ is a 3-regular biconnected multigraph with $2l-2$ vertices.
\end{enumerate}
\end{definition}
\begin{example}
It is shown in \cite{van2009constructing} that there are precisely four rooted and one unrooted level-2 generators. For the rooted case, the root is represented by $E$ while the vertices $C$ or $D$ represent reticulation vertices. Network of type $8a$ and $8d$ contain exactly one tree vertex while network of type $8b$ and $8c$ contain exactly two tree vertices.  

\begin{figure}[H]
    \centering
\subfigure[The unrooted level-2 generator $L_2$]{
    \begin{tikzpicture}[scale=0.5]
    \node (a1) at (0,0)[left]{$A$} ;  
  \node (a2) at (4,0)[right]{$B$};  
\draw (0,0) to[out=90,in=90] node[midway,above] {} (4,0);
\draw (0,0) -- node[midway,below] {} (4,0);
\draw (0,0) to[out=-90,in=-120] node[midway,above] {} (4,0);
    \end{tikzpicture}}
    \qquad 
    \subfigure[Four rooted level-2 generators: type $8a, 8b, 8c,\mbox{and }8d$ (from left to right).]{\
    \begin{tikzpicture}[scale=0.75]
  \node [label=above:$E$](a1) at (0,0)[circle , draw, fill=black!50, inner sep=0pt, very thick ]{} ;  
  \node (a2) at (-1,-1)[circle , draw, fill=black!50, inner sep=0pt, very thick ]{};  
  \node (a3) at (1,-2)[circle , draw, fill=black!50, inner sep=0pt, very thick ]{};  
  \node [label=below:$C$](a4) at (0,-3)[circle , draw, fill=black!50, inner sep=0pt, very thick ]{};  
 
  \draw (a1)--(a2)node[midway, left]{};
  \draw (a1)--(a3)node[midway, right]{};
  \draw (a2)--(a4)node[midway, left]{};
  \draw (a2)--(a3)node[midway, above]{};
  \draw (a3)--(a4)node[midway, right]{};
\end{tikzpicture}
\begin{tikzpicture}[scale=0.75]
  \node [label=above:$E$](a1) at (0,0)[circle , draw, fill=black!50, inner sep=0pt, very thick ]{} ;  
  \node (a2) at (-1,-1)[circle , draw, fill=black!50, inner sep=0pt, very thick ]{};  
  \node (a3) at (0,-2)[circle , draw, fill=black!50, inner sep=0pt, very thick ]{};  
  \node [label=below:$C$](a4) at (-1,-3)[circle , draw, fill=black!50, inner sep=0pt, very thick ]{};  
  \node [label=below:$D$](a5) at (1,-3)[circle , draw, fill=black!50, inner sep=0pt, very thick ]{};
 
  \draw (a1)--(a2)node[midway, left]{};
  \draw (a1)--(a5)node[midway, right]{};
  \draw (a2)--(a3)node[midway, right]{};
  \draw (a3)--(a5)node[midway, left]{};
  \draw (a2)--(a4)node[midway, left]{};
  \draw (a3)--(a4)node[midway, right]{};
\end{tikzpicture}
\begin{tikzpicture}[scale=0.75]
  \node [label=above:$E$](a1) at (0,0)[circle , draw, fill=black!50, inner sep=0pt, very thick ]{} ;  
  \node (a2) at (-1,-1)[circle , draw, fill=black!50, inner sep=0pt, very thick ]{};  
  \node [label=below:$C$](a3) at (0,-1.5)[circle , draw, fill=black!50, inner sep=0pt, very thick ]{};  
  \node (a4) at (1,-1)[circle , draw, fill=black!50, inner sep=0pt, very thick ]{};  
  \node [label=below:$D$](a5) at (0,-3)[circle , draw, fill=black!50, inner sep=0pt, very thick ]{};
 
  \draw (a1)--(a2)node[midway, left]{};
  \draw (a2)--(a3)node[midway, above]{};
  \draw (a1)--(a4)node[midway, right]{};
  \draw (a4)--(a3)node[midway, above]{};
  \draw (a2)--(a5)node[midway, left]{};
  \draw (a4)--(a5)node[midway, right]{};
\end{tikzpicture}
\begin{tikzpicture}[scale=0.75]
  \node [label=above:$E$](a1) at (0,0)[circle , draw, fill=black!50, inner sep=0pt, very thick ]{} ;  
  \node (a2) at (-1,-1)[circle , draw, fill=black!50, inner sep=0pt, very thick ]{};  
  \node (a3) at (-1,-2)[circle , draw, fill=black!50, inner sep=0pt, very thick ]{};  
  \node [label=below:$C$](a4) at (0,-3)[circle , draw, fill=black!50, inner sep=0pt, very thick ]{};  
 
  \draw (a1)--(a2);
  \draw (a2) edge [bend left](a3);
  \draw (a3) edge [bend left] (a2);
  \draw (a1)--(a4);
  \draw (a3)--(a4);
\end{tikzpicture}
   }
   \caption{The level-2 generators.}
\end{figure}
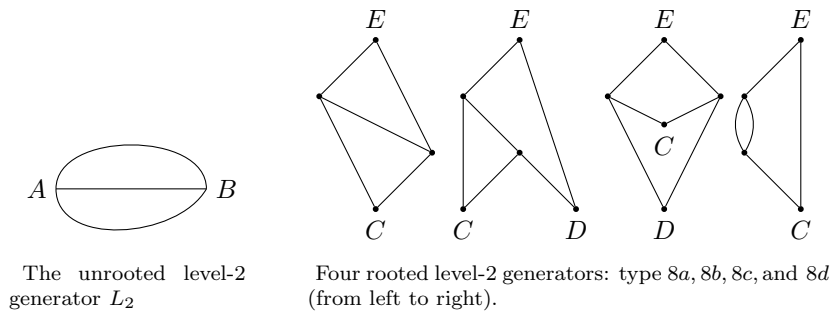

\end{example}
The following procedure describes how to obtain a simple level-$k$ network from a level-$k$ generator.

\begin{definition}[\cite{van2009constructing}, Section 3]\label{def:simplenetwork}
For $k\geq 1$, we can obtain a \textit{rooted/unrooted simple level-k network} by implementing the
following procedure to some rooted/unrooted level-$k$ generator $L$. We initially replace each edge of $L$ by a path graph and then we add a new leaf $x$ and a leaf edge $(u,x)$ for every internal vertex $u$ of the path. Additionally, for the rooted case, we add a leaf and a leaf edge for each vertex $v$ of in-degree two and out-degree zero.
\end{definition}
 An important characterization of a simple level-$k$ network is that the network should not have any non-trivial cut-edges \cite[Lemma 6]{van2009constructing}. As an example, if we forget the direction of all reticulation edges in the two networks in \Cref{ex:level-k example}, then each of them is not a simple unrooted network.

\begin{example}\label{ex:unrooted-level2-two and three leaves}
Up to isomorphism and labeling of the leaves, the following undirected network topologies are  the complete set of distinct unrooted simple level-2 network topologies with two and three leaves.
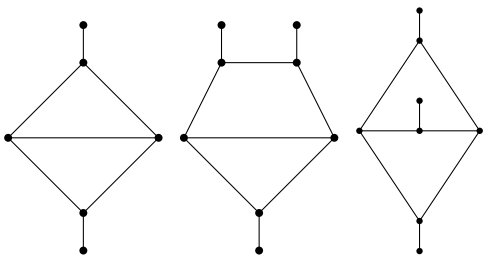
\begin{figure}[H]
    \centering
    \subfigure[]{
    \begin{tikzpicture}[scale=0.5]
       \fill[black] (0,0) circle (3pt) (4,0) circle (3pt) (2,2) circle (3pt)  (2,3) circle (3pt) (2,-2) circle (3pt) (2,-3) circle (3pt);
\draw (0,0) to node[midway,above] {} (2,2);
\draw (2,2) to node[midway,above] {} (4,0);
\draw (0,0) -- node[midway,below] {} (4,0);
\draw (0,0) to node[midway,above] {} (2,-2);
\draw (2,-2) to node[midway,above] {} (4,0);
\draw (2,2)--(2,3);
\draw (2,-2)--(2,-3);
    \end{tikzpicture}}
    \subfigure[]{
    \begin{tikzpicture}[scale=0.5]
       \fill[black] (0,0) circle (3pt) (4,0) circle (3pt) (1,2) circle (3pt)  (1,3) circle (3pt) (3,2) circle (3pt)  (3,3) circle (3pt) (2,-2) circle (3pt) (2,-3) circle (3pt) ;
\draw (0,0) to node[midway,above] {} (1,2);
\draw (1,2) to node[midway,above] {} (3,2);
\draw (3,2)--(4,0);
\draw (0,0) -- node[midway,below] {} (4,0);
\draw (0,0) to node[midway,above] {} (2,-2);
\draw (2,-2) to node[midway,above] {} (4,0);
\draw (1,2)--(1,3);
\draw (3,2)--(3,3);
\draw (2,-2)--(2,-3);
    \end{tikzpicture}}
     \subfigure[]{
    \begin{tikzpicture}[scale=0.4]
       \fill[black] (0,0) circle (3pt) (4,0) circle (3pt) (2,3) circle (3pt)  (2,4) circle (3pt) (2,-3) circle (3pt) (2,-4) circle (3pt) (2,0) circle (3pt)  (2,1) circle (3pt) ;
\draw (0,0) to node[midway,above] {} (2,3);
\draw (2,3) to node[midway,above] {} (4,0);
\draw (0,0) -- node[midway,below] {} (2,0);
\draw (2,0) to node[midway,above] {} (4,0);
\draw (0,0) to node[midway,above] {} (2,-3);
\draw (2,-3) to node[midway,above] {} (4,0);
\draw (2,3)--(2,4);
\draw (2,0)--(2,1);
\draw (2,-3)--(2,-4);
    \end{tikzpicture}}
    \caption{Example of unrooted simple level-2 network topologies with two and three leaves.}
\end{figure}
\end{example}

In the rest of the paper, we will only consider orientable networks due to their biological importance. Given two phylogenetic networks $N'$ and $N$ on $X$ where $N'$ is rooted and $N$ is unrooted, the network $N'$ is said to be an \textit{orientation}
 of $N$ if the network obtained from $N'$ by removing all the directions of the edges and suppressing its root is isomorphic to $N$ \cite{bouvel2020counting}. Additionally, if $N$ has at least one orientation, then $N$ is said to be \textit{orientable}.  In the rest of the paper, for simplicity, the word ‘unrooted orientable semi-directed network’ will be abbreviated to just ‘network’ .

We can obtain an $(n+1)$-leaf unrooted level-$k$ network $N'$ from an $n$-leaf rooted level-$k$ network $N$ by adding a vertex adjacent to the root of $N$ and then labeling it with an extra leaf label. \cite[Theorem 1]{gambette2012quartets} guarantees that $N'$ is also a level-$k$ network. Conversely, we can obtain a rooted level-$k$ network from an unrooted level-$k$ network  \cite{bouvel2020counting,janssen2018exploring} although the resulting rooted network needs not be unique. By definition, a level-$k$ semi-directed network is obtained from a rooted network. Alternatively, we can obtain a level-$k$ semi-directed network from a level-$k$ unrooted network. In order to do this, we must first place a root in a valid root location. Huber et al. \cite{huber2019rooting} provides an algorithm to find a valid root location. Then using this root, for each biconnected component, we can choose at 
most $k$ vertices to be the reticulation vertices and exactly $2k$ edges to be the reticulation edges, two for each reticulation vertex. Finally, we collapse the root vertex and undirect all tree edges.

\begin{example}\label{ex:unrooted-2-leaves}
Up to isomorphism and labeling the leaves, \Cref{fig:Two distinct $2$-leaf unrooted network} and \Cref{fig:Five distinct $3$-leaf network} present distinct unrooted orientable simple strict level-2 semi-directed networks with two and three leaves.
\begin{figure}[H]
    \centering
    \subfigure[]{
    \begin{tikzpicture}[scale=0.5]
       \fill[black] (0,0) node[circle,fill=black!3][above]{}  (4,0) circle (3pt) (2,2) circle (3pt)  (2,3) circle (3pt) (2,-2) circle (3pt) (2,-3) circle (3pt);
\draw [dashed,->,>=Triangle](0,0) to node[midway,above] {} (2,2);
\draw [dashed,<-,>=Triangle](2,2) to node[midway,above] {} (4,0);
\draw (0,0) -- node[midway,below] {} (4,0);
\draw [dashed,->,>=Triangle](0,0) to node[midway,above] {} (2,-2);
\draw [dashed,<-,>=Triangle](2,-2) to node[midway,above] {} (4,0);
\draw (2,2)--(2,3);
\draw (2,-2)--(2,-3);
    \end{tikzpicture}}
    \subfigure[]{
     \begin{tikzpicture}[scale=0.5]
       \fill[black] (0,0) circle (3pt) (4,0) circle (3pt)  (2,2) circle (3pt)  (2,3) circle (3pt) (2,-2) circle (3pt) (2,-3) circle (3pt);
\draw [dashed,->,>=Triangle](0,0) to node[midway,above] {} (2,2);
\draw [dashed,<-,>=Triangle](2,2) to node[midway,above] {} (4,0);
\draw [dashed,<-,>=Triangle](0,0) -- node[midway,below] {} (4,0);
\draw [dashed,<-,>=Triangle](0,0) to node[midway,above] {} (2,-2);
\draw (2,-2) to node[midway,above] {} (4,0);
\draw (2,2)--(2,3);
\draw (2,-2)--(2,-3);
    \end{tikzpicture}}
    \caption{Distinct $2$-leaf unrooted orientable simple strict level-2 semi-directed network topologies.}
    \label{fig:Two distinct $2$-leaf unrooted network}
\end{figure}
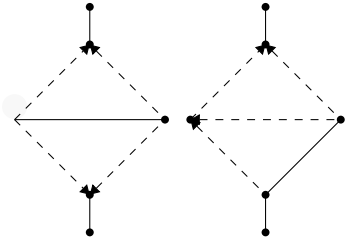
\end{example}

\begin{figure}[H]
    \centering
    \subfigure[]{
    \begin{tikzpicture}[scale=0.5]
       \fill[black] (0,0) circle (3pt) (4,0) circle (3pt) (1,2) circle (3pt)  (1,3) circle (3pt) (3,2) circle (3pt)  (3,3) circle (3pt) (2,-2) circle (3pt) (2,-3) circle (3pt) ;
\draw [dashed,->,>=Triangle](0,0) to node[midway,above] {} (1,2);
\draw [dashed,<-,>=Triangle](1,2) to node[midway,above] {} (3,2);
\draw (3,2)--(4,0);
\draw [dashed,->,>=Triangle](0,0) -- node[midway,below] {} (4,0);
\draw (0,0) to node[midway,above] {} (2,-2);
\draw [dashed,->,>=Triangle](2,-2) to node[midway,above] {} (4,0);
\draw (1,2)--(1,3);
\draw (3,2)--(3,3);
\draw (2,-2)--(2,-3);
    \end{tikzpicture}}
    \subfigure[]{
    \begin{tikzpicture}[scale=0.5]
       \fill[black] (0,0) circle (3pt) (4,0) circle (3pt) (1,2) circle (3pt)  (1,3) circle (3pt) (3,2) circle (3pt)  (3,3) circle (3pt) (2,-2) circle (3pt) (2,-3) circle (3pt) ;
\draw [dashed,->,>=Triangle](0,0) to node[midway,above] {} (1,2);
\draw [dashed,<-,>=Triangle](1,2) to node[midway,above] {} (3,2);
\draw (3,2)--(4,0);
\draw [dashed,<-,>=Triangle](0,0) -- node[midway,below] {} (4,0);
\draw [dashed,<-,>=Triangle](0,0) to node[midway,above] {} (2,-2);
\draw (2,-2) to node[midway,above] {} (4,0);
\draw (1,2)--(1,3);
\draw (3,2)--(3,3);
\draw (2,-2)--(2,-3);
    \end{tikzpicture}}
    \subfigure[]{
    \begin{tikzpicture}[scale=0.5]
       \fill[black] (0,0) circle (3pt) (4,0) circle (3pt) (1,2) circle (3pt)  (1,3) circle (3pt) (3,2) circle (3pt)  (3,3) circle (3pt) (2,-2) circle (3pt) (2,-3) circle (3pt) ;
\draw [dashed,->,>=Triangle](0,0) to node[midway,above] {} (1,2);
\draw [dashed,<-,>=Triangle](1,2) to node[midway,above] {} (3,2);
\draw (3,2)--(4,0);
\draw (0,0) -- node[midway,below] {} (4,0);
\draw [dashed,->,>=Triangle](0,0) to node[midway,above] {} (2,-2);
\draw [dashed,<-,>=Triangle](2,-2) to node[midway,above] {} (4,0);
\draw (1,2)--(1,3);
\draw (3,2)--(3,3);
\draw (2,-2)--(2,-3);
    \end{tikzpicture}}
    \subfigure[]{
    \begin{tikzpicture}[scale=0.4]
       \fill[black] (0,0) circle (3pt) (4,0) circle (3pt) (2,3) circle (3pt)  (2,4) circle (3pt) (2,-3) circle (3pt) (2,-4) circle (3pt) (2,0) circle (3pt)  (2,1) circle (3pt) ;
\draw [dashed,->,>=Triangle](0,0) to node[midway,above] {} (2,3);
\draw [dashed,<-,>=Triangle](2,3) to node[midway,above] {} (4,0);
\draw (0,0) -- node[midway,below] {} (2,0);
\draw (2,0) to node[midway,above] {} (4,0);
\draw [dashed,->,>=Triangle](0,0) to node[midway,above] {} (2,-3);
\draw [dashed,<-,>=Triangle](2,-3) to node[midway,above] {} (4,0);
\draw (2,3)--(2,4);
\draw (2,0)--(2,1);
\draw (2,-3)--(2,-4);
    \end{tikzpicture}}
    \subfigure[]{
    \begin{tikzpicture}[scale=0.4]
       \fill[black] (0,0) circle (3pt) (4,0) circle (3pt) (2,3) circle (3pt)  (2,4) circle (3pt) (2,-3) circle (3pt) (2,-4) circle (3pt) (2,0) circle (3pt)  (2,1) circle (3pt) ;
\draw [dashed,->,>=Triangle](0,0) to node[midway,above] {} (2,3);
\draw [dashed,<-,>=Triangle](2,3) to node[midway,above] {} (4,0);
\draw [dashed,<-,>=Triangle](0,0) -- node[midway,below] {} (2,0);
\draw (2,0) to node[midway,above] {} (4,0);
\draw [dashed,<-,>=Triangle](0,0) to node[midway,above] {} (2,-3);
\draw (2,-3) to node[midway,above] {} (4,0);
\draw (2,3)--(2,4);
\draw (2,0)--(2,1);
\draw (2,-3)--(2,-4);
    \end{tikzpicture}}
    \caption{Distinct $3$-leaf unrooted orientable simple strict level-2 semi-directed network topologies.}
    \label{fig:Five distinct $3$-leaf network}
\end{figure}
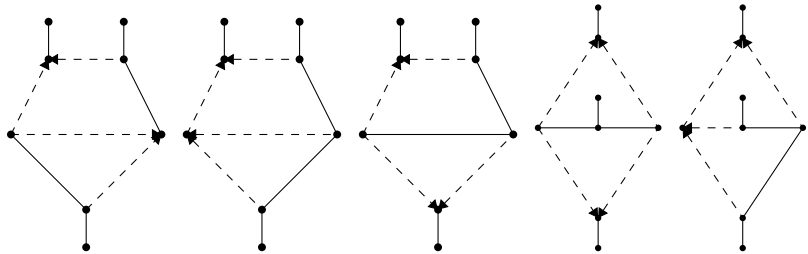

\section{Phylogenetic network model}\label{section:phylogenetic network model}

\subsection{Markov model on networks}\label{sec:networkwith2reticulationvertcies}

In \Cref{sec:networkwith2reticulationvertcies}, we will develop a mathematical model, which is called phylogenetic model, to obtain a probability distribution on $n$-tuples of DNA nucleotides from an $n$-leaf network. A phylogenetic model is a statistical model used to study the phylogeny of a collection of $n$ species based on DNA sequence at a single site. If one considers multiple sites, then it is usually assumed that every site evolves independently and identically.

We now introduce a Markov model on phylogenetic networks, which were previously studied in \cite{gross2018distinguishing,gross2020distinguishing}.  For $k\geq 1$, let $N$ be an $n$-leaf level-$k$ semi-directed network on $X$. For every edge $e$ of $N$, we associate an $l\times l$ transition matrix $M^e$, where $l$ is the cardinality of the state space $S$. In our case, $S=\{A,C,G,T\}$, which is the set of four distinct DNA nucleotides and so $l=4$.

Because $k\geq 1$, there is no unique undirected path connecting
each leaf and the possible valid root location. To build the model, we introduce reticulation edge parameters. Suppose that $N$ has $r\geq k$ reticulation vertices $v_1,\dots, v_r$. For each $v_i$, there are two
edges directed into $v_i$, let us say $e^0_i$ and $e^1_i$. We assign a parameter $\delta_i\in[0, 1]$ to $e^1_i$ and $\delta_i':=1-\delta_i$ to $e^0_i$. For $1\leq i\leq r$, we independently keep $e^1_i$ and delete $e^0_i$ with probability $\delta_i$. If not, we keep $e^0_i$ and delete $e^1_i$ with probability $\delta_i'$.  To each choice we make, we associate a binary vector $\sigma$ of length $r$ where
a zero in the $i$-th entry indicates that edge $e^0_i$ was deleted while a one in the $i$-th entry indicates that edge $e^1_i$ was deleted. Given a binary vector $\sigma$ of length $r$, after deleting $r$ edges corresponding to $\sigma$, one for each reticulation vertex, we obtain a tree $T_\sigma$. In our case,
there are $4^n$ possible observable site patterns at the leaves of $N$. If $(p_N)_\omega$ denotes the joint probability
of observing a specific site pattern $\omega=(g_1,\dots,g_n)$, where the DNA nucleotide $g_i$ is observed at the leaf $i$ of the network $N$, then this site pattern probability is given by
\begin{align*}
    (p_N)_\omega=\sum_{\sigma\in \{0,1\}^r}(\prod_{i=1}^r\delta_i^{1-\sigma_i}(1-\delta_i)^{\sigma_i})(p_{T_\sigma})_\omega,
\end{align*}
where $(p_{T_\sigma})_\omega$ denotes the joint probability of observing the site pattern $\omega$ in $T_\sigma$.

The preceding paragraph tells us that once the collection $\{M^e\}$ of transition matrices and the set of reticulation edge parameters $\{\delta_i\}$ are specified, we can easily compute the joint probability of observing a specific site pattern at the leaves. We shall call the entries of the Markov matrices and the reticulation edge parameters \textit{the numerical parameters} of the model.
Suppose that $\theta_N=S_N\times [0,1]^r$ is the set of numerical parameters associated with the network $N$ where $S_N$ and $[0,1]^r$ correspond to the entries of Markov matrices and the reticulation edge parameters, respectively. The \textit{network model associated with} $N$, denoted by $M_{N}$, is then defined as the image of the following polynomial map:
\begin{equation}\label{eq:p_n}
    \begin{split}
        \varphi_{N}:&\theta_N\rightarrow \Delta^{k^n-1}:=\{p\in \mathbb{R}^{k^n}:p\geq 0, \sum_{i=1}^{k^n}p_i=1\},\\
        &(\theta,\delta_i)\mapsto \mathbf{p}_N=((p_N)_\omega)_{\omega \in S^n}.\\
    \end{split}
\end{equation}
Here, $\Delta^m$ denotes the probability simplex in $\mathbb{R}^{m+1}$. In summary, the model $M_N$ contains all probability distributions obtained from $N$ by varying the numerical parameters in $\theta_N$.
\begin{example}\label{ex:level-1}
Let us consider the 4-leaf level-1 network $N$ presented in \Cref{fig:4leaf-2reticulation1} with two reticulation vertices, $w_1$ and $w_2$. There are  four possible binary vectors of length two: $\alpha=(0,0),\beta=(0,1),\gamma=(1,0),\mbox{ and }\delta=(1,1)$. Thus, we  obtain the following four unrooted trees in \Cref{fig:4embeddedtressin-N1}.
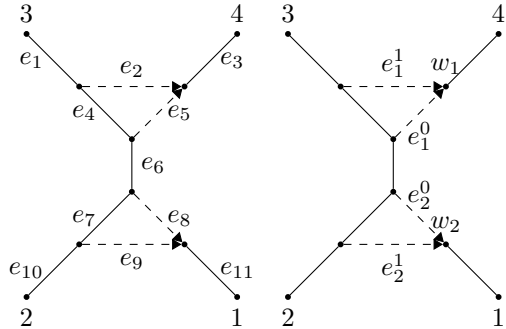
\begin{figure}[H]
    \centering
    \subfigure[]{\begin{tikzpicture}[scale=0.7]
  \node (a1) at (0,0)[circle , draw, fill=black!50, inner sep=0pt, very thick ]{} ;  
  \node (a2) at (-1,-1)[circle , draw, fill=black!50, inner sep=0pt, very thick ]{};
  \node [](a3) at (1,-1)[circle , draw, fill=black!50, inner sep=0pt, very thick ]{};  
  \node [label=below:2](a4) at (-2,-2)[circle , draw, fill=black!50, inner sep=0pt, very thick ]{};  
  \node [label=below:1](a5) at (2,-2) [circle , draw, fill=black!50, inner sep=0pt, very thick ]{}; 
  \node (a6) at (0,1)[circle , draw, fill=black!50, inner sep=0pt, very thick ]{};
  \node (a7) at (-1,2)[circle , draw, fill=black!50, inner sep=0pt, very thick ]{};
  \node [](a8) at (1,2) [circle , draw, fill=black!50, inner sep=0pt, very thick ]{};
  \node [label=above:3](a9) at (-2,3)[circle , draw, fill=black!50, inner sep=0pt, very thick ]{};
  \node [label=above:4](a10) at (2,3)[circle , draw, fill=black!50, inner sep=0pt, very thick ]{};

  \draw [dashed,->,>=Triangle](a1) -- (a3) node[midway, right]{$e_8$}; 
  \draw [dashed,->,>=Triangle](a2) -- (a3) node[midway, below]{$e_9$} ;
  \draw [dashed,->,>=Triangle](a6) -- (a8)node[midway, right]{$e_5$} ;  
  \draw [dashed,->,>=Triangle](a7) -- (a8) node[midway, above]{$e_2$};  
  \draw (a1)--(a2)node[midway,left]{$e_7$};
  \draw (a2)--(a4)node[midway,left]{$e_{10}$};
  \draw (a3)--(a5)node[midway,right]{$e_{11}$};
  \draw (a1)--(a6)node[midway,right]{$e_6$};
  \draw (a6)--(a7)node[midway,left]{$e_4$};
  \draw (a7)--(a9)node[midway,left]{$e_1$};
  \draw (a8)--(a10)node[midway,right]{$e_3$};
\end{tikzpicture}}
\subfigure[]{\begin{tikzpicture}[scale=0.7]
  \node (a1) at (0,0)[circle , draw, fill=black!50, inner sep=0pt, very thick ]{} ;  
  \node (a2) at (-1,-1)[circle , draw, fill=black!50, inner sep=0pt, very thick ]{};  
  \node [label=above:$w_2$](a3) at (1,-1)[circle , draw, fill=black!50, inner sep=0pt, very thick ]{};  
  \node [label=below:2](a4) at (-2,-2)[circle , draw, fill=black!50, inner sep=0pt, very thick ]{};  
  \node [label=below:1](a5) at (2,-2) [circle , draw, fill=black!50, inner sep=0pt, very thick ]{}; 
  \node (a6) at (0,1)[circle , draw, fill=black!50, inner sep=0pt, very thick ]{};
  \node (a7) at (-1,2)[circle , draw, fill=black!50, inner sep=0pt, very thick ]{};
  \node [label=above:$w_1$](a8) at (1,2) [circle , draw, fill=black!50, inner sep=0pt, very thick ]{};
  \node [label=above:3](a9) at (-2,3)[circle , draw, fill=black!50, inner sep=0pt, very thick ]{};
  \node [label=above:4](a10) at (2,3)[circle , draw, fill=black!50, inner sep=0pt, very thick ]{};

  \draw [dashed,->,>=Triangle](a1) -- (a3) node[midway, above]{$e^0_2$}; 
  \draw [dashed,->,>=Triangle](a2) -- (a3) node[midway, below]{$e^1_2$} ;
  \draw [dashed,->,>=Triangle](a6) -- (a8)node[midway, below]{$e^0_1$} ;  
  \draw [dashed,->,>=Triangle](a7) -- (a8) node[midway, above]{$e^1_1$};  
  \draw (a1)--(a2);
  \draw (a2)--(a4);
  \draw (a3)--(a5);
  \draw (a1)--(a6);
  \draw (a6)--(a7);
  \draw (a7)--(a9);
  \draw (a8)--(a10);
\end{tikzpicture}}
\caption{ The 4-leaf semi-directed phylogenetic networks $N$ with two reticulation vertices. The left figure displays the edge labeling, which will be used in deriving a parameterization of $M_N$.}
\label{fig:4leaf-2reticulation1}
\end{figure}

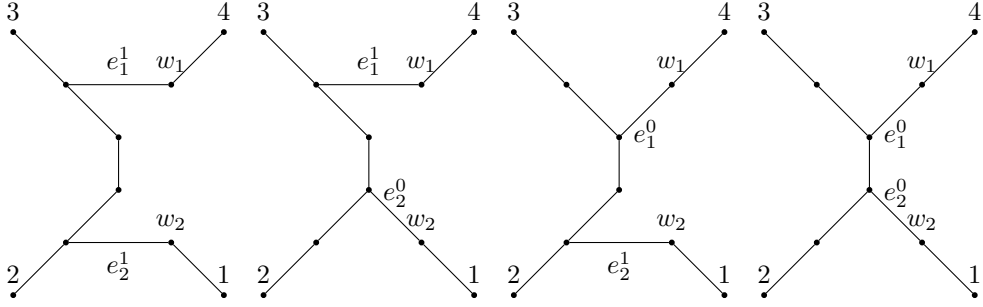
\begin{figure}[H]
    \centering
 \begin{tikzpicture}[scale=0.7]
  \node (a1) at (0,0)[circle , draw, fill=black!50, inner sep=0pt, very thick ]{} ;  
  \node (a2) at (-1,-1)[circle , draw, fill=black!50, inner sep=0pt, very thick ]{};  
  \node [label=above:$w_2$](a3) at (1,-1)[circle , draw, fill=black!50, inner sep=0pt, very thick ]{};  
  \node [label=above:2](a4) at (-2,-2)[circle , draw, fill=black!50, inner sep=0pt, very thick ]{};  
  \node [label=above:1](a5) at (2,-2) [circle , draw, fill=black!50, inner sep=0pt, very thick ]{}; 
  \node (a6) at (0,1)[circle , draw, fill=black!50, inner sep=0pt, very thick ]{};
  \node (a7) at (-1,2)[circle , draw, fill=black!50, inner sep=0pt, very thick ]{};
  \node [label=above:$w_1$](a8) at (1,2) [circle , draw, fill=black!50, inner sep=0pt, very thick ]{};
  \node [label=above:3](a9) at (-2,3)[circle , draw, fill=black!50, inner sep=0pt, very thick ]{};
  \node [label=above:4](a10) at (2,3)[circle , draw, fill=black!50, inner sep=0pt, very thick ]{};
 
  \draw (a2) -- (a3) node[midway, below]{$e^1_2$} ;
  \draw (a7) -- (a8) node[midway, above]{$e^1_1$};  
  \draw (a1)--(a2);
  \draw (a2)--(a4);
  \draw (a3)--(a5);
  \draw (a1)--(a6);
  \draw (a6)--(a7);
  \draw (a7)--(a9);
  \draw (a8)--(a10);
\end{tikzpicture}
\begin{tikzpicture}[scale=0.7]
  \node (a1) at (0,0)[circle , draw, fill=black!50, inner sep=0pt, very thick ]{} ;  
  \node (a2) at (-1,-1)[circle , draw, fill=black!50, inner sep=0pt, very thick ]{};  
  \node [label=above:$w_2$](a3) at (1,-1)[circle , draw, fill=black!50, inner sep=0pt, very thick ]{};  
  \node [label=above:2](a4) at (-2,-2)[circle , draw, fill=black!50, inner sep=0pt, very thick ]{};  
  \node [label=above:1](a5) at (2,-2) [circle , draw, fill=black!50, inner sep=0pt, very thick ]{}; 
  \node (a6) at (0,1)[circle , draw, fill=black!50, inner sep=0pt, very thick ]{};
  \node (a7) at (-1,2)[circle , draw, fill=black!50, inner sep=0pt, very thick ]{};
  \node [label=above:$w_1$](a8) at (1,2) [circle , draw, fill=black!50, inner sep=0pt, very thick ]{};
  \node [label=above:3](a9) at (-2,3)[circle , draw, fill=black!50, inner sep=0pt, very thick ]{};
  \node [label=above:4](a10) at (2,3)[circle , draw, fill=black!50, inner sep=0pt, very thick ]{};

  \draw (a1) -- (a3) node[midway, above]{$e^0_2$}; 
  \draw (a7) -- (a8) node[midway, above]{$e^1_1$};  
  \draw (a1)--(a2);
  \draw (a2)--(a4);
  \draw (a3)--(a5);
  \draw (a1)--(a6);
  \draw (a6)--(a7);
  \draw (a7)--(a9);
  \draw (a8)--(a10);
\end{tikzpicture}
\begin{tikzpicture}[scale=0.7]
  \node (a1) at (0,0)[circle , draw, fill=black!50, inner sep=0pt, very thick ]{} ;  
  \node (a2) at (-1,-1)[circle , draw, fill=black!50, inner sep=0pt, very thick ]{};  
  \node [label=above:$w_2$](a3) at (1,-1)[circle , draw, fill=black!50, inner sep=0pt, very thick ]{};  
  \node [label=above:2](a4) at (-2,-2)[circle , draw, fill=black!50, inner sep=0pt, very thick ]{};  
  \node [label=above:1](a5) at (2,-2) [circle , draw, fill=black!50, inner sep=0pt, very thick ]{}; 
  \node (a6) at (0,1)[circle , draw, fill=black!50, inner sep=0pt, very thick ]{};
  \node (a7) at (-1,2)[circle , draw, fill=black!50, inner sep=0pt, very thick ]{};
  \node [label=above:$w_1$](a8) at (1,2) [circle , draw, fill=black!50, inner sep=0pt, very thick ]{};
  \node [label=above:3](a9) at (-2,3)[circle , draw, fill=black!50, inner sep=0pt, very thick ]{};
  \node [label=above:4](a10) at (2,3)[circle , draw, fill=black!50, inner sep=0pt, very thick ]{};
 
  \draw (a2) -- (a3) node[midway, below]{$e^1_2$} ;
  \draw (a6) -- (a8)node[midway, below]{$e^0_1$} ;  
  \draw (a1)--(a2);
  \draw (a2)--(a4);
  \draw (a3)--(a5);
  \draw (a1)--(a6);
  \draw (a6)--(a7);
  \draw (a7)--(a9);
  \draw (a8)--(a10);
\end{tikzpicture}
\begin{tikzpicture}[scale=0.7]
  \node (a1) at (0,0)[circle , draw, fill=black!50, inner sep=0pt, very thick ]{} ;  
  \node (a2) at (-1,-1)[circle , draw, fill=black!50, inner sep=0pt, very thick ]{};  
  \node [label=above:$w_2$](a3) at (1,-1)[circle , draw, fill=black!50, inner sep=0pt, very thick ]{};  
  \node [label=above:2](a4) at (-2,-2)[circle , draw, fill=black!50, inner sep=0pt, very thick ]{};  
  \node [label=above:1](a5) at (2,-2) [circle , draw, fill=black!50, inner sep=0pt, very thick ]{}; 
  \node (a6) at (0,1)[circle , draw, fill=black!50, inner sep=0pt, very thick ]{};
  \node (a7) at (-1,2)[circle , draw, fill=black!50, inner sep=0pt, very thick ]{};
  \node [label=above:$w_1$](a8) at (1,2) [circle , draw, fill=black!50, inner sep=0pt, very thick ]{};
  \node [label=above:3](a9) at (-2,3)[circle , draw, fill=black!50, inner sep=0pt, very thick ]{};
  \node [label=above:4](a10) at (2,3)[circle , draw, fill=black!50, inner sep=0pt, very thick ]{};
 
  \draw (a1) -- (a3) node[midway, above]{$e^0_2$}; 
  \draw (a6) -- (a8)node[midway, below]{$e^0_1$} ;   
  \draw (a1)--(a2);
  \draw (a2)--(a4);
  \draw (a3)--(a5);
  \draw (a1)--(a6);
  \draw (a6)--(a7);
  \draw (a7)--(a9);
  \draw (a8)--(a10);
\end{tikzpicture}
\caption{The trees $T_\alpha,T_\beta,T_\delta\mbox{ and }T_\gamma$ from left to right.}
\label{fig:4embeddedtressin-N1}
\end{figure}
\end{example}

\subsection{Fourier transformation in phylogenetics}

 The well-studied nucleotide substitution models include the well-known Jukes-Cantor (JC), Kimura 2-parameter (K2P), and 3-parameter (K3P) models. These three models are some examples of a large class of phylogenetic models, called \textit{group-based models}. Given a finite additive abelian group $G$, in a group-based model with underlying group $G$, the entries of the transition matrix $M^{e}$ associated to the edge $e$ has to satisfy
$M^e_{g,h}=f^e(h-g)$
for some function $f^e:G\rightarrow \mathbb{R}.$ It is important to note that the group-based models considered in the literature are time reversible. For the JC, K2P, or K3P model, the underlying group is $\mathbb{Z}_2\times\mathbb{Z}_2$ and we use the following identification of the DNA nucleotides with the group elements: $A=(0,0),C=(0,1),G=(1,0),\mbox{ and } T=(1,1).$ Moreover, the transition matrices $M^e$ have the following forms, respectively: 
$$\begin{pmatrix}
a^e&b^e&b^e&b^e\\
b^e&a^e&b^e&b^e\\
b^e&b^e&a^e&b^e\\
b^e&b^e&b^e&a^e\\
\end{pmatrix},\begin{pmatrix}
a^e&b^e&c^e&c^e\\
b^e&a^e&c^e&c^e\\
c^e&c^e&a^e&b^e\\
c^e&c^e&b^e&a^e\\
\end{pmatrix},
\begin{pmatrix}
a^e&b^e&c^e&d^e\\
b^e&a^e&d^e&c^e\\
c^e&d^e&a^e&b^e\\
d^e&c^e&b^e&a^e\\
\end{pmatrix}.$$

An advantage of using group-based models is that group-based models allow us to apply \textit{the discrete Fourier transform}, which is a linear change of coordinates, to obtain a parameterization of the phylogenetic tree model \cite{EvansSpeed,szekely1993fourier,szekely1992fourier}. By applying this transformation on a phylogenetic tree, we will have a monomial parameterization of the group-based model, instead of a polynomial parameterization. Therefore, the corresponding algebraic varieties are toric \cite{sturmfels2005toric}.

The Fourier coordinates of a group-based model on $n$-leaf network will be denoted by $q_{g_1\dots g_n}$ where $g_i\in G$. Suppose that $G=\mathbb{Z}_2\times\mathbb{Z}_2$ and $T$ is an $n$-leaf tree.  Let $\Sigma(T)$ denote the set of splits induced by the edges of $T$. We associate a set of parameters $\{a^{A|B}_g\}_{g\in G}$ to every split $A|B\in \Sigma(T)$,. Then the Fourier coordinate $q_{g_1\dots g_n}$ is given by the following monomial parameterization:
\[ 
q_{g_1\dots g_n}= \left\{
\begin{array}{ll}
      \prod_{A|B\in \Sigma(T)}a^{A|B}_{\sum_{i\in A}g_i,}&\mbox{ if }\sum_{i=1}^ng_i=0 \\
      0,&\mbox{ otherwise. }
\end{array} 
\right. 
\]
Following \cite{EvansSpeed} and \cite{hendy1993spectral}, the parameters in the representation of the Fourier coordinates in the K3P model on a tree $T$ satisfy $a^e_A=1$ and $a^e_C,a^e_G,a^e_T\in(0,1]$ for each split $e\in \Sigma(T).$ Moreover, in the K2P model and the JC model, they should additionally satisfy $a^e_G=a^e_T$ and $a^e_C=a^e_G=a^e_T$, respectively,  for each split $e.$ We will ignore the condition $a^e_A=1$ as we want to homogenize the polynomial parameterization. The following example illustrates how to obtain the parameterization for the level-1 network in \Cref{ex:level-1}. This parameterization can be extended to level-$k$ network in similar manner. If the network has $r$ reticulation vertices, then each Fourier coordinate for the model will consist of $2^r$ distinct terms.

\begin{example}
Suppose that we label the edge of $N$ as displayed on the left in \Cref{fig:4leaf-2reticulation1}. We want to compute the Fourier coordinate of observing nucleotides T, G, C, and A at the leaves 1, 2, 3 and 4, respectively. We will denote by $a^i_g$, the parameter $a^{e_i}_g$ to simplify the notation. Then 
\begin{align*}
    q_{TGCA}=&\delta_1\delta_2(a^1_Ca^2_Aa^3_Aa^4_{A+C}a^6_{A+C}a^7_{A+C}a^9_{T}a^{10}_Ga^{11}_T)+\delta_1\delta_2'(a^1_Ca^2_Aa^3_Aa^4_{A+C}a^6_{A+C}a^7_{G}a^8_{T}a^{10}_Ga^{11}_T)\\
    &+\delta_1'\delta_2(a^1_Ca^3_Aa^4_{C}a^5_Aa^6_{A+C}a^7_{G+T}a^9_{T}a^{10}_Ga^{11}_T)+\delta_1'\delta_2'(a^1_Ca^3_Aa^4_{C}a^5_Aa^6_{A+C}a^7_{G}a^8_{T}a^{10}_Ga^{11}_T)\\
    =&\delta_1\delta_2(a^1_Ca^2_Aa^3_Aa^4_{C}a^6_{C}a^7_{C}a^9_{T}a^{10}_Ga^{11}_T)+\delta_1\delta_2'(a^1_Ca^2_Aa^3_Aa^4_{C}a^6_{C}a^7_{G}a^8_{T}a^{10}_Ga^{11}_T)\\
    &+\delta_1'\delta_2(a^1_Ca^3_Aa^4_{C}a^5_Aa^6_{C}a^7_{C}a^9_{T}a^{10}_Ga^{11}_T)+\delta_1'\delta_2'(a^1_Ca^3_Aa^4_{C}a^5_Aa^6_{C}a^7_{G}a^8_{T}a^{10}_Ga^{11}_T).\\
\end{align*}
In the above parameterization, the first, the second, the third and the fourth term correspond to the trees $T_\alpha,T_\beta,T_\delta\mbox{ and }T_\gamma$ in \Cref{fig:4embeddedtressin-N1}, respectively. We can further reparameterize this Fourier coordinate by replacing $\delta_1a^2_g$ with $a^2_g$, $\delta_1'a^5_g$ with $a^5_g$, $\delta_2a^9_g$ with $a^9_g$, and $\delta_2'a^8_g$ with $a^8_g$. Therefore, the parameterization of observing nucleotides T, G, C, and A at the leaves 1, 2, 3 and 4, respectively, can be written as:
\begin{align*}
    q_{TGCA}=&a^1_Ca^2_Aa^3_Aa^4_{C}a^6_{C}a^7_{C}a^9_{T}a^{10}_Ga^{11}_T+a^1_Ca^2_Aa^3_Aa^4_{C}a^6_{C}a^7_{G}a^8_{T}a^{10}_Ga^{11}_T\\
    &+a^1_Ca^3_Aa^4_{C}a^5_Aa^6_{C}a^7_{C}a^9_{T}a^{10}_Ga^{11}_T+a^1_Ca^3_Aa^4_{C}a^5_Aa^6_{C}a^7_{G}a^8_{T}a^{10}_Ga^{11}_T.\\
\end{align*}
\end{example}

\subsection{Generic identifiability and algebraic varieties associated with phylogenetic Markov models}
In this subsection, we discuss how to associate an algebraic variety to each network model described earlier. Moreover, the definition of identifiability of a model will be presented. Hence, we hope to employ an algebraic approach to solve the model identifiability problems. 

Given a network $N$, an element of $Im(\varphi_N)$ belongs to a probability simplex. Thus, all the entries of this element belong to the interval $[0,1]$ and they must sum to one.  Additionally, each coordinate function of $\varphi_N$ is a homogeneous polynomial of degree equal to the number of edges of the tree obtained after deleting exactly one of the two reticulation edges directed to each reticulation vertex. If we regard $\varphi_N$ as a complex polynomial map, then the \textit{algebraic variety $V_N$ associated with a network model} $M_N$ is defined as the Zariski closure of the image, i.e. $\overline{Im(\varphi_N)}$.

Associated with an algebraic variety $V$, we have the \textit{vanishing ideal}, denoted by $I_V$, which is the set $\{f\in \mathbb{C}[p]: f(v)=0,\mbox{ }\forall v\in V\}.$ In our case, this ideal lives in the polynomial ring $\mathbb{C}[p_{i_1i_2\dots i_n}:(i_1,\dots,i_n)\in \{A,C,G,T\}^n]$. If a joint probability distribution arises from the network model $M_N$, then its entries are constrained by polynomials belonging to the vanishing ideal of the variety $V_N$. These polynomials are known as \textit{phylogenetic invariants} associated with the phylogenetic network model $M_N$. The use of phylogenetic invariants is often referred as an algebraic approach to study the identifiability of phylogenetic models. This algebraic approach using invariants was first introduced in phylogenetics by Cavender and Felsenstein \cite{CavenderFelsenstein} and Lake \cite{lake1987rate} in the setting of phylogenetic tree. Interested readers can refer to
\cite{allman2003phylogenetic,allman2008phylogenetic,EvansSpeed,steel1995classifying,sturmfels2005toric,szekely1993fourier,szekely1992fourier} for a list of publications on phylogenetic invariants.

Identifiability of the model parameters is an important question to solve as it gives information on the sufficiency of our data to learn about the model parameters. In our case, the identifiability of phylogenetic network model ensures that our genetic data is sufficient to be used to recover the evolutionary history of species. Alternatively, given some genetic data, we can find a unique network topology and a set of numerical parameters that give rise to a probability distribution that agrees with our data. In what follows, we would like to study the identifiability of the network parameter of a network model. As we mentioned earlier in the introduction, for an algebraic model, it is favorable to slightly weaken the notion of (global) identifiability to generic identifiability as presented in the following definitions. Roughly speaking, given a specific network, the network parameter is generically identifiable if the probability distribution coming from a generic choice of numerical parameters could have only arisen from this network topology.

\begin{definition}[\cite{gross2018distinguishing}, Section 3]
Let $\{M_{N}\}_{N\in \mathcal{N}}$ be a class of network models.
\begin{enumerate}
    \item The network parameter is \textit{(globally) identifiable} with respect to the collection   $\{M_{N}\}_{N\in \mathcal{N}}$ if there is no probability distribution that belongs to both $M_{N_1}$ and  $M_{N_2}$ for two distinct networks $N_1,N_2\in \mathcal{N}$.
    
    \item The network parameter is \textit{generically identifiable} with respect to the collection $\{M_{N}\}_{N\in \mathcal{N}}$ if the equality $\varphi_N^{-1}(\varphi_N(\theta))=\{\theta\}$ holds for almost all $\theta\in \theta_N$ and $N\in \mathcal{N}$.
\end{enumerate}

\end{definition}
 
Given a class of networks $\mathcal{N}$, we would like to study the set of all associated varieties $\{V_N\}_{N\in \mathcal{N}}$. To do this, let us recall the notion of distinguishable networks that was introduced in \cite{gross2018distinguishing}.
 
\begin{definition}[\cite{gross2018distinguishing}, Section 3]
Two networks $N_1$ and $N_2$ are said to be \textit{distinguishable} if $V_{N_1}\cap V_{N_2}$ is a proper subvariety of $V_{N_1}$ and $V_{N_2}$. If not, then we call them \textit{indistinguishable}.
\end{definition}
\begin{proposition}[\cite{gross2018distinguishing}, Proposition 3.3] Let $\{M_{N}\}_{N\in \mathcal{N}}$ be a class of network models and $N_1,N_2\in \mathcal{N}$ be any two distinct $n$-leaf networks. If $N_1$ and $N_2$ are distinguishable, then the network parameter is generically identifiable with respect to $\{M_{N}\}_{N\in \mathcal{N}}$.
\end{proposition}
 The above proposition tells us that one should attempt to first check whether two networks $N_1$ and $N_2$  are distinguishable in order to be able to obtain the generic identifiability result. If $V_{N_1}\subseteq V_{N_2}$ or $V_{N_2}\subseteq V_{N_1}$, then $V_{N_1}\cap V_{N_2}$ would not be a proper subvariety of $V_{N_1}$ and of $V_{N_2}$. Thus, in this case, $N_1$ and $N_2$ are indistinguishable. In order to distinguish two networks, one needs to show that $V_{N_1}\not\subseteq V_{N_2}$ and $V_{N_2}\not\subseteq V_{N_1}$. Equivalently, one needs to find two phylogenetic invariants $f_1$ and $f_2$ satisfying $f_1\in I_{N_1}\setminus I_{N_2}$ and $f_2\in I_{N_2}\setminus I_{N_1}$.

\subsection{Nice phylogenetic networks} \label{subsection:nice phylogenetic networks}
As mentioned in \Cref{sec:networkwith2reticulationvertcies}, given an $n$-leaf strict level-$k$ network $N$ on $X$ with $r$ reticulation vertices, we can obtain a tree from $N$ by removing exactly $l$  out of $2l$ reticulation edges in each level-$l$ biconnected component ($l\leq k$), one for each pair of reticulation edges incident  to  each  reticulation vertex.  Let $B_r$ denote the set of all binary vectors of length $r$. If $\sigma\in B_r$, then $T_{\sigma}$ denotes the tree obtained by removing the reticulation edges specified by the entries of $\sigma$. It is understood that the leaf set of $T_{\sigma}$ contains $X$. However, this tree could contain degree two vertices and the leaf set of this tree could be strictly bigger than $X$. For instance, it can be seen from the networks $N_1$ and $N_2$ in \Cref{ex:level-k example}. In $N_1$, if we delete $e_1$ or $e_2$, then the vertex $a$ will have degree two. In $N_2$, if we delete edges $e_1$ and $e_4$, then in the resulting tree, the vertex $a$ will be a new leaf that is not contained in the set of leaves of $N_2$. 

We are interested in a network $N$ such that for any $\sigma\in B_r$, the leaf set of $T_{\sigma}$ is exactly the set $X$. A network with such property will be called \textit{nice}. Therefore, this class of networks will allow us to apply the discrete Fourier transformation to compute the phylogenetic invariants associated with the model. In \Cref{ex:level-k example}, we observe that level-2 networks might contain the following type of vertices which are problematic for the discrete Fourier transformation.

\begin{definition}
A \textit{funnel} is a non-leaf vertex $v$ of the following type:
\begin{enumerate}
    \item type $A$ if $v$ is a child of two reticulation vertices and also a parent of a reticulation vertex, or
    \item type $B$ if $v$ is a parent of two reticulation vertices.
\end{enumerate}
Two types of funnel vertices are illustrated in \Cref{fig:funnel}. Furthermore, A network is said to be \textit{funnel-free} if it does contain any funnel vertices.
\end{definition}
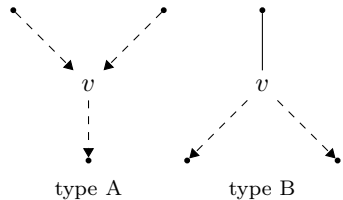
\begin{figure}[H]\label{fig:funnel vertices}
    \centering
    \subfigure[type A]{
    \begin{tikzpicture}[scale=0.5]
  \node (a1) at (0,0){$v$} ;  
  \node (a2) at (2,2)[circle , draw, fill=black!50, inner sep=0pt, very thick ]{} ; 
  \node (a3) at (-2,2)[circle , draw, fill=black!50, inner sep=0pt, very thick ]{} ; 
  \node (a4) at (0,-2)[circle , draw, fill=black!50, inner sep=0pt, very thick ]{} ; 
  \draw [dashed,->,>=Triangle](a2)--(a1);
  \draw [dashed,->,>=Triangle](a3)--(a1);
  \draw [dashed,->,>=Triangle](a1)--(a4);
\end{tikzpicture}}
  \subfigure[type B]{
  \begin{tikzpicture}[scale=0.5]
  \node (a1) at (0,0){$v$} ;  
  \node (a2) at (2,-2)[circle , draw, fill=black!50, inner sep=0pt, very thick ]{} ; 
  \node (a3) at (-2,-2)[circle , draw, fill=black!50, inner sep=0pt, very thick ]{} ; 
  \node (a4) at (0,2)[circle , draw, fill=black!50, inner sep=0pt, very thick ]{} ; 
  \draw [dashed,->,>=Triangle](a1)--(a2);
  \draw [dashed,->,>=Triangle](a1)--(a3);
  \draw (a1)--(a4);
\end{tikzpicture}}
\caption{Funnel vertices.}
    \label{fig:funnel}
\end{figure}

\begin{example}
 We will consider two networks $N_1$ and $N_2$ in \Cref{ex:level-k example}. The network $N_1$ is funnel-free while $N_2$ is not funnel-free because the vertex $a$ is a funnel of type A and the vertex $e$ is a funnel of type B. 
\end{example}\

Although here we introduce funnel-free and nice networks for semi-directed networks, in the literature, for the rooted (directed) phylogenetic network, there is a class of networks called \textit{tree-child} networks \cite[Definition 1]{cardona2008comparison}. In this network, every non-leaf vertex has a child that is either a tree vertex or a leaf. Tree-child networks were originally introduced to adapt the complex biological phenomena in a computationally controllable way \cite{cardona2008comparison}. Biologically, tree-child network can not accommodate every biological phenomena as the mathematical definition of a tree-child network requires that every non-extant species has at least one child that evolved only through mutation. Despite of this limitation, tree-child networks allow us to model many relevant biological scenarios \cite{cardona2019generation}. Additionally, those evolutionary histories that can not be modeled using tree-child networks can be approximated well with nice networks.

The following lemma suggests that being funnel-free semi-directed network is equivalent to being nice.
\begin{lemma}\label{lemma:funnel-free and nice}
\begin{enumerate}
    \item Let $N$ be any semi-directed network on $X$. Then $N$ is funnel-free if and only if $N$ is nice.
    \item Every level-1 semi-directed network is nice.
\end{enumerate}
 
\end{lemma}
\begin{proof}
\begin{enumerate}
    \item Let $N$ be a funnel-free semi-directed network with $r$ reticulation vertices and $\sigma$ be any element of $B_r$. Suppose that the leaf set of the tree $T_\sigma$ is not the set $X$. Then there exists $y\notin X$ such that $y$ is a leaf of $T_\sigma$. It implies that $y$ is obtained from a type A or B funnel vertex, a contradiction. Conversely, let $N$ be a nice semi-directed network. If $N$ contains a type A funnel vertex $v\notin X$ such that there are there reticulation edges $e_1,e_2,$ and $e_3$ such that $e_1$ and $e_2$ are directed to $v$ and $e_3$ is directed away from $v$. If we remove $e_3$ and either $e_1$ or $e_2$, then in the resulting tree, the vertex $v$ is a leaf, a contradiction. Similarly, one can show that $N$ can not contain funnel vertex of type $B$.
    \item Let $N$ be a level-1 semi-directed network. If $N$ is not nice, then by the first part of the lemma, there exists a non-leaf vertex $u$ such that all its children are reticulation vertices.
    
    Now suppose that $u$ is a reticulation vertex.  Then $u$ has only one child $v$ and  $(u,v)$ is a reticulation edge. Let us assume that the parents of $u$ are $u_1$ and $u_2$ and let $v'$ be the other parent of $v.$ Let the vertex $x$ be the least common ancestor of $u_1$ and $u_2$ and the vertex $y$ be the least common ancestor of $u_2$ and $v'$. The paths $(x,\dots, u_1,u)$ and $(x,\dots,u_2,u)$ define a cycle. Similarly, the paths $(y,\dots, u_2,u,v)$ and $(y,\dots,v',v)$ also define a cycle. These two cycles share an edge $(u_2,u)$, contradicting the hypothesis that $N$ is a level-1 network. Thus, $u$ is a tree vertex. See \Cref{fig:level-1 is nice} (a).
    
    Since $u$ is a tree vertex, $u$ has two reticulation children, $v_1$ and $v_2$. Let $v_1'$ and $v_2'$ be the other parent of $v_1$ and $v_2$, respectively. Moreover, let $u'$ be the unique parent of $u$. If the vertex $x$ is least common ancestor of $u'$ and $v_1'$ and the vertex $y$ is least common ancestor of $u'$ and $v_2'$. The paths $(x,\dots, v_1',v_1)$ and $(x,\dots, u',u,v_1)$ define a cycle. Similarly, the paths $(y,\dots, v_2',v_2)$ and $(y,\dots, u',u,v_2)$ also define a cycle. These two cycles share an edge $(u',u)$, contradicting the hypothesis that $N$ is a level-1 network. See \Cref{fig:level-1 is nice} (b). \qedhere
    \end{enumerate}
\end{proof}

    \begin{figure}[H]
    \centering
    \subfigure[(a)]{
    \begin{tikzpicture}[scale=0.75]
       \fill[black] (0,0) circle (3pt) (2,0) circle (3pt) (4,0) circle (3pt) (1,-1) circle (3pt) (2,-2) circle (3pt) (1,1) circle (3pt) (3,1) circle (3pt);
       \node (a1) at (0,0)[left]{$u_1$} ; 
       \node (a1) at (2,0)[right]{$u_2$} ; 
       \node (a1) at (1,-1)[left]{$u$} ;
       \node (a1) at (2,-2)[below]{$v$} ; 
       \node (a1) at (4,0)[right]{$v'$} ; 
       \node (a1) at (1,1)[above]{$x$} ;
       \node (a1) at (3,1)[above]{$y$} ;
       \draw [dashed,->,>=Triangle](0,0)--(1,-1);
       \draw [dashed,->,>=Triangle](2,0)--(1,-1);
       \draw (1,1)--(0,0);
       \draw (1,1)--(2,0);
       \draw [dashed,->,>=Triangle](1,-1)--(2,-2);
       \draw [dashed,->,>=Triangle](4,0)--(2,-2);
       \draw (3,1)--(4,0);
       \draw (3,1)--(2,0);
       \draw (1,1)--(2,0);
    \end{tikzpicture}}
    \subfigure[(b)]{
    \begin{tikzpicture}[scale=0.75]
       \fill[black] (0,0) circle (3pt) (2,0) circle (3pt) (4,0) circle (3pt)  (2,1) circle (3pt) (1,2) circle (3pt) (3,2) circle (3pt) (1,-1) circle (3pt) (3,-1) circle (3pt);
       \node (a1) at (0,0)[left]{$v_1'$} ; 
       \node (a1) at (2,0)[right]{$u$} ; 
       \node (a1) at (4,0)[right]{$v_2'$} ;
       \node (a1) at (1,-1)[below]{$v_1$} ; 
       \node (a1) at (3,-1)[below]{$v_2$} ; 
       \node (a1) at (2,1)[right]{$u'$} ;
       \node (a1) at (1,2)[above]{$x$} ;
       \node (a1) at (3,2)[above]{$y$} ;
\draw [dashed,->,>=Triangle](0,0)--(1,-1);
\draw [dashed,->,>=Triangle](2,0) --(1,-1);
\draw [dashed,->,>=Triangle](2,0) --(3,-1);
\draw [dashed,->,>=Triangle](4,0)--(3,-1);
\draw (2,0)--(2,1);
\draw (0,0)--(1,2);
\draw (1,2)--(2,1);
\draw (2,1)--(3,2);
\draw (3,2)--(4,0);
    \end{tikzpicture}}
    \caption{In a level-1 semi-directed network,  a reticulation vertex can not have a reticulation child and a tree vertex can not have two reticulation children.}
    \label{fig:level-1 is nice}
    \end{figure}
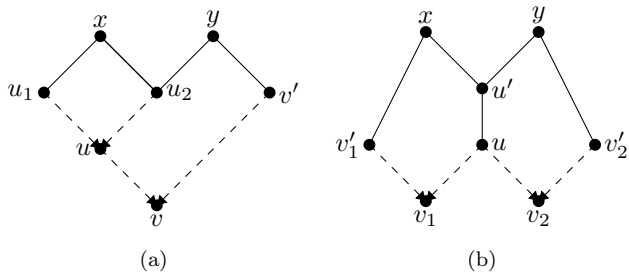

The second statement of \Cref{lemma:funnel-free and nice} allows us to apply the Fourier transformation method to simplify the joint probability distributions for level-1 semi-directed network as presented in \cite{gross2018distinguishing} and \cite{gross2020distinguishing}. This would not be the case for level-2 networks. Thus, we will limit our study to nice level-2 networks as we are interested to obtain phylogenetic invariants associated with a phylogenetic model using discrete Fourier transformation.

\section{Semisimple level-2 networks}\label{section:semisimple}

In this section, we introduce a class of level-2 networks that is more general compared to the class of simple level-2 networks. A network in this class is allowed to contain a tree cherry. A network in this class will be called semisimple. Let us first discuss what happens in level-1 networks. For level-1 networks, the definition of cycle networks was introduced in \cite{gross2018distinguishing}. This class of semisimple networks can be thought as a generalization of cycle networks for the case of level-2 networks. 

\begin{definition}[\cite{gross2018distinguishing}, Definition 2.2]\label{def:cycle network}
A \textit{cycle network} is a level-1 network with exactly one reticulation vertex. A $k$-\textit{cycle network} is a cycle network that contains an undirected cycle of length $k$. Finally, a \textit{cycle vertex} is an internal vertex that is contained in the unique cycle of a cycle network.
\end{definition}

By \Cref{def:cycle network}, the set of cycle vertices of an $n$-leaf $k$-cycle network $N$ induces a partition on the leaf set $[n]$ as follows. Suppose that the set of cycle vertices of $N$ is given by $\{v_1,\dots, v_k\}$.
 For each $v_i$, there is a tree $T_{v_i}$ attached to $v_i.$ Let $A_{v_i}$, or $A_i$ for short, be the set of leaf labels of $T_{v_i}.$ As a consequence, the cycle vertices induce a partition $A_1|A_2|\dots|A_k$ of $[n]$.

\begin{definition}
Let $A_1|A_2|A_3$ be a partition of $[n]$ into two or three parts such that $|A_1|\geq |A_2|\geq |A_3|$. An $n$-leaf simple level-2 network $N$ is said to be of \textit{type} $L_2^{n, (A_1|A_2|A_3)}$ if up to isomorphism, $N$ is obtained by adding $|A_1|,|A_2|,\mbox{ and }|A_3|$ leaf edges with set of labels given by $A_1,A_2,\mbox{ and } A_3$, respectively, to the top, middle and bottom edge, respectively, of the unrooted level-2 generator $L_2.$
\end{definition}
\begin{example}
The two middle networks in \Cref{fig:Examples of strict simple and semisimple level-2 networks} are of the same type but they have distinct underlying undirected topology.
\begin{figure}[H]
    \centering
    \subfigure[A simple network of type $L_2^{3,(\{1,3\}|\{2\})}$]{
    \begin{tikzpicture}[scale=0.7]
       \fill[black] (0,0) circle (3pt) (4,0) circle (3pt) (1,1) circle (3pt) (3,1) circle (3pt) 
       (2,-1) circle (3pt) (1,2) circle (3pt) (3,2) circle (3pt) (2,-2) circle (3pt) ;
       \node (a1) at (1,2)[above]{1} ; 
       \node (a1) at (3,2)[above]{3} ; 
       \node (a1) at (2,-2)[below]{2} ; 
\draw (0,0)--(4,0);
\draw (0,0)--(1,1);
\draw (1,1) --(3,1);
\draw (3,1)--(4,0);
\draw (1,1)--(1,2);
\draw (3,1)--(3,2);
\draw (2,-1)--(0,0);
\draw (2,-1)--(4,0);
\draw (2,-1)--(2,-2);
    \end{tikzpicture}}
    \qquad
    \subfigure[Two distinct simple networks of type $L_2^{6,(\{4,5,6\}|\{2,3\}|\{1\})}$]{
    \begin{tikzpicture}[scale=0.7]
       \fill[black] (0,0) circle (3pt) (4,0) circle (3pt) (2,2) circle (3pt) (2,3) circle (3pt) (1,0) circle (3pt) (3,0) circle (3pt) (1,-1) circle (3pt) (2,-1) circle (3pt) (3,-1) circle (3pt) (1,-2) circle (3pt) (2,-2) circle (3pt) (3,-2) circle (3pt) (2,3) circle (3pt) (1,0.5) circle (3pt) (3,0.5) circle (3pt) (1,-2) circle (3pt) (2,-2) circle (3pt) (3,-2) circle (3pt);
       \node (a1) at (2,3)[above]{1} ; 
       \node (a1) at (1,0.5)[right]{2} ; 
       \node (a1) at (3,0.5)[left]{3} ;
       \node (a1) at (1,-2)[below]{4} ; 
       \node (a1) at (2,-2)[below]{5} ; 
       \node (a1) at (3,-2)[below]{6} ;
\draw (0,0)--(2,2);
\draw (2,2) --(4,0);
\draw (0,0)--(1,0);
\draw (1,0)--(3,0);
\draw (3,0)--(4,0);
\draw (0,0)--(1,-1);
\draw (1,-1)--(2,-1);
\draw (2,-1)--(3,-1);
\draw(3,-1)--(4,0);
\draw (2,2)--(2,3);
\draw (1,0)--(1,0.5);
\draw (3,0)--(3,0.5);
\draw (1,-1)--(1,-2);
\draw (2,-1)--(2,-2);
\draw (3,-1)--(3,-2);
    \end{tikzpicture}
    \begin{tikzpicture}[scale=0.7]
       \fill[black] (0,0) circle (3pt) (4,0) circle (3pt) (2,2) circle (3pt) (2,3) circle (3pt) (1,0) circle (3pt) (3,0) circle (3pt) (1,-1) circle (3pt) (2,-1) circle (3pt) (3,-1) circle (3pt) (1,-2) circle (3pt) (2,-2) circle (3pt) (3,-2) circle (3pt) (2,3) circle (3pt) (1,0.5) circle (3pt) (3,0.5) circle (3pt) (1,-2) circle (3pt) (2,-2) circle (3pt) (3,-2) circle (3pt);
       \node (a1) at (2,3)[above]{1} ; 
       \node (a1) at (1,0.5)[right]{2} ; 
       \node (a1) at (3,0.5)[left]{3} ;
       \node (a1) at (1,-2)[below]{4} ; 
       \node (a1) at (2,-2)[below]{6} ; 
       \node (a1) at (3,-2)[below]{5} ;
\draw (0,0)--(2,2);
\draw (2,2) --(4,0);
\draw (0,0)--(1,0);
\draw (1,0)--(3,0);
\draw (3,0)--(4,0);
\draw (0,0)--(1,-1);
\draw (1,-1)--(2,-1);
\draw (2,-1)--(3,-1);
\draw(3,-1)--(4,0);
\draw (2,2)--(2,3);
\draw (1,0)--(1,0.5);
\draw (3,0)--(3,0.5);
\draw (1,-1)--(1,-2);
\draw (2,-1)--(2,-2);
\draw (3,-1)--(3,-2);
    \end{tikzpicture}}
    \qquad
    \subfigure[A semisimple network]{
    \begin{tikzpicture}[scale=0.7]
       \fill[black] (0,0) circle (3pt) (4,0) circle (3pt) (2,2) circle (3pt) (1,0) circle (3pt) (3,0) circle (3pt) (1,-1) circle (3pt) (2,-1) circle (3pt) (3,-1) circle (3pt) (1,-2) circle (3pt) (2,-2) circle (3pt) (3,-2) circle (3pt) (1,3) circle (3pt) (1,0.5) circle (3pt) (3,0.5) circle (3pt) (1,-2) circle (3pt) (2,-2) circle (3pt) (3,-2) circle (3pt) (3,3) circle (3pt) (2,2.5) circle (3pt);
       \node (a1) at (1,3)[above]{1} ; 
       \node (a1) at (1,0.5)[right]{2} ; 
       \node (a1) at (3,0.5)[left]{3} ;
       \node (a1) at (1,-2)[below]{4} ; 
       \node (a1) at (2,-2)[below]{6} ; 
       \node (a1) at (3,-2)[below]{5} ;
       \node (a1) at (3,3)[above]{7} ;
\draw (0,0)--(2,2);
\draw (2,2) --(4,0);
\draw (0,0)--(1,0);
\draw (1,0)--(3,0);
\draw (3,0)--(4,0);
\draw (0,0)--(1,-1);
\draw (1,-1)--(2,-1);
\draw (2,-1)--(3,-1);
\draw(3,-1)--(4,0);
\draw (2,2.5)--(1,3);
\draw (1,0)--(1,0.5);
\draw (3,0)--(3,0.5);
\draw (1,-1)--(1,-2);
\draw (2,-1)--(2,-2);
\draw (3,-1)--(3,-2);
\draw (2,2.5)--(3,3);
\draw (2,2.5)--(2,2);
\end{tikzpicture}}
    \caption{Examples of simple and semisimple undirected strict level-2 networks.}
    \label{fig:Examples of strict simple and semisimple level-2 networks}
    \end{figure}
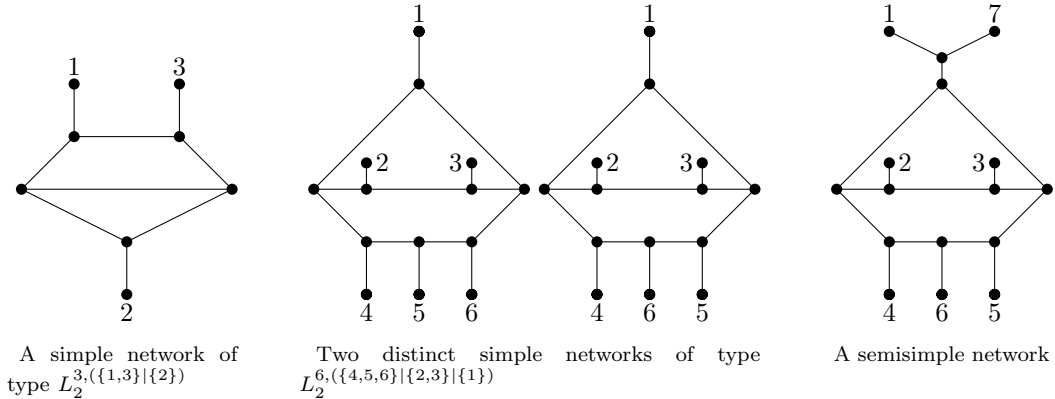
\end{example}
For level-1 networks, the set of $n$-leaf cycle networks is larger than the set of simple networks, which are also known as \textit{sunlet networks}.  Similarly, we would also like to generalize the notion of simple level-2 networks by defining the class of semisimple level-2 networks, which is given in the following definition. By this definition, every simple level-2 network is semisimple. But the converse is not true.

\begin{definition}\label{def:semisimple level-2}
A strict level-2 network with exactly two reticulation vertices is called \textit{semisimple level-2 network}.
\end{definition}

 One can construct a semisimple level-2 network using the following steps. We first start with a network $N$ of type $L_2^{n,(A_1|A_2|A_3)}$ where 
$$A_i=\{a_{i1},a_{i2},\dots,a_{ik_i}\}$$ 
where $a_{ij}\in[n]$ and $k_i\geq 1.$ Suppose that we are given the following set of graphs: 
\begin{align*}
    T=&\{T_{ij}: \mbox{ for }1\leq i\leq 3,1\leq j \leq k_i, T_{ij} \mbox{ is an unrooted binary tree or an empty graph such that for two }\\
    & \mbox{ distinct }T_{\alpha}\mbox{ and }T_{\beta},\mbox{ the set of leaf labels of }T_{\alpha}\mbox{ and }T_{\beta}\mbox{ are disjoint}\}.\\
\end{align*}

We may assume that the union of the set of leaf labels of all $T_{ij}\in T$ is equal to $[m]$ for some $m\in \mathbb{N}$. For each $1\leq i\leq 3$ and $1\leq j\leq  k_i$, we associate to the leaf labeled by $a_{ij}$ the graph $T_{ij}\in T$. If $T_{ij}$ is an unrooted binary tree, then we assume that $T_{ij}$ has a distinguished leaf vertex labeled by $b_{ij}.$ Let $a'_{ij}$ and $b'_{ij}$ be the unique parent of the leaf labeled by $a_{ij}$ and $b_{ij}$, respectively. We then identify the leaf edge adjacent to the leaf vertex $a_{ij}$ and $b_{ij}$ such that the leaf vertex labeled by $a_{ij}$ is identified with $b'_{ij}$ while the leaf vertex labeled by $b_{ij}$ is identified with $a'_{ij}.$ If $T_{ij}$ is an empty set, we do nothing to the leaf labeled by $a_{ij}$. We repeat these identification steps for each $a_{ij}$. At the end of the process, we will obtain an $m$-leaf semisimple strict level-2 network where $$m=\sum_{T_{ij}\neq \emptyset} (L(T_{ij})-1)+|\{T_{ij}|T_{ij}=\emptyset, 1\leq i\leq 3,1\leq j\leq k_i\}|$$
and $L(T_{ij})$ denotes the number of leaves in $T_{ij}.$ If all $T_{ij}$'s are empty, then after these edge identifications, the starting simple network remains unchanged. An illustration of this procedure is displayed in \Cref{semisimple construction}.

 For each $1\leq i\leq 3, 1\leq j\leq k_i$, the leaf labels of $T_{ij}\setminus\{b_{ij}\}$ if $T_{ij}$ is a tree and the singleton $\{a_{ij}\}$ if $T_{ij}$ is empty, are called \textit{branches}. The set of branches $\{A_{11},\dots,A_{1k_1},A_{21},\dots,A_{2k_2},A_{31},\dots,A_{3k_3}\}$ of an $m$-leaf semisimple level-2 network $N$ induces a partition of $[m]$. For instance, the rightmost network in \Cref{fig:Examples of strict simple and semisimple level-2 networks} induces the partition $\{\{1,7\},\{2\},\{3\},\{4\},\{5\},\{6\}\}$ of $[7].$ Suppose that we have a semisimple nice level-2 network $N$ which is given by the partition $P=\{A_1,\dots, A_t\}$ induces by the branches $N.$ For $1\leq i\leq t$, let $e_i=(v_i,w_i)$ be the cut-edge corresponding to the branch $A_i$ such that $v_i$ belongs to the unique strict level-2 biconnected component of $N$. A branch $A_i$ is said to be a \textit{reticulation branch} if $v_i$ is a reticulation vertex. Therefore, the number of reticulation branches of the network $N$ is either 1 or 2, due to \Cref{lemma:r(N)}. 
 
 We will use the same terminologies for cycle networks as well. Namely, suppose that the set of cycle vertices of an $n$-leaf $k$-cycle network $N$ induces a partition $A_1|\dots|A_k$ of the leaf set $[n]$. Then the set of branches of $N$ is given by $\{A_1,\dots, A_k\}$ and a branch $A_i$ is said to be a reticulation branch if the corresponding cycle vertex is a reticulation vertex of $N$.

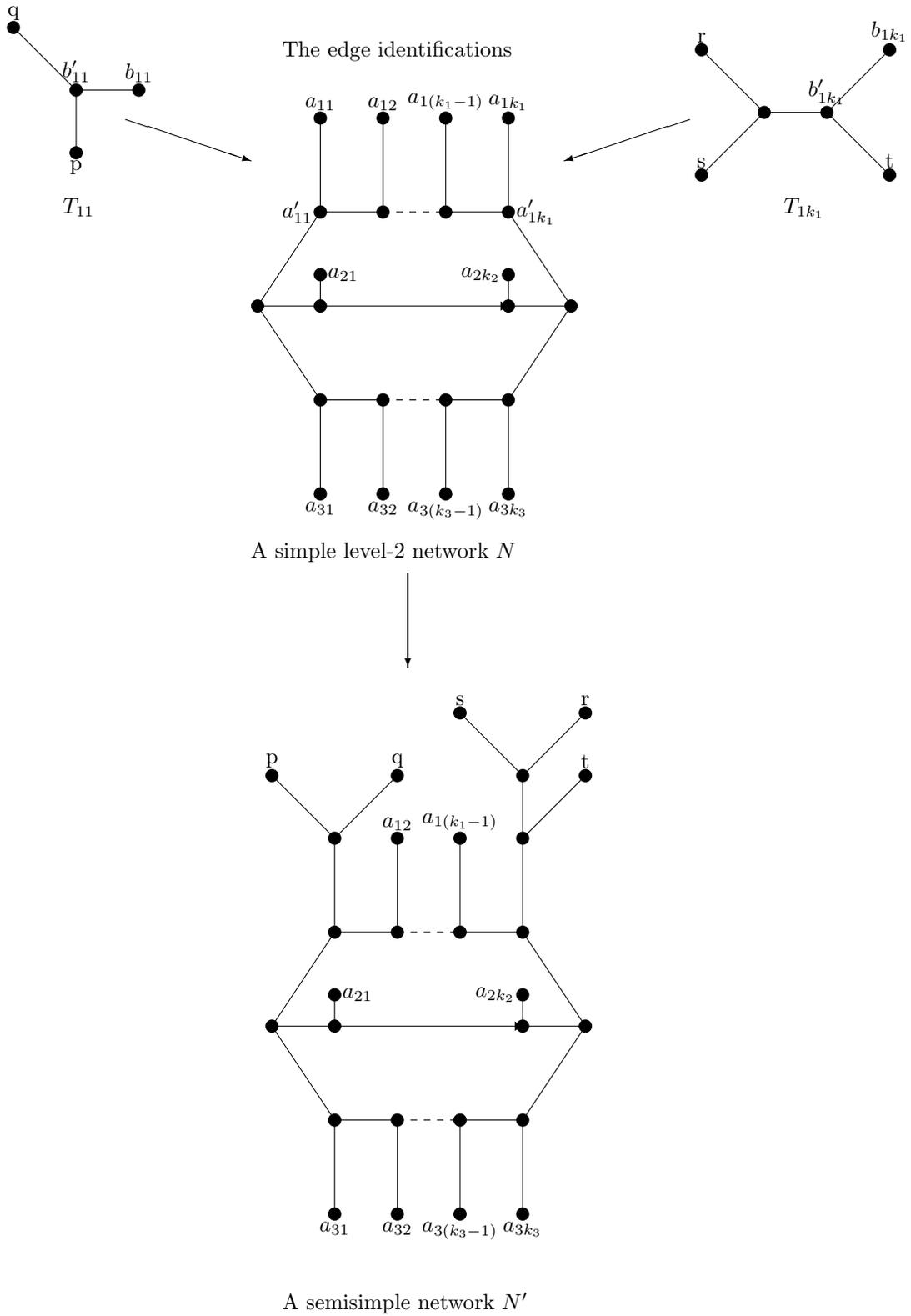
\begin{figure}
\setlength{\unitlength}{1 cm}    
\begin{picture}(15,20)
\put(1,18){\begin{tikzpicture}[scale=1]
       \fill[black] (0,0) circle (3pt) (1,0) circle (3pt) (0,-1) circle (3pt) (-1,1) circle (3pt);
        \node (a1) at (1,0)[above]{$b_{11}$} ; 
       \node (a1) at (0,-1)[below]{p} ; 
       \node (a1) at (-1,1)[above]{q} ;
       \node (a1) at (0,0)[above]{$b'_{11}$} ;
      
\draw (0,0)--(1,0);
\draw (0,0)--(0,-1);
\draw (0,0)--(-1,1);

    \end{tikzpicture}}
\put(5,12.5){\begin{tikzpicture}[scale=1]
       \fill[black] (0,0) circle (3pt) (1,1.5) circle (3pt) (2,1.5) circle (3pt)  (3,1.5) circle (3pt) (4,1.5) circle (3pt) (4,3) circle (3pt) (5,0) circle (3pt) (1,-1.5) circle (3pt) (2,-1.5) circle (3pt)  (3 ,-1.5) circle (3pt) (4,-1.5) circle (3pt)(1,3) circle (3pt) (2,3) circle (3pt)  (3,3) circle (3pt) (1,-3) circle (3pt) (2,-3) circle (3pt)  (3 ,-3) circle (3pt) (4,-3) circle (3pt)  (1,0) circle (3pt) (4,0) circle (3pt) (1,0.5) circle (3pt) (4,0.5) circle (3pt) ;
        \node (a1) at (1,3)[above]{$a_{11}$} ; 
       \node (a1) at (2,3)[above]{$a_{12}$} ; 
       \node (a1) at (3,3)[above]{$a_{1(k_1-1)}$} ;
       \node (a1) at (4,3)[above]{$a_{1k_1}$} ;
       \node (a1) at (1,-3)[below]{$a_{31}$} ; 
       \node (a1) at (2,-3)[below]{$a_{32}$} ; 
       \node (a1) at (3,-3)[below]{$a_{3(k_3-1)}$} ;
       \node (a1) at (4,-3)[below]{$a_{3k_3}$} ;
        \node (a1) at (1,0.5)[right]{$a_{21}$} ;
       \node (a1) at (4,0.5)[left]{$a_{2k_2}$} ;
       \node (a1) at (1,1.5)[left]{$a'_{11}$} ;
       \node (a1) at (4,1.5)[right]{$a'_{1k_1}$} ;
\draw (0,0)--(1,0);
\draw (4,0)--(5,0);
\draw (0,0)--(1,1.5);
\draw (1,1.5) --(2,1.5);
\draw [dashed](2,1.5) --(3,1.5);
\draw (3,1.5)--(4,1.5);
\draw (4,1.5)--(5,0);
\draw (0,0)--(1,-1.5);
\draw (1,-1.5)--(2,-1.5);
\draw [dashed](2,-1.5)--(3,-1.5);
\draw (3,-1.5)--(4,-1.5);
\draw (4,-1.5)--(5,0);
\draw (1,1.5)--(1,3);
\draw (2,1.5)--(2,3);
\draw (3,1.5)--(3,3);
\draw (4,1.5)--(4,3);
\draw (1,-1.5)--(1,-3);
\draw (2,-1.5)--(2,-3);
\draw (3,-1.5)--(3,-3);
\draw (4,-1.5)--(4,-3);
\draw (1,0)--(1,0.5);
\draw [->,>=Triangle](1,0)--(4,0);
\draw (4,0)--(4,0.5);
    \end{tikzpicture}}
\put(12,18){\begin{tikzpicture}[scale=1]
       \fill[black] (0,0) circle (3pt) (1,0) circle (3pt) (2,1) circle (3pt) (2,-1) circle (3pt) (-1,1) circle (3pt) (-1,-1) circle (3pt);
        \node (a1) at (2,1)[above]{$b_{1k_1}$} ; 
       \node (a1) at (1,0)[above]{$b'_{1k_1}$} ;
         \node (a1) at (-1,1)[above]{r} ; 
       \node (a1) at (-1,-1)[above]{s} ;
       \node (a1) at (2,-1)[above]{t} ;
      
\draw (0,0)--(1,0);
\draw (1,0)--(2,1);
\draw (1,0)--(2,-1);
\draw (0,0)--(-1,1);
\draw (0,0)--(-1,-1);

    \end{tikzpicture}}
\put(3,19){\vector(3,-1){2}}
\put(12,19){\vector(-3,-1){2}}
\put(5.5,20){The edge identifications}
\put(2,17.5){$T_{11}$}
\put(13.5,17.5){$T_{1k_1}$}
\put(5,12){A simple level-2 network $N$}
\put(7.5,11.75){\vector(0,-1){1.5}}
\put(5,1){ \begin{tikzpicture}[scale=1]
       \fill[black] (0,0) circle (3pt) (1,1.5) circle (3pt) (2,1.5) circle (3pt)  (3,1.5) circle (3pt) (4,1.5) circle (3pt) (4,3) circle (3pt) (5,0) circle (3pt) (1,-1.5) circle (3pt) (2,-1.5) circle (3pt)  (3 ,-1.5) circle (3pt) (4,-1.5) circle (3pt)(1,3) circle (3pt) (2,3) circle (3pt)  (3,3) circle (3pt) (1,-3) circle (3pt) (2,-3) circle (3pt)  (3 ,-3) circle (3pt) (4,-3) circle (3pt)  (1,0) circle (3pt) (4,0) circle (3pt) (1,0.5) circle (3pt) (4,0.5) circle (3pt)(0,4) circle (3pt) (2,4) circle (3pt) (5,4) circle (3pt) (4,4) circle (3pt) (5,5) circle (3pt) (3,5) circle (3pt) ;
       \node (a1) at (2,3)[above]{$a_{12}$} ; 
       \node (a1) at (3,3)[above]{$a_{1(k_1-1)}$} ;
       \node (a1) at (1,-3)[below]{$a_{31}$} ; 
       \node (a1) at (2,-3)[below]{$a_{32}$} ; 
       \node (a1) at (3,-3)[below]{$a_{3(k_3-1)}$} ;
       \node (a1) at (4,-3)[below]{$a_{3k_3}$} ;
        \node (a1) at (1,0.5)[right]{$a_{21}$} ;
       \node (a1) at (4,0.5)[left]{$a_{2k_2}$} ;
        \node (a1) at (0,4)[above]{p};
         \node (a1) at (2,4)[above]{q};
          \node (a1) at (5,4)[above]{t};
           \node (a1) at (5,5)[above]{r};
            \node (a1) at (3,5)[above]{s};
\draw (0,0)--(1,0);
\draw (4,0)--(5,0);
\draw (0,0)--(1,1.5);
\draw (1,1.5) --(2,1.5);
\draw [dashed](2,1.5) --(3,1.5);
\draw (3,1.5)--(4,1.5);
\draw (4,1.5)--(5,0);
\draw (0,0)--(1,-1.5);
\draw (1,-1.5)--(2,-1.5);
\draw [dashed](2,-1.5)--(3,-1.5);
\draw (3,-1.5)--(4,-1.5);
\draw (4,-1.5)--(5,0);
\draw (1,1.5)--(1,3);
\draw (2,1.5)--(2,3);
\draw (3,1.5)--(3,3);
\draw (4,1.5)--(4,3);
\draw (1,-1.5)--(1,-3);
\draw (2,-1.5)--(2,-3);
\draw (3,-1.5)--(3,-3);
\draw (4,-1.5)--(4,-3);
\draw (1,0)--(1,0.5);
\draw [->,>=Triangle](1,0)--(4,0);
\draw (4,0)--(4,0.5);
\draw (1,3)--(0,4);
\draw (1,3)--(2,4);
\draw (4,3)--(5,4);
\draw (4,3)--(4,4);
\draw (4,4)--(5,5);
\draw (4,4)--(3,5);
    \end{tikzpicture}}
\put(5.5,0){A semisimple network $N'$}
\end{picture}
\caption{An illustration for the identification process given a simple level-2 network $N$ and two trees $T_{11}$ and $T_{1k_1}$. The resulting network after edge identifications is the semisimple level-2 network $N'$, which has three more leaves.}
\label{semisimple construction}
   \end{figure}  
  The following lemma provides some properties of a semisimple level-2 network. Let us recall that it is assumed that $|X|\geq 2$ and the networks in our consideration are orientable.

\begin{lemma}\label{prop:A_icontains reticulation leaf}
Let $N$ be a simple nice strict level-2 network. Then $N$ contains a reticulation leaf. 
\end{lemma}
\begin{proof}
 The network $N$ has a reticulation vertex as it is a strict level-2 network. Since $N$ is orientable, then it has at least one orientation. Let $N'$ be an orientation of $N$ with a root $\rho$. Suppose that in $N'$, the vertex $v$ is a reticulation vertex with maximum distance to $\rho$. Let $p$ be the longest path from $\rho$ to $v$. Let $v_1$ and $v_2$ be the parents of $v$. If for some $i\in\{1,2\}$, the vertex $v_i$ is a reticulation vertex, then $v_i$ is a funnel vertex of type A, contradicting our hypothesis that $N$ is nice. Therefore, $u_1$ and $u_2$ are both tree vertices.
 
 The path $p$ passes through either $v_1$ or $v_2$. Without loss of generality, we may assume that $p$ passes through $v_1$. Let $u_1$ be the child of $v_1$ that is not $v$. The vertex $u_1$ is either a tree vertex or a leaf. Otherwise, $v_1$ is a funnel vertex of type B, a contradiction. By our construction, no reticulation vertices can be reached by a directed path from $u_1$ or $v$ as $v$ is a reticulation vertex with maximum distance to the root $\rho$. If there are at least two leaves which are reachable from $v$, then $N$ contains a tree cherry. This statement contradicts our hypothesis as $N$ is simple and hence does not contain a tree cherry. It implies that the only vertex that can be reached from $v$ is a single leaf $x$. Thus, $x$ is a reticulation leaf in $N'$ and hence in $N$ after collapsing the root $\rho$ and contracting all degree two vertices. 
\end{proof}
  
\begin{lemma}\label{lemma:different A_i}
 Let $N$ be a semisimple nice strict level-2 network obtained from a simple network of type $L_2^{n,(A_1|A_2|A_3)}$ by the identification process mentioned earlier. Suppose that $N$ contains two reticulation branches. Then the two reticulation branches correspond to two different $A_i$'s.
\end{lemma}
\begin{proof}
If two reticulation branches correspond to a single set $A_i$, then the network $N$ will not be orientable as there will be no valid root location \cite{huber2019rooting}.
\end{proof}

Let us recall that a simple level-$k$ network is a biconnected network with at most $k$ reticulation vertices. Since there are at most $k$ reticulation vertices in a simple level-$k$ network, then there are at most $k$ reticulation leaves in the network. The following lemmas provide a characterization of a simple nice level-2 network based on the reticulation leaves of the network. 
\begin{lemma}\label{lemma:r(N)}
 Let $N$ be a simple nice level-2 network. Let $r(N)$ denote the set of reticulation leaves of $N$. If $r(N)$ is empty, then $N$ is a tree. If $|r(N)|=1$, then $N$ is either a simple strict level-1 or strict level-2 network. If $|r(N)|=2$, then $N$ is a simple strict level-2 network.
 \end{lemma}
 \begin{proof}
  If $N$ is a tree, then $r(N)$ is empty. If $r(N)$ is empty, then \Cref{prop:A_icontains reticulation leaf} tells us that $N$ can not be a strict level-2 network. If $N$ is a simple strict level-1 network, then $r(N)=1$. If $N$ is a strict level-2 network with $r(N)$ is empty, then the two purely interior vertices of the network are reticulation vertices. By \Cref{prop:A_icontains reticulation leaf}, we then reach a contradiction as $N$ does not have a reticulation leaf. Hence, if $r(N)$ is empty, then $N$ is a tree. The number $|r(N)|$ can not be bigger than two as $N$ is a simple level-2 network. If $|r(N)|=2,$ then $N$ can not be level-1 network as simple level-1 network has exactly one cycle of length at least three with only one reticulation vertex. 
 \end{proof}

The instances of 4-leaf simple nice strict level-2 network $N$ such that the number $|r(N)|\in\{1,2\}$ are presented in \Cref{fig:unrooted four leaves}. In this figure, the number of reticulation leaves of type 1 and type 2 networks is 2 and 1, respectively.

\section{Distinguishing 4-leaf nice level-2 networks}\label{section:distinguishing nice networks}

In this section, we will be concerned with the 4-leaf nice level-2 networks. In particular, we will derive some relationships between the varieties associated with level-2 networks with four leaves. These relationships are obtained using phylogenetic invariants associated with each network model. These results will serve as a foundation towards distinguishing level-2 networks with at least five leaves.

We are now ready to present two results on the number of small simple level-2 networks.
\begin{lemma}\label{lemma:small simple nice level-2}
 There is no simple nice strict level-2 network with two or three leaves. Additionally, up to isomorphism and labeling of the leaves, there are four distinct simple nice strict level-2 networks.
\end{lemma}
\begin{proof}
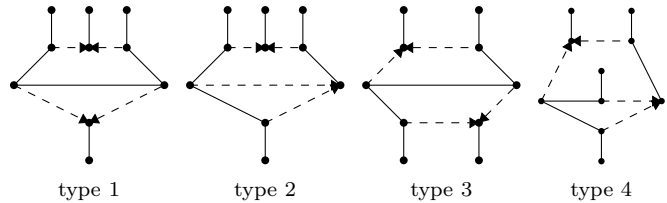
\begin{figure}[H]
    \centering
    \subfigure[type 1]{
    \begin{tikzpicture}[scale=0.5]
       \fill[black] (0,0) circle (3pt) (1,1) circle (3pt) (2,1) circle (3pt)  (3,1) circle (3pt) (4,0) circle (3pt) (2,-1) circle (3pt) (1,2) circle (3pt) (2,2) circle (3pt)  (3,2) circle (3pt) (2,-2) circle (3pt);
\draw (0,0)--(4,0);
\draw (0,0)--(1,1);
\draw [dashed,->,>=Triangle](1,1) --(2,1);
\draw [dashed,<-,>=Triangle](2,1) --(3,1);
\draw (3,1)--(4,0);
\draw [dashed,->,>=Triangle](0,0)--(2,-1);
\draw [dashed,<-,>=Triangle](2,-1)--(4,0);
\draw (1,1)--(1,2);
\draw (2,1)--(2,2);
\draw (3,1)--(3,2);
\draw (2,-1)--(2,-2);
    \end{tikzpicture}}
    \subfigure[type 2]{
    \begin{tikzpicture}[scale=0.5]
       \fill[black] (0,0) circle (3pt) (1,1) circle (3pt) (2,1) circle (3pt)  (3,1) circle (3pt) (4,0) circle (3pt) (2,-1) circle (3pt) (1,2) circle (3pt) (2,2) circle (3pt)  (3,2) circle (3pt) (2,-2) circle (3pt);
\draw [dashed,->,>=Triangle](0,0)--(4,0);
\draw (0,0)--(1,1);
\draw [dashed,->,>=Triangle](1,1) --(2,1);
\draw [dashed,<-,>=Triangle](2,1) --(3,1);
\draw (3,1)--(4,0);
\draw (0,0)--(2,-1);
\draw [dashed,->,>=Triangle](2,-1)--(4,0);
\draw (1,1)--(1,2);
\draw (2,1)--(2,2);
\draw (3,1)--(3,2);
\draw (2,-1)--(2,-2);
    \end{tikzpicture}}
    \subfigure[type 3]{
    \begin{tikzpicture}[scale=0.5]
       \fill[black] (0,0) circle (3pt) (1,1) circle (3pt) (3,1) circle (3pt) (4,0) circle (3pt) (1,-1) circle (3pt) (3,-1) circle (3pt) (1,2) circle (3pt)  (3,2) circle (3pt) (1,-2) circle (3pt) (3,-2) circle (3pt);
\draw (0,0)--(4,0);
\draw [dashed,->,>=Triangle](0,0)--(1,1);
\draw [dashed,<-,>=Triangle](1,1) --(3,1);
\draw (3,1) --(4,0);
\draw (0,0)--(1,-1);
\draw [dashed,->,>=Triangle](1,-1)--(3,-1);
\draw [dashed,<-,>=Triangle](3,-1)--(4,0);
\draw (1,1)--(1,2);
\draw (3,1)--(3,2);
\draw (1,-1)--(1,-2);
\draw (3,-1)--(3,-2);
    \end{tikzpicture}}
    \subfigure[type 4]{
     \begin{tikzpicture}[scale=0.4]
       \fill[black] (0,0) circle (3pt) (1,2) circle (3pt) (3,2) circle (3pt) (4,0) circle (3pt) (2,-1) circle (3pt)  (1,3) circle (3pt)  (3,3) circle (3pt) (2,-2) circle (3pt) (2,0) circle (3pt) (2,1) circle (3pt);
\draw [dashed,->,>=Triangle](0,0)--(1,2);
\draw [dashed,<-,>=Triangle](1,2) --(3,2);
\draw (3,2) --(4,0);
\draw (0,0)--(2,-1);
\draw [dashed,->,>=Triangle](2,-1)--(4,0);
\draw (0,0)--(2,0);
\draw [dashed,->,>=Triangle](2,0)--(4,0);
\draw (1,2)--(1,3);
\draw (3,2)--(3,3);
\draw (2,0)--(2,1);
\draw (2,-1)--(2,-2);
    \end{tikzpicture}}
    \caption{4-leaf simple nice strict level-2 networks.}
    \label{fig:unrooted four leaves}
\end{figure}
In \Cref{ex:unrooted-2-leaves}, each of the two possible 2-leaf simple level-2 networks contains a funnel vertex. Similarly, in \Cref{fig:Five distinct $3$-leaf network}, each of the five possible 3-leaf simple strict level-2 networks contains a funnel vertex. The first assertion of the lemma follows \Cref{lemma:funnel-free and nice}. For four leaves case, by \Cref{prop:A_icontains reticulation leaf}, at least one of the leaves should be a reticulation leaf. There are four distinct 4-leaf simple nice strict level-2 networks displayed in \Cref{fig:unrooted four leaves}. These four networks are obtained from the three possible unrooted 4-leaf simple networks. If the network contains two reticulation leaves, then we obtain network of type 1 and 3. If the network contains one reticulation leaf, then we obtain network of type 2 and 4.
\end{proof}

\begin{lemma}
 Up to isomorphism and labeling of the leaves, there are $\{\frac{(n+3)^2}{12}\}-1$ distinct unrooted simple strict level-2 undirected network topologies with $n$ leaves where $\{x\}$ denotes the nearest integer to $x$.
\end{lemma}
\begin{proof}
Let $p(n,m)$ be the number of distinct partitions of the integer $n$ into at most $m$ parts where $1\leq m \leq n$. Since by \Cref{lemma:characterization of generator}, an unrooted level-2 generator is the 3-regular graph on two vertices, then  the number of distinct unrooted undirected $n$-leaf simple strict level-2 topologies is equal to $p(n,3)-p(n,1)$. It is immediate that $p(n,1)=1$. Additionally, it is shown in \cite[Chapter 6]{andrews2004integer} that $p(n,3)=\{\frac{(n+3)^2}{12}\}$ using generating function method.
\end{proof}

The following lemma states that semisimple nice strict level-2 networks that are not simple should have at least five leaves.

\begin{lemma}\label{lemma:small semisimple networks}
 There is no semisimple nice strict level-2 network with two or three leaves. The only 4-leaf semisimple nice strict level-2 networks are given by the simple ones.
\end{lemma}
\begin{proof}
The lemma follows from \Cref{lemma:small simple nice level-2} and the fact that a non-simple semisimple network has strictly more edges than the corresponding simple ones.
\end{proof}

 We will now discuss some earlier works on level-1 networks. It has been shown in \cite{gross2018distinguishing} that modulo the linear invariants induced by the labeling of the group elements, under the JC or K2P model, the phylogenetic ideal of a level-1 network with three leaves is the zero ideal while under the K3P model, a single invariant of degree three was found. On the other hand, under the JC, K2P or K3P model, two distinct 4-leaf 4-cycle level-1 networks are distinguishable  \cite[Lemma 1]{gross2020distinguishing}. 
 
We conclude this section by providing some results on the non-containment of varieties associated with 4-leaf nice level-2 networks, which are given by simple networks due to \Cref{lemma:small simple nice level-2} and \Cref{lemma:small semisimple networks}.
\begin{proposition}\label{lemma:containment level-2 4 leaves}
Let $N_1$ and $N_2$ be two distinct 4-leaf level-2 networks.
\begin{itemize}
    \item[(i)]If $N_1$ is a simple nice strict level-2 network and $N_2$ is a level-1 network, then under the JC, K2P or K3P model, $V_{N_1}\not\subseteq V_{N_2}.$
    \item[(ii)] Suppose that $N_1$ and $N_2$ are both simple nice strict level-2 networks. If $N_1$ is of type 1,2 or 3 and $N_2$ is of type 4, then under the JC model, $V_{N_2}\not\subseteq V_{N_1}.$
\end{itemize}
\end{proposition}
\begin{proof}
Let us first give the main ideas of the proof. To prove the proposition and to reduce the number of computations to obtain the invariants, we consider the following approach from \cite[Lemma 1]{gross2020distinguishing}.  Suppose that we want to first prove (i) for a single-triangle network $N_2$. It is enough to provide a single invariant that vanishes on exactly one of the single-triangle network varieties, but on none of the simple nice strict level-2 network varieties. We assume the second statement is not true. Then there exist two networks $N_1$ and $N_2$, where $N_1$ is a simple nice level-2 network and $N_2$ is a single-triangle network such that $V_{N_1}\subseteq V_{N_2}$. Equivalently, we have that $I_{V_{N_2}}\subseteq I_{V_{N_1}}.$ We can obtain any single-triangle network $N_2'$ from $N_2$ by permuting the leaf labels using some permutation $\sigma \in S_4$. If we permute the leaf labels of $N_1$ according to permutation $\sigma$, then $I_{V_{N_2'}}\subseteq I_{V_{N_1'}}.$ Thus, if we can find a single single-triangle network whose vanishing ideal is not contained in the vanishing ideal of any 4-leaf simple nice strict level-2 networks, then we will arrive at a contradiction. A similar method can be used to prove the statement when $N_2$ is a tree, a double-triangle, or a 4-cycle network.

Let us consider the following three polynomials in the polynomial ring $\mathbb{C}[q_{g_1\dots g_4}:g_i\in \{A,C,G,T\}]$:
\begin{align*}
    f_1&=q_{CTTC}-q_{GCGC},\\
     f_2&=q_{AACC}q_{CGCG}q_{GAGA}q_{TAAT} - q_{AACC}q_{CGAT}q_{GAGA}q_{TACG}\\
     &+q_{AACC}q_{CAGT}q_{GGAA}q_{TACG} - q_{AAAA}q_{CAGT}q_{GGCC}q_{TACG}, \mbox{ and }\\
     f_3&=q_{TGGT}^3-q_{TGGT}^2q_{TGGT}-q_{TGGT}q_{TGTG}^2+q_{TGTG}^3-q_{TGGT}^2q_{TTGG}+2q_{TGGT}q_{TGTG}q_{TTGG}\\&-q_{TGTG}^2q_{TTGG}-q_{TGGT}q_{TTGG}^2-q_{TGTG}q_{TTGG}^2+q_{TTGG}^3-q_{TGGT}^2q_{TTTT}+2q_{TGGT}q_{TGTG}q_{TTTT}\\&-q_{TGTG}^2q_{TTTT}+2q_{TGGT}q_{TTGG}q_{TTTT}+2q_{TGTG}q_{TTGG}q_{TTTT}-q_{TTGG}^2q_{TTTT}-q_{TGGT}q_{TTTT}^2\\&-q_{TGTG}q_{TTTT}^2-q_{TTGG}q_{TTTT}^2+q_{TTTT}^3.\\
\end{align*}

The polynomials $f_1$ and $f_2$ were also used in \cite[Lemma 1]{gross2020distinguishing}. All the following statements can be verified using \textit{Macaulay2} code that is available at
{\fontfamily{pcr}\selectfont https://github.com/ardiyam1/Distinguishing-Level-2\\-Phylogenetic-Networks-Using-Phylogenetic-Invariants}.

Under the JC model, the polynomial $f_3$ vanishes on the parameterization of one of the trees, one of the single-triangle or one of  the double-triangle networks but on none of the simple nice strict level-2 networks. Additionally, the polynomial $f_2$ vanishes on the parameterization of one of the 4-cycle networks but on none of the strict simple nice level-2 networks. Thus, in all cases under the JC model, there exists a polynomial $f$ such that $f\in I_{V_{N_2}}$ but $f\notin I_{V_{N_1}}$. For the K2P or K3P models, the polynomial $f_2$ vanishes on the parameterization of one of the tree, single-triangle, double-triangle, or 4-cycle networks but on none of the strict simple nice level-2 networks. Altogether, this completes the proof of the first part of the proposition.

Under the JC model, the polynomial $f_1$ vanishes on the parameterization of network type 1, 2, or 3 but not on the parameterization of network type 4. Thus, $f_3\in I_{V_{N_1}}$ but $f_3\notin I_{V_{N_2}}$, which implies that $I_{V_{N_1}}\not\subseteq I_{V_{N_2}}$. Reversing the non-inclusion, we have that  $V_{N_2}\not\subseteq V_{N_1}$, verifying the second statement. 

\end{proof}
 
 Together with earlier results in \cite{gross2018distinguishing,gross2020distinguishing}, the following table summarizes the relationships between varieties associated with a tree, level-1 and simple nice level-2 networks with four leaves under the JC, K2P, or K3P models. The non-containment of varieties associated with type 4 and the other types of networks is only valid under the JC model.
\begin{table}[H]
\begin{center}
\begin{tabular}{ |c|c|c|c|c|c|c| } 
 \hline
  & type 1, 2, or 3& type 4& single-triangle & double-triangle& 4-cycle & tree \\ \hline
  type 1, 2, or 3&- &$\not\supseteq$&$\not\subseteq$&$\not\subseteq$&$\not\subseteq$&$\not\subseteq$\\\hline
  type 4& $\not\subseteq$&-& $\not\subseteq$&$\not\subseteq$&$\not\subseteq$&$\not\subseteq$\\\hline
  single-triangle&$\not\supseteq$&$\not\supseteq$&$\sim$&$\not\supseteq$&$\not\supseteq$&$\not\subseteq$\\\hline
  double-triangle&$\not\supseteq$&$\not\supseteq$&$\not\subseteq$&$\sim$&$\checkmark$&$\not\subseteq$\\\hline
  4-cycle&$\not\supseteq$&$\not\supseteq$&$\not\subseteq$&$\checkmark$&$\checkmark$&$\not\subseteq$\\\hline
  tree&$\not\supseteq$&$\not\supseteq$&$\not\supseteq$&$\not\supseteq$&$\not\supseteq$&$\checkmark$\\
 \hline
\end{tabular}
\end{center}
\caption{An overview for two distinct 4-leaf level-2 networks $N_1$ and $N_2$. The rows of the table represent the topology of $N_1$ and the columns of the table represent the topology of $N_2$. We denote by $\not\subseteq$ if $V_{N_1}\not\subseteq V_{N_2}$ and by  $\not\supseteq$ if $V_{N_1}\not\supseteq V_{N_2}$. Furthermore, if $N_1$ and $N_2$ are distinguishable, then we denote it by $\checkmark$ while we will denote it by $\sim$ if they are distinguishable given that the underlying undirected topologies are distinct. Lastly, the notation ‘-’ means that we can not determine whether $V_{N_1}\not\subseteq V_{N_2}$ or $V_{N_2}\not\subseteq V_{N_1}$ given two networks of the same type.}
\label{table}
\end{table}

\section{Distinguishing simple nice level-2 networks with at least five leaves}\label{subsub:Distinguishing simple level-2 networks with at least five leaves}

In this section, we will study whether two simple nice level-2 networks with at least five leaves are distinguishable. As mentioned in the first section, data on the set of subnetworks may be implemented to recover the evolutionary history of a set of species. Thus, it is reasonable to restrict the analysis of a network with many leaves to some smaller subnetworks and then study these subnetworks. \Cref{def:restriction of network} describes a restriction procedure which provides a way to obtain a subnetwork given a subset of the leaves.

\begin{definition}[\cite{gross2018distinguishing}, Definition 4.1]\label{def:restriction of network}
Let $N$ be an $n$-leaf semi-directed network on $X$ and $S\subseteq X$ be a subset of the leaves. The \textit{restriction of $N$ to $S$}, denoted by $N|_S$, is the semi-directed network on $S$ obtained from $N$ by first taking the union of all directed paths connecting any two leaves in $S$, contracting all degree two vertices to one of its neighbors, and finally removing all parallel edges. In this definition, we regard the undirected edges as bidirected.
\end{definition}
\begin{example}
\Cref{fig:restriction} displays the restriction of the network $N$ to a subset $S=\{a,c,d,e\}$, which is a 4-leaf 4-cycle network with a reticulation vertex $d$. Moreover, $N|_S$ is a nice network. If we take $S'=\{b,c,d,e\}$, then $N|_S$ is not a nice network as depicted in \Cref{fig:restriction not nice}.

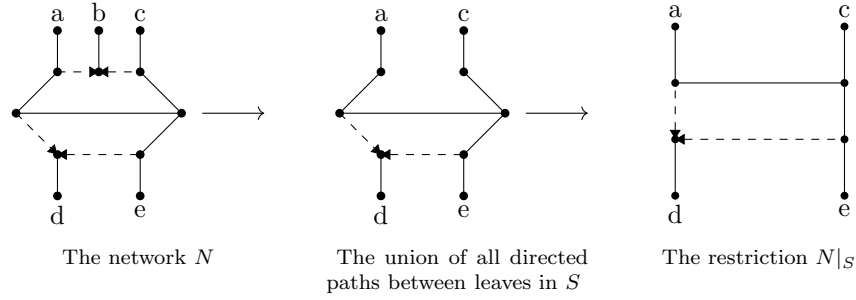
\begin{figure}[H]
    \centering
    \subfigure[The network $N$]{
    \begin{tikzpicture}[scale=0.55]
       \fill[black] (0,0) circle (3pt) (1,1) circle (3pt) (2,1) circle (3pt)  (3,1) circle (3pt) (4,0) circle (3pt) (1,-1) circle (3pt) (3,-1) circle (3pt) (1,2) circle (3pt) (2,2) circle (3pt)  (3,2) circle (3pt) (3,-2) circle (3pt) (1,-2) circle (3pt) (4.5,0) (6,0);
        \node (a1) at (1,2)[above]{a} ; 
       \node (a1) at (2,2)[above]{b} ; 
       \node (a1) at (3,2)[above]{c} ;
       \node (a1) at (1,-2)[below]{d} ; 
       \node (a1) at (3,-2)[below]{e} ; 
\draw (0,0)--(4,0);
\draw (0,0)--(1,1);
\draw [dashed,->,>=Triangle](1,1) --(2,1);
\draw [dashed,<-,>=Triangle](2,1) --(3,1);
\draw (3,1)--(4,0);
\draw [dashed,->,>=Triangle](0,0)--(1,-1);
\draw [dashed,<-,>=Triangle](1,-1)--(3,-1);
\draw (3,-1)--(4,0);
\draw (1,1)--(1,2);
\draw (2,1)--(2,2);
\draw (3,1)--(3,2);
\draw (1,-1)--(1,-2);
\draw (3,-1)--(3,-2);
\draw [->] (4.5,0)--(6,0);
    \end{tikzpicture}}
    \qquad
    \subfigure[The union of all directed paths between leaves in $S$]{
    \begin{tikzpicture}[scale=0.55]
       \fill[black] (0,0) circle (3pt) (1,1) circle (3pt) (3,1) circle (3pt) (4,0) circle (3pt) (1,-1) circle (3pt) (3,-1) circle (3pt) (1,2) circle (3pt) (3,2) circle (3pt) (3,-2) circle (3pt) (1,-2) circle (3pt) (4.5,0) (6,0);
        \node (a1) at (1,2)[above]{a} ; 
       \node (a1) at (3,2)[above]{c} ;
       \node (a1) at (1,-2)[below]{d} ; 
       \node (a1) at (3,-2)[below]{e} ; 
\draw (0,0)--(4,0);
\draw (0,0)--(1,1);
\draw (3,1)--(4,0);
\draw [dashed,->,>=Triangle](0,0)--(1,-1);
\draw [dashed,<-,>=Triangle](1,-1)--(3,-1);
\draw (3,-1)--(4,0);
\draw (1,1)--(1,2);
\draw (3,1)--(3,2);
\draw (1,-1)--(1,-2);
\draw (3,-1)--(3,-2);
\draw [->] (4.5,0)--(6,0);
    \end{tikzpicture}}
    \qquad
    \subfigure[The restriction $N|_S$]{
    \begin{tikzpicture}[scale=0.75]
       \fill[black] (0,0) circle (2pt) (3,0) circle (2pt) (3,1) circle (2pt) (0,1) circle (2pt) (0,-1) circle (2pt) (3,-1) circle (2pt) (0,2) circle (2pt) (3,2) circle (2pt);
        \node (a1) at (0,-1)[below]{d} ; 
       \node (a1) at (3,-1)[below]{e} ;
       \node (a1) at (3,2)[above]{c} ; 
       \node (a1) at (0,2)[above]{a} ; 
\draw (3,0)--(3,1);
\draw (0,1)--(3,1);
\draw (0,0)--(0,-1);
\draw (3,0)--(3,-1);
\draw (3,1)--(3,2);
\draw (0,1)--(0,2);
\draw [dashed,<-,>=Triangle](0,0)--(3,0);
\draw [dashed,<-,>=Triangle](0,0)--(0,1);
    \end{tikzpicture}}
    \caption{The restriction procedure of the network $N$ to a subset $S=\{a,c,d,e\}$.}
    \label{fig:restriction}
\end{figure}
\begin{figure}[H]
    \centering
  \subfigure[]{
    \begin{tikzpicture}[scale=0.55]
        \fill[black] (0,0) circle (3pt) (1,1) circle (3pt)  (3,1) circle (3pt) (4,0) circle (3pt) (1,-1) circle (3pt) (3,-1) circle (3pt)  (1,2) circle (3pt)  (3,2) circle (3pt) (3,-2) circle (3pt) (1,-2) circle (3pt);
       \node (a1) at (1,2)[above]{b} ; 
       \node (a1) at (3,2)[above]{c} ;
       \node (a1) at (1,-2)[below]{d} ; 
       \node (a1) at (3,-2)[below]{e} ; 
\draw (0,0)--(4,0);
\draw [dashed,->,>=Triangle](0,0) --(1,1);
\draw [dashed,<-,>=Triangle](1,1) --(3,1);
\draw (3,1)--(4,0);
\draw [dashed,->,>=Triangle](0,0)--(1,-1);
\draw [dashed,<-,>=Triangle](1,-1)--(3,-1);
\draw (3,-1)--(4,0);
\draw (1,1)--(1,2);
\draw (3,1)--(3,2);
\draw (1,-1)--(1,-2);
\draw (3,-1)--(3,-2);
    \end{tikzpicture}}
    \caption{The restriction $N|_{S'}$, which is not a nice network.}
    \label{fig:restriction not nice}
\end{figure}
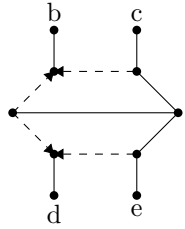
\end{example}

As we can see from the previous example, a restriction of a nice network needs not be nice. If $N$ is a tree, then the restriction to any subset of the leaves is again a tree, which is nice by \Cref{lemma:funnel-free and nice}. We would like to know whether a restriction of a nice network to any subset of its leaves is again nice network. The following sequence of technical lemmas present some properties of the restriction of semisimple nice level-2 networks. In what follows, we will denote by $r(N)$ the set of reticulation leaves of a simple network $N$. Similarly, the number $R(N)$ denotes the set of reticulation branches of semisimple network $N$.

\begin{lemma}\label{lemma:level of restriction}
 \begin{enumerate}
     \item Let $N$ be a cycle network on $X$. Then for any $S\subseteq X\setminus R(N)$, $N|_S$ is a tree. This statement also holds when $N$ is a semisimple strict level-2 network on $X$ with two reticulation branches.
     \item Let $N$ be a cycle network on $X$. Then for any $S\subseteq X$, $N|_S$ is a level-1 network.
     \item Let $N$ be a semisimple nice strict level-2 network on $X$ and $R(N)=\{A_1,A_2\}$. Let $S\subseteq X$ be any subset such that $S$ does not intersect both $A_1$ and $A_2$. Then $N|_S$ is a level-1 network. 
      \item Let $N$ be a semisimple nice strict level-2 network on $X$ and $R(N)=\{A_1\}$. Let $S\subseteq X$ be any subset such that $S$ does not intersect $A_1$. Then $N|_S$ is a level-1 network. 
       \item Let $N$ be a semisimple nice strict level-2 network on $X$ and $R(N)=\{A_1,A_2\}$. If $S\subseteq X$ is any subset containing two distinct elements $s_1,s_2\in S$ such that $s_1\in A_1$ and $s_2\in A_2$, then $N|_S$ is a strict level-2 network.
       \item Let $N$ be a semisimple nice strict level-2 network on $X$ and $R(N)=\{A_1\}$. Suppose that $N$ is obtained from a simple network of type $L_2^{n,(B_1|B_2|B_3)}$ by the identification process mentioned earlier and suppose that $A_1$ is contained in $B_i$. 
       \begin{enumerate}
           \item  If $S\subseteq X$ is any subset containing two distinct elements $s_1\in A_1$ and $s_2\in B_j$ for some $j\neq i$, then $N|_S$ is a strict level-2 network.
           \item If $S\subseteq X$ is any subset containing an element $s\in S\cap A_1$ and for all $s'\in S\setminus\{s\},s'\in B_i$, then $N|_S$ is a level-1 network.
       \end{enumerate}
 \end{enumerate}
\end{lemma}
\begin{proof}
\begin{enumerate}
    \item We will prove the statement for a cycle network $N$. Let $S\subseteq X\setminus R(N).$ If the restriction $N|_S$ contains a cycle, then $S$ should contain an element of a reticulation branch corresponding to the unique reticulation vertex of $N$, a contradiction. Similar arguments will work for the other cases as well.
    \item Suppose that $R(N)=\{A_1\}$. Furthermore, let $S$ be the the unique cycle of $N$ and $a\in S$ be the reticulation vertex of $N$.  If $S\subseteq X\setminus R(N)$, then part (1) applies. If $S\subseteq A_1$, then it is immediate $N|_S$ is a tree as well. Suppose now that $S\cap A_1$ is not empty but $S\not\subseteq A_1$. If $S\setminus A_1$ is contained in a single branch of $N$, let us say $A_2$, then $N|_S$ is a tree. Indeed,  suppose that the branch $A_2$ is connected to $S$ via the cut-edge $e=(v,w)$ where $v\in S$. For every $s\in A_2$ and $t\in A_1$, any path connecting $s$ and $t$ passes through $v$ and $a$. In the restriction procedure $N|_S$, after collapsing all degree two vertices, we have a parallel edge connecting $v$ and $a$. This parallel edge will be further collapsed to just a single edge connecting $v$ and $a$, which gives us a tree.
    
    Now suppose that $S\setminus A_1$ is not contained in a single branch of $N$. Then there exists $s_1,s_2\in S$ such that $s_1$ belongs to, let us say, the branch $A_3$ and $s_2$ belongs to the branch $A_4$ where $A_3\neq A_4$. Suppose that the branch $A_3$ is connected to $S$ via the cut-edge $e_3=(v_3,w_3)$ where $v_3\in S$ and the branch $A_4$ is connected to $S$ via the cut-edge $e_4=(v_4,w_4)$ where $v_4\in S$.  For any $t\in A_1$, any path connecting $s_1$ or $s_2$ to $t$ pass through $v_3,v_4,$ and $a$. Thus, in the restriction procedure $N|_S$, the unique cycle in $N|_S$ has size at least three. Hence, in this case, $N|_S$ is a strict level-1 network.
    
    \item Suppose now that If $S$ does not intersect $A_1$ and $A_2$, then by part (1), $N|_S$ is a tree. Without loss of generality, we now suppose that $S$ intersects $A_1$ but not $A_2$. Then $N|_S$ can not be a strict level-2 network. Indeed, if $N|_S$ is a strict level-2 network, then $N|_S$ contains a strict level-2 biconnected component, let us say $B$. This level-2 biconnected component $B$ is unique as $N$ is semisimple. Moreover, $B$ contains two reticulation vertices. Hence, $S$ contains at least two leaf labels, one contained in $A_1$ and the other is contained in $A_2$, a contradiction.
    \item The proof of (4) is similar to (3).
    
    \item Let $B$ be the unique strict level-2 biconnected component of $N$. Suppose that $s_1$ and $s_2$ are connected to $B$ via cut-edges $e_1=(v_1,w_1)$ and $e_2=(v_2,w_2)$, respectively, where $v_1,v_2\in B.$ Any path connecting $s_1$ and $s_2$ passes through $v_1$ and $v_2$ and it contains the two purely interior vertices of $N.$ Indeed, if there exists a path connecting $s_1$ and $s_2$ that does not pass through $v_1$ and $v_2$ simultaneously, then there exist at least two distinct paths connecting $s_i$ to $B$ for some $i\in \{1,2\}$. Without loss of generality, assume that there exists at least two distinct paths connecting $s_1$ to $B$, namely the path $p_1$, connecting $s_1$ and $v_1$, and the path $p_2$, connecting $s_1$ and $v_3\in B$ for some $v_3\notin\{v_1,v_2\}$. Removing the edge $e_1=(v_1,w_1)$ would not disconnect $s_1$ from $B$ as $s_1$ is connected to $B$ via the path $p_2$, a contradiction since $e_1$ is a cut-edge.
    
    If $N$ is obtained from a simple network of type $L_2^{n,(B_1|B_2|B_3)}$ by the identification process mentioned earlier, then by \Cref{lemma:different A_i}, the two reticulation branches correspond to two distinct $B_i$'s. Hence $N|_{\{s_1,s_2\}}$ is a 2-leaf simple strict level-2 network. It implies that  $N|_{S}$ is again strict level-2 network because adding more element of $S$ to the restriction means adding more branches and leaves.
    \item The proof of (6a) is similar to (5). For (6b), the assumptions implies that in the restriction procedure, there will be parallel edges connecting the two purely interior vertices of $N$ and hence  they will be collapsed to a single edge. \qedhere
\end{enumerate}
\end{proof}

\begin{lemma}\label{lemma:restriction of level-1 network}
 Let $N$ be an $n$-leaf cycle network on $X$. Then for any non-empty subset $S\subseteq X$, $N|_S$ is a nice network.
\end{lemma}
\begin{proof}
For any non-empty subset $S\subseteq X=[n]$, the restriction $N|_S$ is a level-1 network by \Cref{lemma:level of restriction} and hence $N|_S$ is a nice network by \Cref{lemma:funnel-free and nice} (2).
\end{proof}
\begin{lemma}\label{lemma:restriction of network with two reticulation leaves}
 Let $N$ be an $n$-leaf simple nice strict level-2 network on $X$ where $|X|\geq 4$. Moreover, let us assume that $N$ has two reticulation leaves. Given two distinct leaves $a,b\in X$, there exist two distinct leaves $c,d\in X\setminus\{a,b\}$ such that $N|_{\{a,b,c,d\}}$ is a nice network. 
\end{lemma}
\begin{proof}
Let $u$ and $v$ be the two purely interior vertices of $N$.  For any leaf $x$, let $p^x_u$ be the undirected shortest path connecting the unique parent of $x$, $up(x)$ and the vertex $u$ that does not contain $v$. Similarly, we can define $p^x_v$. The proof will be divided into three cases. For illustrations in each case in the proof, see \Cref{table:illustration two reticulation leaves} in \Cref{sec:appendix}.
\begin{enumerate}
    \item Case 1: suppose that $r(N)=\{a,b\}$. Because $N$ is nice, $d(a,u)=d(b,u)=2$. Indeed, $d(a,u)=d(b,u)=2$ implies that $u$ or $v$ is a funnel vertex, respectively. Here $d(x,y)$ denotes the distance between vertex $x$ and $y$ in the network $N$. Therefore, there are three possibilities: 
    $$d(a,u)=2,d(b,u)>2; \qquad d(a,u)>2,d(b,u)=2; \qquad \mbox{ or }\qquad d(a,u)>2,d(b,u)>2.$$
    \begin{itemize}
        \item Subcase 1a: suppose that $d(a,u)=2$  and $d(b,u)>2$. Because $d(b,u)>2$, there exists $c\in X\setminus\{a,b\}$ such that $d(b,c)=3$ and $up(c)\in p^b_u.$ If $d(a,v)=2$, as $N$ is nice network, then $d(b,v)>2$. This implies that there exists $d\in X\setminus\{a,b,c\}$ such that $d(b,d)=3$ and $up(d)\in p^b_v.$ By By our construction, $N|_{\{a,b,c,d\}}$ is a simple nice strict level-2  network. If $d(a,v)>2$, then there exists $d'\in X\setminus\{a,b,c\}$ such that $d(a,d')=3$ and $up(d')\in p^a_v.$ Again, by our construction, $N|_{\{a,b,c,d'\}}$ is a simple nice strict level-2 network. The case for $d(a,u)>2$ and $d(b,u)=2$ can be treated similarly. 
        \item Subcase 1b: suppose that $d(a,u)>2$ and $d(b,u)>2$. Since $d(a,u)>2$, then there exists $c\in X\setminus\{a,b\}$ such that $d(a,c)=3$ and $up(c)\in p^a_u.$ Similarly, since $d(b,u)>2$, then there exists $d\in X\setminus\{a,b,c\}$ such that $d(b,d)=3$ and $up(d)\in p^b_u$. If $d(a,v)=2$, then $d(b,v)>2$. Then there exists $e\in X\setminus\{a,b,c,d\}$ such that $d(b,e)=3$ and $up(e)\in p^b_v.$ Hence $N|_{\{a,b,c,e\}}$ is a simple nice strict level-2 network. If $d(a,v)>2$, then there exists $e'\in X\setminus\{a,b,c,d\}$ such that $d(e',a)=3$ and $up(e')\in p^a_v$. Hence $N|_{\{a,b,c,e'\}}$ is a simple nice strict level-2 network.
    \end{itemize}
    \item Case 2: without loss of generality, we now suppose that $r(N)=\{a,c\}$ where $c\neq b.$ Suppose that $d(a,b)=3$ and moreover, without loss of generality, we also suppose that $up(b)\in p^a_u$. If $d(a,v)=2$, then $d(c,v)>2.$ Then there exists $d\in X\setminus\{a,b,c\}$ such that $d(c,d)=3$ and $up(d)\in p^c_v.$ Hence, $N|_{\{a,b,c,d\}}$ is a simple nice strict level-2 network. If $d(a,v)>2,$ then there exist $d'\in X\setminus\{a,b,c\}$ such that $d(a,d')=3$ and $up(d')\in p^a_v.$ Thus, $N|_{\{a,b,c,d'\}}$ is a simple nice strict level-2 network. The case for $d(b,c)=3$ can be treated similarly. 
    
    We are left with the case $d(a,b)>3$ and $d(b,c)>3.$ As in part (1), we will deal with three possibilities: 
     $$d(a,u)=2,d(c,u)>2; \qquad d(a,u)>2,d(c,u)=2; \qquad \mbox{ or }\qquad d(a,u)>2,d(c,u)>2.$$
     \begin{itemize}
         \item Subcase 2a: suppose that $d(a,u)=2$  and $d(c,u)>2$. Because $d(c,u)>2$, there exists $d\in X\setminus\{a,c\}$ such that $d(c,d)=3$ and $up(d)\in p^c_u$. If $d(a,v)=2$, then $d(c,v)>2$. Then there exists $e\in X\setminus\{a,c,d\}$ such that $d(c,e)=3$ and $up(e)\in p^c_v$. If $b\in\{d,e\}$, then $N|_{\{a,c,d,e\}}$ is a simple nice strict level-2 network. Now suppose that $b\notin\{d,e\}$. If $p^b_u\subseteq p^c_u$, then $N|_{\{a,b,c,e\}}$ is a simple nice strict level-2 network. If $p^b_v\subseteq p^c_v$, then $N|_{\{a,b,c,d\}}$ is a simple nice strict level-2 network. If $up(b)$ is neither contained in $p^c_u$ nor $p^c_v$, then $N|_{\{a,b,d,e\}}$ is a simple level-1 network, which is a nice network. 
         
         If $d(a,v)>2$, then there exists $e'\in X\setminus\{a,c,d\}$ such that $d(a,e')=3$ and $up(e')\in p^a_v$. If $b\in\{d,e'\}$, then $N|_{\{a,c,d,e'\}}$ is a simple nice strict level-2 network. Now suppose that $b\notin\{d,e'\}$. If $p^b_v\subseteq p^a_v$, then $N|_{\{a,b,c,d\}}$ is a simple nice strict level-2 network. If $p^b_u\subseteq p^c_u$, then $N|_{\{a,b,c,e'\}}$ is a simple nice strict level-2 network. If $up(b)$ is neither contained in  $p^a_v$ nor $p^c_u$, then $N|_{\{a,b,d,e'\}}$ is a simple level-1 network, which is a nice network. The case for $d(a,u)>2,d(c,u)=2$ can be treated similarly.
         \item Subcase 2b: suppose that $d(a,u)>2$ and $d(c,u)>2$. Then there exist two distinct elements $d,e\in X\setminus \{a,c\}$ such that $d(a,d)=3=d(c,e)$, $up(d)\in p^a_u$, and $up(e)\in p^c_u$. Now suppose that $d(a,v)=2.$ Then there exists $e\in X\setminus\{a,c,d,e\}$ such that $d(c,e')=3$ and $up(e')\in p^c_v$. If $b\in \{d,e'\}$, then  $N|_{\{a,c, d,e'\}}$ is a simple nice strict level-2 network. Now suppose that $b\notin \{d,e'\}$.  If $p^b_u\subseteq p^a_u$, then $N|_{\{a,b,c,e'\}}$ is a simple nice strict level-2 network. If $p^b_u\subseteq p^c_u$, then $N|_{\{a,b,d,e'\}}$ is a level-1 network.  If $p^b_v\subseteq p^c_v$, then $N|_{\{a,b,c,d\}}$ is a simple nice strict level-2 network. If $up(b)$ is not contained in the union of $p^a_u,p^c_u,$ and $p^c_v$, then $N|_{\{a,b,d,e'\}}$ is a level-1 network.
         
         If $d(a,v)>2$, then there exists $f\in X\setminus\{a,c,d,e\}$ such that $d(a,f)=3$ and $f\in p^a_v$. If $b\in \{d,e\}$, then $N|_{\{a,b,c,f\}}$ is a nice strict level-2 network. If $b=f$, then $N|_{\{a,b,c,e\}}$ is a nice strict level-2 network. Now suppose that $b\notin \{d,e,f\}$. If $p^b_u\subseteq p^a_u$ or $p^b_u\subseteq p^c_u$, then $N|_{\{a,b,c,f\}}$ is a nice strict level-2 network. If $p^b_v\subseteq p^a_v$ or $p^b_v\subseteq p^c_v$, then $N|_{\{a,b,c,d\}}$ is a nice strict level-2 network. If $up(b)$ is not contained in the union of $p^a_u,p^a_v,p^c_u$, and $p^c_v$, then $N|_{\{a,b,c,f\}}$ is a nice strict level-2 network.
     \end{itemize}
    \item Case 3: finally, we consider the case for both $a$ and $b$ are not reticulation vertices. We may assume $|X|\geq 5$ as the lemma statement is trivial for $|X|=4.$ Let $r(N)=\{c,d\}$. Since $|X|\geq 5$, there exists $e\in X\setminus\{a,b,c,d\}$ which is not a reticulation vertex. Then by \Cref{lemma:level of restriction} (3), $N|_{\{a,b,c,e\}}$ is a level-1 network, which is nice. \qedhere
\end{enumerate}
\end{proof}

\begin{lemma}\label{lemma:restriction of network with one reticulation leaf}
 Let $N$ be an $n$-leaf simple nice strict level-2 network on $X$ where $|X|\geq 4$. Moreover, let us assume that $N$ has only one reticulation leaf. Given a leaf $a\in X$, there exist three distinct leaves $b,c,d\in X\setminus\{a\}$ such that $N|_{\{a,b,c,d\}}$ is a nice network. 
\end{lemma}
\begin{proof}
Let $u$ and $v$ be the two purely interior vertices of $N$ and $r(N)$ denote the set of reticulation leaves of $N$. Without loss of generality, suppose $r(N)=\{u,b\}$ for some $b\in X.$ Let $e_1$ and $e_2$ be two reticulation edges directed to $u$. Moreover, suppose that $p_1$ and $p_2$ denote the undirected shortest path connecting $v$ and $u$ containing $e_1$ and $e_2$, respectively. Furthermore, for any leaf $x$, let $p^x_u$ be the undirected shortest path connecting the unique parent of $x$, $up(x)$ and the vertex $u$ that does not contain $v$. Similarly, we can define $p^x_v$.We divide the proof into two cases. For illustrations in each case in the proof, see \Cref{table:illustration one reticulation leaf} in \Cref{sec:appendix}.

\begin{enumerate}
    \item Case 1: the vertex $a$ is a reticulation leaf, i.e. $a=b$. Since $u\in r(N)$, $d(a,u)>2.$ Then there exists $c\in X\setminus\{a\}$ such that $d(a,c)=3$ and $p^c_u\subseteq p^a_u$.  If $d(a,v)=2$, then $d(u,v)>2$. Then there exist two distinct elements $d,e\in X\setminus\{a,c\}$ such that $d(d,u)=2=d(e,u)$, $up(d)\in p_1$ and $up(e)\in p_2$. Then $N|_{\{a,c,d,e\}}$ is a simple nice strict level-2 network. If $d(a,v)>2$, then there exists $c'\in X\setminus \{a,c\}$ such that $d(a,c')=3$ and $p^{c'}_v\subseteq p^a_v$. Then choose any element $d'\in X\setminus\{a,c,c'\}$ such that $up(d')\in p_1$ or $up(d')\in p_2.$ Otherwise, there are parallel edges connecting $u$ and $v$. Hence $N|_{\{a,c,c',d'\}}$ is a simple nice strict level-2 network.
    
    \item Case 2: the vertex $a$ is not a reticulation vertex, i.e. $b\neq a$. Since $u\in r(N)$, $d(b,u)>2.$ Then there exists $c\in X\setminus\{b\}$ such that $d(b,c)=3$ and $p^c_u\subseteq p^b_u.$ If $d(b,v)=2$, then $d(u,v)>2$. Then there exist two distinct elements $d,e\in X\setminus\{b,c\}$ such that $d(d,u)=2=d(e,u)$, $up(d)\in p_1$ and $up(e)\in p_2$. If $a\in \{c,d,e\}$, then $N|_{\{b,c,d,e\}}$ is a simple nice strict level-2 network. Suppose that $a\notin \{c,d,e\}$. If $p^a_u\subseteq p^b_u$, then $N|_{\{a,b,d,e\}}$ is a simple nice strict level-2 network. If  $p^a_u\subseteq p_1$ or $p^a_u\subseteq p_2$, then $N|_{\{b,c,d,e\}}$ is a simple nice strict level-2 network. 
    
    If $d(b,v)>2$, then there exists $c'\in X\setminus\{b,c\}$ such that $d(c',b)=3$ and $p^{c'}_v\subseteq p^b_v$.Then choose any element $d'\in X\setminus\{b,c,c'\}$ such that $up(d')\in p_1$ or $up(d')\in p_2.$ If $a\in \{c,c',d'\}$, then $N|_{\{b,c,c',d'\}}$ is a simple nice strict level-2 network. Now suppose that $a\notin \{c,c',d'\}$. If $p^a_u\subseteq p^b_u$, then $N|_{\{a,b,c',d'\}}$ is a simple nice strict level-2 network. If $p^a_v\subseteq p^b_v$, then $N|_{\{a,b,c,d'\}}$ is a simple nice strict level-2 network. If $p^a_u\subseteq p_1$ or $p^a_u\subseteq p_2$, then $N|_{\{a,b,c,c'\}}$ is a simple nice level-2 network. \qedhere
\end{enumerate}
\end{proof}

The network restriction procedure enables us to use the following powerful lemma, which tells us that the information of being able to distinguish smaller subnetworks on some subsets of the leaf set can be used to distinguish the original networks.

\begin{lemma}[\cite{gross2018distinguishing}, Proposition 4.3]\label{lemma:variety of restriction}
 Let $N$ and $M$ be two distinct $n$-leaf networks on $X$ and $S\subseteq X$. If $V_{N|_S}\not\subseteq V_{M|_S}$, then $V_{N}\not\subseteq V_{M}.$
\end{lemma}

To derive results for level-2 networks with at least five leaves, our strategy is that we will restrict our observation to some smaller subnetworks. Using \Cref{lemma:containment level-2 4 leaves}, we can derive results for these subnetworks. Then we will use \Cref{lemma:variety of restriction} to derive results on the original networks with at least five leaves.  Another main tool to obtain the results is some earlier results stating that two of large cycle networks are distinguishable, which are provided in \cite{gross2018distinguishing} under the JC model and in \cite{hollering2020identifiability} under the K2P and K3P models. Additionally, the result for general triangle-free level-1 networks is provided in \cite{gross2020distinguishing}. Results provided in this paper are obtained under the JC, K2P or K3P model unless specifically stated otherwise. We will start with the following proposition, which provides some results on the varieties associated with simple nice strict level-2 and level-1 networks with at least five leaves.

\begin{proposition}\label{theorem:comparison-level-2-and-level-1}
For $n\geq 5,$ let $N_1$ be an $n$-leaf simple nice strict level-2 network.
\begin{enumerate}  
    \item If $N_2$ is an $n$-leaf cycle network, then $V_{N_1}\not\subseteq V_{N_2}.$
    \item  If $N_1$ has two reticulation leaves and $N_2$ is an $n$-leaf $n$-cycle network such that the unique reticulation leaf of $N_2$ is not a reticulation leaf in $N_1$, then they are distinguishable.
\end{enumerate}
  
\end{proposition}
\begin{proof}
Suppose that $N_1$ is a simple level-2 network of type $L_2^{n,(A_1|A_2|A_3)}$ and the vertices $u$ and $v$ be the two purely interior vertices of $N_1$.
\begin{enumerate}
    \item  By \Cref{prop:A_icontains reticulation leaf}, there exists $1\leq i\leq 3$ such that $A_i$ contains a reticulation leaf. For any $S\subseteq X$ of size four, if $N_2$ is a tree, then $N_2|_S$ is a 4-leaf tree. If $N_2$ a strict level-1 network, then \Cref{lemma:level of restriction} (2) implies that $N_2|_S$ is a 4-leaf level-1 network. 
 
We firstly suppose that $N_1$ contains two reticulation leaves. Then by \Cref{lemma:different A_i}, there exists an index $j\neq i$ such that $A_j$ also contains a reticulation leaf. Let $a\in A_i$ and $b\in A_j$ be the two reticulation leaves of $N_1$. By the proof of \Cref{lemma:restriction of network with two reticulation leaves} and \Cref{lemma:level of restriction} (5), there exist two distinct elements $c$ and $d$ such that $N_1|_S$ is a 4-leaf simple nice strict level-2 network where $S=\{a,b,c,d\}$. Now suppose that $N_1$ contains only one reticulation leaf. Then $A_j$ does not contain any reticulation leaves for any $j\neq i.$  Let $a\in A_i$ be the only reticulation leaf in $N_1$. By the proof of \Cref{lemma:restriction of network with one reticulation leaf} and \Cref{lemma:level of restriction} (6), there exist three distinct leaves $b,c,\mbox{ and }d$, where at least one of $b,c\mbox{ or }d$ belongs to some $A_j$ for $j\neq i$, such that $N_1|_S$ is a 4-leaf simple nice strict level-2 network where $S=\{a,b,c,d\}$. In both cases, by \Cref{lemma:containment level-2 4 leaves} (i), we have that $V_{N_1|_S}\not\subseteq V_{N_2|_S}$ and \Cref{lemma:variety of restriction} implies that $V_{N_1}\not\subseteq V_{N_2}.$

\item By the first part of this lemma, we have that $V_{N_1}\not\subseteq V_{N_2}.$ Thus, we only need to show that $V_{N_2}\not\subseteq V_{N_1}.$

Let $a$ be the unique reticulation leaf of $N_2.$ and $\{b,c\}$ be the reticulation leaves of $N_1$. By \Cref{lemma:different A_i}, $b$ and $c$ belong to different $A_i$'s. From the lemma assumptions, $a\notin\{b,c\}.$ We now choose any two distinct elements $d,e\in X\setminus\{a,b,c\}$. If $S=\{a,c,d,e\}$, then by \Cref{lemma:level of restriction} (2) and (3), $N_1|_S$ and $N_2|_S$ are level-1 networks. In particular, $N_1|_S$ is either a 4-leaf 4-cycle, single-triangle network or a tree. Indeed, let $p_u$ be the undirected path connecting $up(b)$ to $u$ that does not contain $v$. Similarly, let $p_v$ be the undirected path connecting $up(b)$ to $v$ that does not contain $u$. Let $T=\{a,d,e\}$ and $S'=\{a',d',e'\}$, where $a'=up(a),d'=up(d),$ and $e'=up(e)$. If $\forall s\in S',s\in p_u$ or $\forall s\in S',s\in p_v$, then $N_1|_S$ is a tree. Now let us fix $U\subseteq S'$ of size two. If ($\forall x\in U, x\in p_u$ and $S'\setminus U'\not\subseteq p_u$) or ($\forall x\in U, x\in p_v$ and $S'\setminus U'\not\subseteq p_v$), then $N_1|_S$ is a single-triangle network. The remaining cases imply that $N_1|_S$ is a 4-leaf 4-cycle network.

Additionally, $N_2|_S$ is a 4-leaf 4-cycle network. Indeed, by the proof of \Cref{lemma:level of restriction} (2), $N_2|_S$ can not be a tree as $S$ contain the reticulation vertex $a$. The restriction $N_2|_S$ can not be a double-triangle network because $N_2$ only contains one reticulation vertex and $N_2|_S$ can not be a single-triangle network because $N_2$ is simple. If both restricted networks are 4-cycle networks, then they have distinct semi-directed network topologies as $c$ is the reticulation vertex in $N_1$ but $a$ is the reticulation vertex in $N_2$. In any case of the restriction of $N_1,$ \cite[Lemma 1]{gross2020distinguishing} implies that $V_{N_2}\not\subseteq V_{N_1}.$ \qedhere
\end{enumerate} \end{proof}

 Let $N_1$ be an $n$-leaf simple nice strict level-2 network and $N_2$ be an $n$-leaf $n$-cycle network with $n\geq 5$. In the second assertion of \Cref{theorem:comparison-level-2-and-level-1}, we assume that $N_1$ has two reticulation leaves and the unique reticulation leaf of $N_2$ is not a reticulation leaf in $N_1$. \Cref{example:one of two assumptions is violated} in \Cref{sec:appendix} presents two pairs of networks such that one of the assumptions is violated. In these examples, we can not say anything about about their distinguishability as we only have one-sided non-containment.

 For the rest of this section, let us recall that the number of reticulation leaves of a strict simple nice level-2 network is either one or two. We now introduce the following definitions. To avoid confusion, $d_N(x,y)$ denotes the distance between vertex $x$ and $y$ in the network $N$.
 
 \begin{definition}\label{definition:N1-RC and N2-RC}
 Let $N_1$ and $N_2$ be two simple nice level-2 networks on $X$.
 \begin{enumerate}
     \item A pair $(N_1,N_2)$ is said to be \textit{$N_1$-reticulation-conflicting ($N_1$-RC)} if for every $x\in r(N_1)$, there exists $y\in r(N_2)$ such that $d_{N_1}(x,y)=3$.
     \item Similarly, a pair $(N_1,N_2)$ is said to be \textit{$N_2$-reticulation-conflicting ($N_2$-RC)} if for every $x\in r(N_2)$, there exists $y\in r(N_1)$ such that $d_{N_2}(x,y)=3$.
 \end{enumerate}
 \end{definition}
 
 \begin{lemma}\label{lemma: not N1-RC}
  Let $N_1$ and $N_2$ be two distinct $n$-leaf simple nice strict level-2 networks on $X$. If $r(N_2)\subseteq r(N_1)$, Then the pair $(N_1,N_2)$ is not $N_1$-RC.
 \end{lemma}
 \begin{proof}
 Let $x\in r(N_1)$ be any element. If there exists $y\in r(N_2)=r(N_1)$ such that $d_{N_1}(x,y)=3$, then $N_1$ would not be a nice network. Thus, $(N_1,N_2)$ is not $N_1$-RC.
 \end{proof}
 \begin{corollary}
  Let $N_1$ and $N_2$ be two distinct $n$-leaf simple nice strict level-2 networks on $X$. If $r(N_1)=r(N_2)$, Then the pair $(N_1,N_2)$ is neither $N_1$-RC nor $N_2$-RC.
 \end{corollary}
 
 \begin{example}
 Let us consider the following six networks in \Cref{fig:RC}. The pair $(N_1,N_2)$ is both $N_1$-RC and $N_2$-RC. The pair $(N_3,N_4)$ is neither $N_1$-RC nor $N_2$-RC. The pair $(N_5,N_6)$ is $N_1$-RC but not $N_2$-RC.
   \begin{figure}[H]
    \centering
     \subfigure[$N_1$]{
    \begin{tikzpicture}[scale=0.5]
       \fill[black] (0,0) circle (3pt) (1,1) circle (3pt) (3,1) circle (3pt) (4,0) circle (3pt) (1,-1) circle (3pt) (3,-1) circle (3pt) (1,2) circle (3pt)  (3,2) circle (3pt) (1,-2) circle (3pt) (3,-2) circle (3pt);
        \node (a1) at (1,2)[above]{a} ; 
       \node (a1) at (3,2)[above]{c} ; 
       \node (a1) at (1,-2)[below]{d} ; 
       \node (a1) at (3,-2)[below]{b} ; 
       \draw (0,0)--(2,0);
\draw (0,0)--(4,0);
\draw [dashed,->,>=Triangle](0,0)--(1,1);
\draw [dashed,<-,>=Triangle](1,1) --(3,1);
\draw (3,1) --(4,0);
\draw (0,0)--(1,-1);
\draw [dashed,->,>=Triangle](1,-1)--(3,-1);
\draw [dashed,<-,>=Triangle](3,-1)--(4,0);
\draw (1,1)--(1,2);
\draw (3,1)--(3,2);
\draw (1,-1)--(1,-2);
\draw (3,-1)--(3,-2);
    \end{tikzpicture}}
   \subfigure[$N_2$]{
    \begin{tikzpicture}[scale=0.5]
       \fill[black] (0,0) circle (3pt) (1,1) circle (3pt) (3,1) circle (3pt) (4,0) circle (3pt) (1,-1) circle (3pt) (3,-1) circle (3pt) (1,2) circle (3pt)  (3,2) circle (3pt) (1,-2) circle (3pt) (3,-2) circle (3pt);
        \node (a1) at (1,2)[above]{c} ; 
       \node (a1) at (3,2)[above]{a} ; 
       \node (a1) at (1,-2)[below]{b} ; 
       \node (a1) at (3,-2)[below]{d} ; 
       \draw (0,0)--(2,0);
\draw (0,0)--(4,0);
\draw [dashed,->,>=Triangle](0,0)--(1,1);
\draw [dashed,<-,>=Triangle](1,1) --(3,1);
\draw (3,1) --(4,0);
\draw (0,0)--(1,-1);
\draw [dashed,->,>=Triangle](1,-1)--(3,-1);
\draw [dashed,<-,>=Triangle](3,-1)--(4,0);
\draw (1,1)--(1,2);
\draw (3,1)--(3,2);
\draw (1,-1)--(1,-2);
\draw (3,-1)--(3,-2);
    \end{tikzpicture}}
    \subfigure[$N_3$]{
    \begin{tikzpicture}[scale=0.55]
       \fill[black] (0,0) circle (3pt) (1,1) circle (3pt) (2,1) circle (3pt)  (3,1) circle (3pt) (4,0) circle (3pt) (1,-1) circle (3pt) (3,-1) circle (3pt) (1,2) circle (3pt) (2,2) circle (3pt)  (3,2) circle (3pt) (3,-2) circle (3pt) (1,-2) circle (3pt);
        \node (a1) at (1,2)[above]{c} ; 
       \node (a1) at (2,2)[above]{a} ; 
       \node (a1) at (3,2)[above]{d} ;
       \node (a1) at (1,-2)[below]{e} ; 
       \node (a1) at (3,-2)[below]{b} ; 
\draw (0,0)--(4,0);
\draw (0,0)--(1,1);
\draw [dashed,->,>=Triangle](1,1) --(2,1);
\draw [dashed,<-,>=Triangle](2,1) --(3,1);
\draw (3,1)--(4,0);
\draw (0,0)--(1,-1);
\draw [dashed,->,>=Triangle](1,-1)--(3,-1);
\draw [dashed,<-,>=Triangle](3,-1)--(4,0);
\draw (1,1)--(1,2);
\draw (2,1)--(2,2);
\draw (3,1)--(3,2);
\draw (1,-1)--(1,-2);
\draw (3,-1)--(3,-2);
    \end{tikzpicture}}
     \subfigure[$N_4$]{
    \begin{tikzpicture}[scale=0.55]
       \fill[black] (0,0) circle (3pt) (1,1) circle (3pt) (2,1) circle (3pt)  (3,1) circle (3pt) (4,0) circle (3pt) (1,-1) circle (3pt) (3,-1) circle (3pt) (1,2) circle (3pt) (2,2) circle (3pt)  (3,2) circle (3pt) (3,-2) circle (3pt) (1,-2) circle (3pt);
        \node (a1) at (1,2)[above]{b} ; 
       \node (a1) at (2,2)[above]{a} ; 
       \node (a1) at (3,2)[above]{d} ;
       \node (a1) at (1,-2)[below]{e} ; 
       \node (a1) at (3,-2)[below]{c} ; 
\draw (0,0)--(4,0);
\draw (0,0)--(1,1);
\draw [dashed,->,>=Triangle](1,1) --(2,1);
\draw [dashed,<-,>=Triangle](2,1) --(3,1);
\draw (3,1)--(4,0);
\draw (0,0)--(1,-1);
\draw [dashed,->,>=Triangle](1,-1)--(3,-1);
\draw [dashed,<-,>=Triangle](3,-1)--(4,0);
\draw (1,1)--(1,2);
\draw (2,1)--(2,2);
\draw (3,1)--(3,2);
\draw (1,-1)--(1,-2);
\draw (3,-1)--(3,-2);
    \end{tikzpicture}}
    \subfigure[$N_5$]{
    \begin{tikzpicture}[scale=0.55]
       \fill[black] (0,0) circle (3pt) (1,1) circle (3pt) (2,1) circle (3pt)  (3,1) circle (3pt) (4,0) circle (3pt) (1,-1) circle (3pt) (3,-1) circle (3pt) (1,2) circle (3pt) (2,2) circle (3pt)  (3,2) circle (3pt) (3,-2) circle (3pt) (1,-2) circle (3pt) (2,-1) circle (3pt) (2,-2) circle (3pt);
        \node (a1) at (1,2)[above]{a} ; 
       \node (a1) at (2,2)[above]{c} ; 
       \node (a1) at (3,2)[above]{e} ;
       \node (a1) at (1,-2)[below]{f} ; 
       \node (a1) at (2,-2)[below]{d} ;
       \node (a1) at (3,-2)[below]{b} ; 
\draw (0,0)--(4,0);
\draw (0,0)--(1,1);
\draw [dashed,->,>=Triangle](1,1) --(2,1);
\draw [dashed,<-,>=Triangle](2,1) --(3,1);
\draw (3,1)--(4,0);
\draw (0,0)--(1,-1);
\draw [dashed,->,>=Triangle](2,-1)--(3,-1);
\draw [dashed,<-,>=Triangle](3,-1)--(4,0);
\draw (1,1)--(1,2);
\draw (2,1)--(2,2);
\draw (3,1)--(3,2);
\draw (1,-1)--(1,-2);
\draw (3,-1)--(3,-2);
\draw (1,-1)--(2,-1);
\draw (2,-1)--(2,-2);
    \end{tikzpicture}}
    \subfigure[$N_5$]{
    \begin{tikzpicture}[scale=0.55]
       \fill[black] (0,0) circle (3pt) (1,1) circle (3pt) (2,1) circle (3pt)  (3,1) circle (3pt) (4,0) circle (3pt) (1,-1) circle (3pt) (3,-1) circle (3pt) (1,2) circle (3pt) (2,2) circle (3pt)  (3,2) circle (3pt) (3,-2) circle (3pt) (1,-2) circle (3pt) (2,-1) circle (3pt) (2,-2) circle (3pt);
        \node (a1) at (1,2)[above]{c} ; 
       \node (a1) at (2,2)[above]{e} ; 
       \node (a1) at (3,2)[above]{a} ;
       \node (a1) at (1,-2)[below]{b} ; 
       \node (a1) at (2,-2)[below]{f} ;
       \node (a1) at (3,-2)[below]{d} ; 
\draw (0,0)--(4,0);
\draw (0,0)--(1,1);
\draw [dashed,->,>=Triangle](1,1) --(2,1);
\draw [dashed,<-,>=Triangle](2,1) --(3,1);
\draw (3,1)--(4,0);
\draw (0,0)--(1,-1);
\draw [dashed,->,>=Triangle](2,-1)--(3,-1);
\draw [dashed,<-,>=Triangle](3,-1)--(4,0);
\draw (1,1)--(1,2);
\draw (2,1)--(2,2);
\draw (3,1)--(3,2);
\draw (1,-1)--(1,-2);
\draw (3,-1)--(3,-2);
\draw (1,-1)--(2,-1);
\draw (2,-1)--(2,-2);
    \end{tikzpicture}}
    \caption{Six simple nice level-2 networks $N_i$, $1\leq i\leq 6$.}
    \label{fig:RC}
\end{figure}
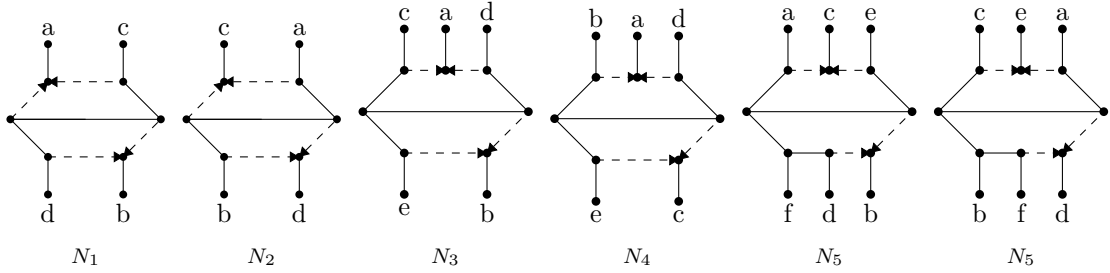
 \end{example}

 We will now provide some results on simple nice level-2 networks based on their set of reticulation leaves.

\begin{theorem}\label{proposition:distinguishability r(N_1)=r(N_2)=2}
Let $N_1$ and $N_2$ be two distinct $n$-leaf simple nice strict level-2 networks for $n\geq 5$ such that $|r(N_1)|=|r(N_2)|=2$ and $r(N_1)\neq r(N_2)$. If the pair $(N_1,N_2)$ is neither $N_1$-RC nor $N_2$-RC, then  $N_1$ and $N_2$ are distinguishable.
\end{theorem}
\begin{proof}
First suppose that $r(N_1)\cap r(N_2)=\emptyset.$ Suppose that $r(N_1)=\{a,b\}$ and $r(N_2)=\{c,d\}$ where $a,b\notin\{c,d\}, a\neq b\mbox{ and }c\neq d.$ By \Cref{lemma:restriction of network with two reticulation leaves} and \Cref{lemma:level of restriction} (5), there exist two distinct elements $p,q\in X\setminus\{a,b\}$ such that $N_1|_S$ is a 4-leaf simple nice strict level-2 network for $S=\{a,b,p,q\}.$ By the proof of the first case in \Cref{lemma:restriction of network with two reticulation leaves}, we may choose $p$ and $q$ such that in $N_1$, $p$ and $q$ have distance three to either $a$ or $b$. If both $p$ and $q$ are contained in $\{c,d\}$, then the pair $(N_1,N_2)$ is $N_1$-RC, a contradiction. Thus, at least one of $p$ or $q$ is not contained in $\{c,d\}.$ 
 \begin{itemize}
 \item Case 1: without loss of generality, we may assume that $p=c$ but $q\notin\{c,d\}.$ By \Cref{lemma:level of restriction} (3), $N_2|_S$ is a level-1 network. In this case, by \Cref{lemma:containment level-2 4 leaves} (i), we have that $V_{N_1|_S}\not\subseteq V_{N_2|_S}$.
 \item Case 2: $p,q\notin\{c,d\}$. By \Cref{lemma:level of restriction} (1), the restriction $N_2|_S$ is a tree. In this case, by \Cref{lemma:containment level-2 4 leaves} (i), we have that $V_{N_1|_S}\not\subseteq V_{N_2|_S}$.
 \end{itemize}
 Therefore, \Cref{lemma:variety of restriction} suggests that $V_{N_1}\not\subseteq V_{N_2}$. Conversely, we can use similar arguments and the assumption that the pair $(N_1,N_2)$ is not $N_2$-RC to obtain $V_{N_2}\not\subseteq V_{N_1}$.
 
 Next we suppose that $|r(N_1)\cap r(N_2)|=1.$ Suppose that $r(N_1)=\{a,b\}$ and $r(N_2)=\{a,c\}$ where $a\neq b,a\neq c\mbox{ and }b\neq c.$ By \Cref{lemma:restriction of network with two reticulation leaves} and \Cref{lemma:level of restriction} (5), there exist two distinct elements $p',q'\in X\setminus\{a,b\}$ such that $N_1|_S$ is a 4-leaf simple nice strict level-2 network for $S=\{a,b,p',q'\}.$ Now we have the following two cases. 
 \begin{itemize}
     \item Case 1': without loss of generality, $p'=c$ and $q'\neq c$. As $n\geq 5$, choose $e\in X\setminus\{a,b,p,q\}$. Then as $N_1|_{\{a,b,p,q,e\}}$ is a nice strict level-2 network, \Cref{lemma:level of restriction} (3) implies that $N_1|_{S'}$ is a level-1 network where $S'=\{b,p,q,e\}$. Similarly, $N_2|_{S'}$ is also a level-1 network. Moreover, the vertex $b$ is the reticulation leaf of $N_1|_{S'}$ while the vertex $c$ is the reticulation leaf of $N_2|_{S'}$. \cite[Theorem 1.1]{gross2018distinguishing} implies that $N_1$ and $N_2$ are distinguishable.
     \item Case 2': $p\neq c$ and $q\neq c$. In this case, $N_2|_S$ is a level-1 network. Therefore, \Cref{lemma:variety of restriction} suggests that $V_{N_1}\not\subseteq V_{N_2}$. Similarly, one can use similar arguments to show that $V_{N_2}\not\subseteq V_{N_1}$.
 \end{itemize}
 This completes the proof of the theorem.
 \end{proof}
 
 If the assumption that $r(N_1)\neq r(N_2)$ in \Cref{proposition:distinguishability r(N_1)=r(N_2)=2} is removed, then we may not be able to distinguish the networks $N_1$ and $N_2$. \Cref{example: r(N1)=r(N_2)=2} in \Cref{sec:appendix} suggests that for two networks with $|r(N_1)|=|r(N_2)|=2$ but $r(N_1)=r(N_2)$, we may not be able to say anything about their distinguishability.

\begin{theorem}\label{proposition:distinguishability r(N_1)>=r(N_2)}
Let $N_1$ and $N_2$ be two distinct $n$-leaf simple nice strict level-2 networks on $X$ for $n\geq 5$.
 Suppose that $r(N_1)=\{a,b\}$ and $r(N_2)=\{c\}$ such that $c\notin\{a,b\}$. Let $u$ be the purely interior reticulation vertex of $N_2$. If $d_{N_1}(x,y)\geq 4$ for any $x\in r(N_1)$ and $y\in r(N_2)$, then $V_{N_1}\not\subseteq V_{N_2}.$ Additionally, if the pair $(N_1,N_2)$ is not $N_2$-RC, $d_{N_2}(a,u)\geq 3,$ and $d_{N_2}(b,u)\geq 3$, then $N_1$ and $N_2$ are distinguishable.
\end{theorem}
 \begin{proof}
 By \Cref{lemma:restriction of network with two reticulation leaves} and \Cref{lemma:level of restriction} (5), there exists $p,q\in X\setminus\{a,b\}$ such that for $S=\{a,b,p,q\}$, $N_1|_S$ is a 4-leaf simple nice strict level-2 network. In the proof of the first case in \Cref{lemma:restriction of network with two reticulation leaves}, we may choose $p$ and $q$ such that in $N_1$, they have distance three to either $a$ or $b$. 
  By our hypothesis, $p\neq c$ and $q\neq c$. Indeed, if $p=c$, then $d_{N_1}(p,a)=3$ or $d_{N_1}(p,b)=3$, a contradiction. Similar argument hold for $q=c$. In this case, $N_2|_S$ is a level-1 network. Therefore, \Cref{lemma:variety of restriction} suggests that $V_{N_1}\not\subseteq V_{N_2}$.

  Additionally, let us assume that the pair $(N_1,N_2)$ is not $N_2$-RC. By \Cref{lemma:restriction of network with one reticulation leaf} and \Cref{lemma:level of restriction} (6), there exist $p',q',r'\in X\setminus\{c\}$ such that $N_2|_{S'}$ is a 4-leaf simple nice strict level-2 network for $S'=\{c,p',q',r'\}$. In the proof of the first case in \Cref{lemma:restriction of network with one reticulation leaf}, we may choose $p',q',$ and $r'$ such that in $N_2$, they have distance three to $c$ or distance two to $u$. If at least one of $p',q',$ or $r'$ is $a$ or $b$, then the pair is $N_2$-RC, $d_{N_2}(a,u)=2$, or $d_{N_2}(b,u)=2
  $, a contradiction. Thus, $a,b\notin \{p',q',r'\}$. Thus, in this case, $N_1|_{S'}$ is a level-1 network. Therefore, \Cref{lemma:variety of restriction} suggests that $V_{N_2}\not\subseteq V_{N_1}$. \qedhere
\end{proof}

For the case when $r(N_2)\subseteq r(N_1)$, the pair $(N_1,N_2)$ is clearly not $N_1$-RC by \Cref{lemma: not N1-RC}. Therefore, the most natural assumption to add in order to distinguish the pair is that the pair should not be $N_2$-RC as follows. \Cref{example: counterexample second statement} in \Cref{sec:appendix} suggests that a pair of networks which is $N_2$-RC. In this example, we can not say anything about their distinguishability.

\begin{proposition}\label{proposition:distinguishability r(N_1)>=r(N_2) (2)}
Let $N_1$ and $N_2$ be two distinct $n$-leaf simple nice strict level-2 networks on $X$ for $n\geq 5$. Suppose that $r(N_1)=\{a,b\}$ and $r(N_2)=\{a\}$. Let $u$ be the purely interior reticulation vertex of $N_2$. If the pair $(N_1,N_2)$ is not $N_2$-RC, $d_{N_2}(a,u)\geq 3,$ and $d_{N_2}(b,u)\geq 3$, then $V_{N_2}\not\subseteq V_{N_1}$.
\end{proposition}
\begin{proof}
The proof follows by applying similar arguments as in the second paragraph of the proof of \Cref{proposition:distinguishability r(N_1)>=r(N_2)}.
\end{proof}

\begin{theorem}\label{proposition:distinguishability r(N_1)=r(N_2)=1}
   Let $N_1$ ans $N_2$ be two distinct $n$-leaf simple nice strict level-2 networks on $X$ for $n\geq 5$. Suppose that $r(N_1)=\{a\}$ and $r(N_2)=\{b\}$ such that $a\neq b$. Let $u$ be the purely interior reticulation vertex in both networks. If the pair $(N_1,N_2)$ is neither $N_1$-RC nor $N_2$-RC, $d_{N_1}(b,u)\geq 3$, and  $d_{N_2}(a,u)\geq 3$, then $N_1$ and $N_2$ are distinguishable.
\end{theorem}
\begin{proof}
By \Cref{lemma:restriction of network with one reticulation leaf} and \Cref{lemma:level of restriction} (6), there exist three distinct elements $p,q,r\in X\setminus\{a\}$ such that $N_1|_S$ is a simple nice strict level-2 network for $S=\{a,p,q,r\}.$ In the proof of the first case in \Cref{lemma:restriction of network with one reticulation leaf}, we may choose $p',q',$ and $r'$ such that in $N_1$, they have distance three to $c$ or distance two to $u$. If $c\in \{p,q,r\}$, then the pair is $N_1$-RC or $d_{N_2}(c,u)=2$, a contradiction. Thus, $c\notin\{p,q,r\}.$ It implies that $N_2|_S$ is a level-1 network by \Cref{lemma:level of restriction} (4). Therefore, $V_{N_1}\not\subseteq V_{N_2}.$ To obtain the other non-containment, we can proceed using similar arguments.
\end{proof}

If the assumption that $r(N_1)\neq r(N_2)$ in \Cref{proposition:distinguishability r(N_1)=r(N_2)=1} is removed, then we may not be able to distinguish the networks $N_1$ and $N_2$. \Cref{example: counterexample first statement} in \Cref{sec:appendix} suggests that for two networks with $|r(N_1)|=|r(N_2)|=1$ but $r(N_1)=r(N_2)$, we may not be able to say anything about their distinguishability as we only have one-sided non-containment of network varieties.

\section{Distinguishing semisimple level-2 networks with at least five leaves}\label{section:Distinguishing semisimple level-2 networks with at least five leaves and beyond}

In this section, we will study level-2 semisimple networks and beyond. As semisimple networks could represent more complex evolutionary history compared to their simple counterparts, we would like generalize the results presented in \Cref{subsub:Distinguishing simple level-2 networks with at least five leaves} for the class of semisimple networks.

We will now begin by comparing the varieties of semisimple nice strict level-2 networks and of level-1 networks. The following theorem generalizes \Cref{theorem:comparison-level-2-and-level-1}.

\begin{theorem}\label{thm:semisimple level-2 and level-1}
  For $n\geq 4$, let $N_1$ be an $n$-leaf semisimple nice strict level-2 network and $N_2$ be an $n$-leaf level-1 cycle network.  Then $V_{N_1}\not\subseteq V_{N_2}.$  
\end{theorem}  
\begin{proof}
Suppose that the set of branch's leaves of $N_1$ is given by the partition
 $P=\{A_1,\dots,A_t\}$ of the leaf set $[n]$. By \Cref{lemma:small semisimple networks}, $t\geq 4$. If $n=4$, then the lemma follows from \Cref{lemma:small semisimple networks} and  \Cref{lemma:containment level-2 4 leaves}. Thus, we are left with $n\geq 5$.
   
 We first suppose that $N_2$ is a tree. For each $1\leq i\leq t$, choose an element $a_{i}\in A_{i}.$ Let $S$ be the union of all such $a_{i}$'s. Then by our construction, $N_1|_S$ is simple. Since $S$ contains an element from each reticulation branch of $N_1$, $N_1|_S$ is a strict level-2 network. The fact that $N_1|_S$ is a nice network is straightforward. Indeed, if the simple network $N_1|_S$ is not nice, then neither is $N_1$ as we can obtain $N_1$ from $N_1|_S$ by adding more leaves to each branch. Moreover, it is immediate that $N_2|_S$ is a tree. As $n\geq 5$, both $N_1|_S$ and $N_2|_S$ have at least five leaves. By \Cref{theorem:comparison-level-2-and-level-1}, $V_{N_1}\not\subseteq V_{N_2}.$
     
    Now we suppose that $N_2$ is a strict level-1 network and that the induced partition of $[n]$ in $N_2$ is given by 
 $Q=\{B_1,B_2,\dots, B_s\}$ for $s\geq 3.$ Indeed, if $s\leq 2$, then $N_2$ is a tree. Let us recall that by \Cref{lemma:level of restriction}, for any $S\subseteq X$ of size 4, $N_2|_S$ is a level-1 network. We will distinguish four cases. 
 \begin{enumerate}
     \item If $|P|<|Q|$, then there exist $i,j,k$ with $i\neq j$ and $a,b\in [n]$ such that $a,b\in A_k$ while $a\in B_i$ and $b\in B_j.$ As $n\geq 5$ and $t\geq 4$, the proof of \Cref{lemma:restriction of network with two reticulation leaves} and \Cref{lemma:restriction of network with one reticulation leaf} suggests that we can find three distinct elements $c,d,e\in [n]\setminus A_k$ such that $N_1|_S$ is a 4-leaf simple nice strict level-2 network for $S=\{a,c,d,e\}$. Then, $V_{N_1|_S}\not\subseteq V_{N_2|_S}$ by \Cref{lemma:containment level-2 4 leaves}.
     
     \item  If $|P|>|Q|$, then there exist $i,j,k$ with $i\neq j$ and $a,b\in [n]$ such that $a,b\in B_k$ while $a\in A_i$ and $b\in A_j.$ Then by the proof of \Cref{lemma:restriction of network with two reticulation leaves} and \Cref{lemma:restriction of network with one reticulation leaf}, there exist two distinct element $c,d\in [n]\setminus (A_i\cup A_j)$ belonging to two distinct branch's leaves of $N_1$ such that $N_1|_S$ is a 4-leaf simple nice strict level-2 network for $S=\{a,b,c,d\}$. Again, by \Cref{lemma:containment level-2 4 leaves}, $V_{N_1|_S}\not\subseteq V_{N_2|_S}.$
     
     \item  Next we assume that $|P|=|Q|=t$ but $P\neq Q$. Then, there exist $p,q,r\in [t]$  where $q\neq r$ and $a,b\in [n]$ such that $a,b\in B_p$ while $a\in A_q$ and $b\in A_r.$ Now the proof of \Cref{lemma:restriction of network with two reticulation leaves} and \Cref{lemma:restriction of network with one reticulation leaf} suggests that we can find two distinct leaves $c,d\in [n]\setminus\{a,b\}$ such that $N_1|_S$ is a 4-leaf simple nice strict level-2 network where $S=\{a,b,c,d\}$. For this set $S$, $N_2|_S$ is a 4-leaf 3-cycle network or tree. In this case, again $V_{N_1|_S}\not\subseteq V_{N_2|_S}.$ 
     
     \item Lastly, we are left with the case $P=Q$. If there exist $i,j\in [s]$ such that $A_i=B_j$ but $N_1|_{A_i}\neq N_1|_{B_j},$ then the theorem follows from the fact that trees are distinguishable. Thus, we may assume that if $A_i=B_j$, then $N_1|_{A_i}= N_1|_{B_j}.$ This will imply that the partition $B_1|\dots|B_t$ is a reordering of $A_1|\dots|A_t.$ Suppose that $A_1|\dots|A_t=B_{i_1}|\dots|B_{i_t}.$ Now we can view the network $N_1$ as a $t$-leaf simple nice strict level-2 network while the network $N_2$ as $t$-leaf $t$-cycle network. In $N_1$, we replace the tree corresponding to the partition $A_i$ by a leaf labeled by $i.$ In $N_2$, we replace the tree corresponding to the partition $B_{i_k}$ by a leaf labeled by $k.$ Then the non-containment $V_{N_1}\not\subseteq V_{N_2}$ follows from \Cref{lemma:containment level-2 4 leaves} for $t=4$ and from \Cref{theorem:comparison-level-2-and-level-1} for $t\geq 5$.\qedhere
 \end{enumerate}
 \end{proof}
 
 We may have seen earlier that level-2 networks could explain more complex evolutionary histories of group of species compared to level-1 networks as level-2 networks allow us to have two reticulation vertices in each biconnected component of the network. \Cref{thm:semisimple level-2 and level-1} suggests that it is reasonable to model the evolution of a group of species on a level-2 network as the associated algebraic variety would not be contained in the varieties associated with simpler networks, which in this case are the level-1 networks.
 
 In addition to \Cref{lemma:level of restriction}, we now present more properties of restriction of a semisimple network. These lemmas follow by adapting the proof of \Cref{lemma:restriction of network with two reticulation leaves} and \Cref{lemma:restriction of network with one reticulation leaf}.
 
 \begin{lemma}\label{lemma:restriction semisimple two reticulation branches}
  Let $N$ be an $n$-leaf semisimple nice strict level-2 network on $X$ where $|X|\geq 4$. Moreover, let us assume that $N$ has two reticulation branches. Given two distinct leaves $a,b\in X$ belonging to distinct branches, there exist two distinct leaves $c,d\in X\setminus\{a,b\}$ such that $N|_{\{a,b,c,d\}}$ is a nice network. 
 \end{lemma}

 \begin{lemma}\label{lemma:restriction semisimple one reticulation branch}
  Let $N$ be an $n$-leaf semisimple nice strict level-2 network on $X$ where $|X|\geq 4$. Moreover, let us assume that $N$ has one reticulation branch. Given $a\in X$, there exists three distinct leaves $b,c,d\in X\setminus\{a\}$ such that $N|_{\{a,b,c,d\}}$ is a nice network. 
 \end{lemma}
 
 In \Cref{lemma:restriction semisimple two reticulation branches}, the elements $a,b,c,$ and $d$ belong to pairwise distinct branches. Otherwise, $N|_{\{a,b,c,d\}}$ has fewer than four leaves and hence it is not nice by \Cref{lemma:small semisimple networks}. Thus, $N|_{\{a,b,c,d\}}$ is simple. Similar statement hold for the elements $a,b,c,$ and $d$ in \Cref{lemma:restriction semisimple one reticulation branch}.

  In what follows, let us recall that the number $R(N)$ denotes the set of reticulation branches of a semisimple network $N$. In particular, if $N$ is simple, then $R(N)$ is the set of its reticulation leaves. Next, given a pair $(N_1,N_2)$, we introduce the following definitions which generalize the definitions of $N_1$-RC and $N_2$-RC.
  
  \begin{definition}\label{definition: N1-RBC and N2-RBC}
  Let $N_1$ and $N_2$ be two semisimple nice networks strict level-2 networks on $X$.
  \begin{enumerate}
       \item A pair $(N_1,N_2)$ is said to be \textit{$N_1$-reticulation-branch-conflicting ($N_1$-RBC)} if for every $A\in R(N_1)$ and $a\in A$, there exist $B\in R(N_2)$ and $b\in B$ such that $d_{N_1|_{\{a,b\}\cup X\setminus (A\cup B)}}(x,y)=3$.
     \item Similarly, a pair $(N_1,N_2)$ is said to be \textit{$N_2$-reticulation-branch-conflicting ($N_2$-RBC)} if for every $A\in R(N_2)$ and $a\in A$, there exist $B\in R(N_1)$ and $b\in B$ such that $d_{N_2|_{\{a,b\}\cup X\setminus (A\cup B)}}(x,y)=3$.
  \end{enumerate}
  \end{definition}
  
The following lemma is an immediate consequence of \Cref{definition:N1-RC and N2-RC} and \Cref{definition: N1-RBC and N2-RBC}, together with the fact that every simple strict level-2 network is semisimple.
 \begin{lemma}
  Let $N_1$ and $N_2$ be two simple nice strict level-2 networks on $X$. Then for $i\in\{1,2\},$ the pair $(N_1,N_2)$ is $N_i$-RC if and only if it is $N_i$-RBC.
 \end{lemma}

 We will now provide some results on the distinguishability of semisimple nice level-2 networks based on their set of reticulation branches.

\begin{theorem}\label{proposition:distinguishability semisimple R(N_1)=R(N_2)=2}
Let $N_1$ and $N_2$ be two distinct $n$-leaf semisimple nice strict level-2 networks for $n\geq 5$ such that $|R(N_1)|=|R(N_2)|=2$. Let $\{A_1,\dots,A_{l_1}\}$ and $\{B_1,\dots,B_{l_2}\}$ be the set of branches of $N_1$ and $N_2$, respectively. Suppose that $R(N_1)=\{A_1,A_2\}$ and $R(N_2)=\{B_1,B_2\}$. If the pair $(N_1,N_2)$ is neither $N_1$-RBC nor $N_2$-RBC and $A_i\cap B_j=\emptyset$ for $1\leq i,j\leq 2$, then  $N_1$ and $N_2$ are distinguishable.
\end{theorem}
 \begin{proof}
 First, we can choose any two elements $a_1\in A_1$ and $a_2\in A_2$. Then by \Cref{lemma:restriction semisimple two reticulation branches} and \Cref{lemma:level of restriction} (5), there exist two distinct leaves $c,d\in X\setminus\{a_1,a_2\}$ such that $N_1|_S$ is a simple nice strict level-2 network for $S=\{a_1,a_2,c,d\}$. We may choose $c$ and $d$ such that they belong to different branches of $N_1$. Suppose that $c\in A_i$ and $d\in A_j$ where $A_i$ and $A_j$ are two distinct branches of $N_1$ and $i,j\notin\{1,2\}$. Let $B$ the unique strict level-2 biconnected component of $N_1$. Let us suppose that the branch $A_l$ is attached to $B$ via cut-edges $e_l=(v_l,w_l)$ such that $v_l\in B$. Finally, we may assume that in $N_1$, the vertex $v_i$ and $v_j$ have distance one to either $v_1$ or $v_2$.
 
 Suppose that both $c$ and $d$ is contained in $B_1\cup B_2$. Since $c$ and $d$ belong to distinct branches, without loss of generality, we may assume that $c\in B_1$ and $d\in B_2$. Then ($d_{N_1|_{\{a_1,c\}\cup X\setminus(A_1\cup A_i)}}(c,a_1)=3$ or  $d_{N_1|_{\{a_2,c\}\cup X\setminus(A_2\cup A_i)}}(c,a_2)=3$) and ($d_{N_1|_{\{a_1,d\}\cup X\setminus(A_1\cup A_j)}}(c,a_1)=3$ or  $d_{N_1|_{\{a_2,d\}\cup X\setminus(A_2\cup A_j)}}(c,a_2)=3$). This implies that the pair $(N_1,N_2)$ is $N_1$-RBC, a contradiction. Thus, at least one of $c$ or $d$ is not contained in $B_1\cup B_2.$ We have the following two cases.
 \begin{enumerate}
     \item Case 1: without loss of generality, we may assume that $c\in B_1$ but $d\notin B_1\cup B_2$. By our hypothesis, $a_1,a_2\notin B_1\cup B_2$. Therefore, \Cref{lemma:level of restriction} (3), $N_2|_S$ is a level-1 network. Hence $V_{N_1|_S}\not\subseteq V_{N_2|_S}.$
     \item Case 2: $c,d\notin B_1\cup B_2$. Again, \Cref{lemma:level of restriction} (1) suggests that $N_2|_S$ is a tree. Hence $V_{N_1|_S}\not\subseteq V_{N_2|_S}.$
 \end{enumerate}
 Therefore, in all cases, \Cref{lemma:variety of restriction} suggests that $V_{N_1}\not\subseteq V_{N_2}$. Conversely, we can apply similar arguments to show that $V_{N_2}\not\subseteq V_{N_1}$. This completes the proof of the theorem.
 \end{proof}

The assumption on the reticulation branches that $A_i\cap B_j=\emptyset$ mentioned in \Cref{proposition:distinguishability semisimple R(N_1)=R(N_2)=2} is essential for distinguishing two semisimple networks. \Cref{example:R(N1)neq  R(N2)} in \Cref{sec:appendix} suggests that removing this assumption may result in an inconclusive statement on the distinguishability of the two given networks.

\begin{theorem}\label{proposition:distinguishability semisimple R(N_1)>=R(N_2) (1)}
Let $N_1$ and $N_2$ be two distinct $n$-leaf semisimple nice strict level-2 networks on $X$ for $n\geq 5$. Let $\{A_1,\dots,A_{l_1}\}$ and $\{B_1,\dots,B_{l_2}\}$ be the set of branches of $N_1$ and $N_2$, respectively. Suppose that $R(N_1)=\{A_1,A_2\}$ and $R(N_2)=\{B_1\}$. Let $u$ be the purely interior reticulation vertex of $N_2$. If $d_{N_1|_{\{x,y\}\cup X\setminus(A_i\cup B_1)}}(x,y)\geq 4$ for any $x\in A_i$ and $y\in B_1$ and $A_1\cap B_1=\emptyset=A_2\cap B_1$, then $V_{N_1}\not\subseteq V_{N_2}$. Additionally, if $(N_1,N_2)$ is not $N_2$-RBC and $d_{N_2|_{\{x\}\cup X\setminus A_i}}(x,u)\geq 3$ for any $x\in A_i$, then $N_1$ and $N_2$ are distinguishable.
\end{theorem}
\begin{proof}
We choose any elements $a_1\in A_1$ and $a_2\in A_2$. By \Cref{lemma:restriction semisimple two reticulation branches} and \Cref{lemma:level of restriction} (5), there exists $c,d\in X\setminus\{a_1,a_2\}$ such that for $S=\{a_1,a_2,c,d\}$, $N_1|_S$ is a 4-leaf simple nice strict level-2 network.  As before, we may choose $c$ and $d$ such that they belong to different branches of $N_1$. Suppose that $c\in A_i$ and $d\in A_j$ where $A_i$ and $A_j$ are two distinct branches of $N_1$ and $i,j\notin\{1,2\}$. Let $B$ the unique strict level-2 biconnected component of $N_1$. Let us suppose that the branch $A_l$ is attached to $B$ via cut-edges $e_l=(v_l,w_l)$ such that $v_l\in B$. Finally, we may assume that in $N_1$, the vertex $v_i$ and $v_j$ have distance one to either $v_1$ or $v_2$. If $c\in  B_1$, then $d_{N_1|_{\{c,a_k\}\cup X\setminus(A_k\cup B_1)}}(c,a_k)=3$ for some $k\in \{1,2\}$, a contradiction. Similarly, we can show that it is impossible that $d\in B_1$.  Thus, $c,d\notin B_1$. By our hypothesis, $a_1,a_2\notin B_1$. In this case, $N_2|_S$ is a level-1 network and hence $V_{N_1}\not\subseteq V_{N_2}.$ 

Additionally, let us assume that $(N_1,N_2)$ is not $N_2$-RBC. Pick any $b_1\in B_1$. By \Cref{lemma:restriction semisimple one reticulation branch} and \Cref{lemma:level of restriction} (6), there exists $p',q',r'\in X\setminus\{b_1\}$ such that $N_2|_{S'}$ is a simple nice strict level-2 network for $S'=\{b_1,p',q',r'\}.$ As before, we may choose $p',q',$ and $r'$ such that they belong to pairwise different branches of $N_2$. Suppose that $p'\in B_i,q'\in B_j,$ and $r'\in B_k$ such that $i,j,k\notin\{1,2\}$ and $i\neq j,j\neq k,k\neq i$. Let $B'$ the unique strict level-2 biconnected component of $N_2$. Let us suppose that the branch $B_l$ is attached to $B'$ via cut-edges $e'_l=(v'_l,w'_l)$ such that $v'_l\in B'$. Finally, we may assume that in $N_2$, the vertices $v'_i,v_j'$ and $v'_k$ have distance one to either $v'_1,v'_2,$ or $u$. If $p'$ is contained in $A_k$ for some $k\in \{1,2\}$, then the pair is $N_2$-RBC or $d_{N_2|_{\{p'\}\cup X\setminus A_k}}(p',u)= 2$, a contradiction. Similarly, we can argue that it is impossible that $q'\in A_k$ or $r'\in A_k$ for any $k\in \{1,2\}$. Thus, $N_1|_{S'}$ is a level-1 network and hence  $V_{N_2}\not\subseteq V_{N_1}.$
\end{proof}

\begin{proposition}\label{proposition:distinguishability semisimple R(N_1)>=R(N_2) (2)}
Let $N_1$ and $N_2$ be two distinct $n$-leaf semisimple nice strict level-2 networks on $X$ for $n\geq 5$. Let $\{A_1,\dots,A_{l_1}\}$ and $\{B_1,\dots,B_{l_2}\}$ be the set of branches of $N_1$ and $N_2$, respectively. Suppose that and $R(N_1)=\{A_1,A_2\}$ and $R(N_2)=\{A_1\}$. Let $u$ be the purely interior reticulation vertex of $N_2$. If $(N_1,N_2)$ is not $N_2$-RBC and $d_{N_2|_{\{x\}\cup X\setminus A_i}}(x,u)\geq 3$ for any $x\in A_i$, then $V_{N_2}\not\subseteq V_{N_1}$.
\end{proposition}
\begin{proof}
The proof follows by applying similar arguments as in the second paragraph of the proof of \Cref{proposition:distinguishability semisimple R(N_1)>=R(N_2) (1)}.
\end{proof}

\begin{theorem}\label{proposition:distinguishability semisimple R(N_1)=R(N_2)=1}
Let $N_1$ and $N_2$ be two distinct $n$-leaf semisimple nice strict level-2 networks on $X$ for $n\geq 5$. Let $\{A_1,\dots,A_{l_1}\}$ and $\{B_1,\dots,B_{l_2}\}$ be the set of branches of $N_1$ and $N_2$, respectively. Suppose that and $R(N_1)=\{A_1\}$ and $R(N_2)=\{B_1\}$ such that $A_1\cap B_1=\emptyset$. Let $u$ be the purely interior reticulation vertex of both networks. If $(N_1,N_2)$ is neither $N_1$-RBC nor $N_2$-RBC, $d_{N_1|_{\{x\}\cup X\setminus B_1}}(x,u)\geq 3$ for any $x\in B_1$, and  $d_{N_2|_{\{x\}\cup X\setminus A_1}}(x,u)\geq 3$ for any $x\in A_1$, then $N_1$ and $N_2$ are distinguishable.
\end{theorem}
\begin{proof}
Pick any element $a_1\in A_1$. By \Cref{lemma:restriction semisimple one reticulation branch} and \Cref{lemma:level of restriction} (6), there exists $p,q,r\in X\setminus\{b_1\}$ such that $N_1|_{S}$ is a simple nice strict level-2 network for $S=\{a_1,p,q,r\}.$ As before, we may choose $p,q,$ and $r$ such that they belong to pairwise different branches of $N_1$. Suppose that $p\in A_i,q\in A_j,$ and $r\in A_k$ such that $i,j,k\notin\{1,2\}$ and $i\neq j,j\neq k,k\neq i$. Let $B$ the unique strict level-2 biconnected component of $N_1$. Let us suppose that the branch $A_l$ is attached to $B$ via cut-edges $e_l=(v_l,w_l)$ such that $v_l\in B$. Finally, we may assume that in $N_1$, the vertices $v_i,v_j,$ and $v_k$ have distance one to either $v_1,v_2,$ or $u$. If $p$ is contained in $B_1$, then the pair is $N_1$-RBC or $d_{N_1|_{\{p\}\cup X\setminus B_1}}(p,u)=2$, a contradiction. Similarly, we can argue that it is impossible that $q\in B_1$ or $r\in B_1$. Thus, $p,q,r\notin B_1$. Thus, $N_2|_{S}$ is a level-1 network and hence  $V_{N_1}\not\subseteq V_{N_2}.$ Conversely, we can proceed using similar arguments to show $V_{N_2}\not\subseteq V_{N_1}$.
\end{proof}

\section{Discussion}

Firstly, in \Cref{subsub:Distinguishing simple level-2 networks with at least five leaves}, namely in \Cref{proposition:distinguishability r(N_1)=r(N_2)=2}, \Cref{proposition:distinguishability r(N_1)>=r(N_2)}, and \Cref{proposition:distinguishability r(N_1)=r(N_2)=1}, we have shown that under some assumptions on the set of reticulation leaves, two simple nice strict level-2 networks with at least five leaves are distinguishable under the Jukes-Cantor, Kimura 2-parameter, and 3-parameter model constraints. Secondly, in \Cref{section:Distinguishing semisimple level-2 networks with at least five leaves and beyond} we also have proved that under some more general assumptions, two semisimple nice strict level-2 networks with at least five leaves are distinguishable as well in \Cref{proposition:distinguishability semisimple R(N_1)=R(N_2)=2}, \Cref{proposition:distinguishability semisimple R(N_1)>=R(N_2) (1)}, and \Cref{proposition:distinguishability semisimple R(N_1)=R(N_2)=1}. These distinguishability results show that under some assumptions on the set of species created by reticulation events, the biological data, which in our case is given by DNA sequences, provide ample information that can be used to recover the evolutionary histories of a group of species on a network. Furthermore, we have also shown in \Cref{thm:semisimple level-2 and level-1} that the network varieties associated with $n$-leaf semisimple nice strict level-2 networks are not contained in the network varieties associated with $n$-leaf level-1 networks.

The proof of our main results combines the algebraic methods provided by the use of Fourier transformations and the combinatorial methods using the network restriction procedures to be able to distinguish specific subnetworks of the two given networks. We employed a random search strategy described in \cite{hollering2020identifiability} and  \cite{gross2020distinguishing} which allowed us to search for some subsets of variables in the Fourier coordinates to find the necessary phylogenetic invariants to establish our results. In \Cref{lemma:containment level-2 4 leaves}, we could not provide distinguishability results on the varieties associated with two 4-leaf simple networks belonging to the same type. However, the first part of the proposition provides a one-sided non-containment between varieties associated with network of type 4 and of type 1, 2, or 3 under the Jukes-Cantor model. For instance, one would see this limitation in \Cref{example: counterexample first statement} and \Cref{example: counterexample second statement} in \Cref{sec:appendix}. Therefore, something similar to the random search strategy needs to be employed in order to find phylogenetic invariants that can distinguish more nice strict level-2 networks.

The main open problem is to derive non-containment result on the varieties associated with any two nice strict level-2 networks under the Jukes-Cantor, Kimura 2-parameter, and 3-parameter model. We have seen earlier that some of our results rely on some specific assumptions on the set of reticulation branches. It would be another point of interest to see whether our results remain true if we drop this assumption. Additionally, it would be desirable to study the generic identifiability for higher-level networks as higher-level networks can be used to model more complex evolutionary relationships between species. In particular, one might attempt to answer the following question, which will be the first step towards studying the identifiability of higher-lever networks. 
\begin{question}
Let $N_1$ be an $n$-leaf strict level-$k$ network and $N_2$ be an $n$-leaf strict level-$l$ network for $k<l.$ Is it true that $V_{N_2}\not\subseteq V_{N_1}$? 
\end{question}
\noindent
In this paper, we have shown that if we restrict our attention to the class of semisimple networks, then we have a positive answer for $l=2$. Once we know that the varieties associated with a strict level-$l$ network is not contained in the varieties associated with a strict level-$k$ network with $k<l$, the next prominent class of network to consider is the class of strict level-$l$ networks consisting of networks that can be obtained from simple level-$l$ networks by attaching a tree to each leaf edge. This network can be thought as a generalization of semisimple strict level-2 networks to higher-level networks. From the only unrooted level-2 generator $L_2$, we have four distinct simple nice strict level-2 networks with four leaves that have to be considered for the four leaves case. In general, \cite[Proposition 2.4]{gambette2009structure} states that the number of level-$l$ generators is at least $2^{l-1}$. Therefore, we expect a combinatorial explosion on the number of distinct type of networks that has to be taken into consideration.  

In addition to the growth of the complexity of the problem, we also expect solving this problem to be challenging as the number of parameters used in the Fourier parameterization increases. Consequently, this increase will make the computations more difficult and time-consuming. As we have seen in \Cref{subsection:nice phylogenetic networks}, nice networks enable us to use the powerful Fourier transform in order to obtain the phylogenetic invariants associated with a network model. It would be exciting to see whether any other algebraic approaches can be applied to study non-nice networks in order to recover its network topology.

\newpage
\appendix
\section{Appendix}\label{sec:appendix}
\begin{longtable}{|p{5cm}|c|}
     \hline
     Case 1a: $d(a,u)=2$ and $d(b,u)>2$ &
     \begin{tikzpicture}[scale=0.4]
       \fill[black] (0,0) circle (3pt) (2,1) circle (3pt) (2,2) circle (3pt) (4,0) circle (3pt) (1,-1) circle (3pt) (1,-2) circle (3pt) (2,-1) circle (3pt) (2,-2) circle (3pt) (3,-1) circle (3pt) (3,-2) circle (3pt);
        \node (a1) at (2,2)[above]{$a$} ; 
       \node (a1) at (1,-2)[below]{$c$} ; 
       \node (a1) at (2,-2)[below]{$b$} ;
       \node (a1) at (3,-2)[below]{$d$} ;
       \node (a1) at (0,0)[left]{$u$};
       \node (a1) at (4,0)[right]{$v$};
\draw (0,0)--(4,0);
\draw [dashed,->,>=Triangle](0,0)--(2,1);
\draw [dashed,<-,>=Triangle](2,1) --(4,0);
\draw (0,0)--(1,-1);
\draw [dashed,->,>=Triangle](1,-1)--(2,-1);
\draw [dashed,<-,>=Triangle](2,-1)--(3,-1);
\draw (2,1)--(2,2);
\draw (1,-1)--(1,-2);
\draw (2,-1)--(2,-2);
\draw (3,-1)--(3,-2);
\draw (3,-1)--(4,0);
    \end{tikzpicture}
    \qquad
    \begin{tikzpicture}[scale=0.4]
       \fill[black] (0,0) circle (3pt) (1,1) circle (3pt) (3,1) circle (3pt) (4,0) circle (3pt) (1,-1) circle (3pt) (3,-1) circle (3pt) (1,2) circle (3pt)  (3,2) circle (3pt) (1,-2) circle (3pt) (3,-2) circle (3pt);
        \node (a1) at (1,2)[above]{$a$} ; 
       \node (a1) at (3,2)[above]{$d'$} ; 
       \node (a1) at (1,-2)[below]{$c$} ;
       \node (a1) at (3,-2)[below]{$b$} ;
       \node (a1) at (0,0)[left]{$u$};
       \node (a1) at (4,0)[right]{$v$};
\draw (0,0)--(4,0);
\draw [dashed,->,>=Triangle](0,0)--(1,1);
\draw [dashed,<-,>=Triangle](1,1) --(3,1);
\draw (3,1) --(4,0);
\draw (0,0)--(1,-1);
\draw [dashed,->,>=Triangle](1,-1)--(3,-1);
\draw [dashed,<-,>=Triangle](3,-1)--(4,0);
\draw (1,1)--(1,2);
\draw (3,1)--(3,2);
\draw (1,-1)--(1,-2);
\draw (3,-1)--(3,-2);
    \end{tikzpicture}
    \\
    \hline
    Case 1b: $d(a,u)>2$ and $d(b,u)>2$ &
    \begin{tikzpicture}[scale=0.4]
       \fill[black] (0,0) circle (3pt) (1,1) circle (3pt) (3,1) circle (3pt) (4,0) circle (3pt) (1,-1) circle (3pt) (3,-1) circle (3pt) (1,2) circle (3pt)  (3,2) circle (3pt) (1,-2) circle (3pt) (3,-2) circle (3pt) (2,-1) circle (3pt) (2,-2) circle (3pt) ;
        \node (a1) at (1,2)[above]{$c$} ; 
       \node (a1) at (3,2)[above]{$a$} ; 
       \node (a1) at (1,-2)[below]{$d$} ;
       \node (a1) at (3,-2)[below]{$b$} ;
       \node (a1) at (2,-2)[below]{$e$} ;
       \node (a1) at (0,0)[left]{$u$};
       \node (a1) at (4,0)[right]{$v$};
\draw (0,0)--(4,0);
\draw (0,0)--(1,1);
\draw [dashed,->,>=Triangle](1,1) --(3,1);
\draw [dashed,<-,>=Triangle](3,1) --(4,0);
\draw [dashed,->,>=Triangle](0,0)--(1,-1);
\draw [dashed,<-,>=Triangle](1,-1)--(2,-1);
\draw (2,-1)--(2,-2);
\draw (2,-1)--(3,-1);
\draw (3,-1)--(4,0);
\draw (1,1)--(1,2);
\draw (3,1)--(3,2);
\draw (1,-1)--(1,-2);
\draw (3,-1)--(3,-2);
    \end{tikzpicture}
    \qquad
    \begin{tikzpicture}[scale=0.4]
       \fill[black] (0,0) circle (3pt) (1,1) circle (3pt) (3,1) circle (3pt) (4,0) circle (3pt) (1,-1) circle (3pt) (3,-1) circle (3pt) (1,2) circle (3pt)  (3,2) circle (3pt) (1,-2) circle (3pt) (3,-2) circle (3pt) (2,1) circle (3pt) (2,2) circle (3pt) ;
        \node (a1) at (1,2)[above]{$c$} ; 
       \node (a1) at (3,2)[above]{$e'$} ; 
       \node (a1) at (1,-2)[below]{$d$} ;
       \node (a1) at (3,-2)[below]{$b$} ;
       \node (a1) at (2,2)[above]{$a$} ;
       \node (a1) at (0,0)[left]{$u$};
       \node (a1) at (4,0)[right]{$v$};
\draw (0,0)--(4,0);
\draw (0,0)--(1,1);
\draw [dashed,->,>=Triangle](1,1) --(2,1);
\draw [dashed,<-,>=Triangle](2,1) --(3,1);
\draw (0,0)--(1,-1);
\draw [dashed,->,>=Triangle](1,-1)--(3,-1);
\draw (2,1)--(2,2);
\draw (3,1)--(4,0);
\draw [dashed,<-,>=Triangle](3,-1)--(4,0);
\draw (1,1)--(1,2);
\draw (3,1)--(3,2);
\draw (1,-1)--(1,-2);
\draw (3,-1)--(3,-2);
    \end{tikzpicture}\\
    \hline
    Case 2: $d(a,b)=3$&  \begin{tikzpicture}[scale=0.4]
       \fill[black] (0,0) circle (3pt) (1,1) circle (3pt) (3,1) circle (3pt) (4,0) circle (3pt) (1,-1) circle (3pt) (3,-1) circle (3pt) (1,2) circle (3pt)  (3,2) circle (3pt) (1,-2) circle (3pt) (3,-2) circle (3pt);
        \node (a1) at (1,2)[above]{$b$} ; 
       \node (a1) at (3,2)[above]{$a$} ; 
       \node (a1) at (1,-2)[below]{$c$} ;
       \node (a1) at (3,-2)[below]{$d$} ;
       \node (a1) at (0,0)[left]{$u$};
       \node (a1) at (4,0)[right]{$v$};
\draw (0,0)--(4,0);
\draw [dashed,->,>=Triangle](0,0)--(1,1);
\draw [dashed,<-,>=Triangle](1,1) --(3,1);
\draw (3,1) --(4,0);
\draw (0,0)--(1,-1);
\draw [dashed,->,>=Triangle](1,-1)--(3,-1);
\draw [dashed,<-,>=Triangle](3,-1)--(4,0);
\draw (1,1)--(1,2);
\draw (3,1)--(3,2);
\draw (1,-1)--(1,-2);
\draw (3,-1)--(3,-2);
    \end{tikzpicture}
    \qquad
     \begin{tikzpicture}[scale=0.4]
       \fill[black] (0,0) circle (3pt) (1,1) circle (3pt) (3,1) circle (3pt) (4,0) circle (3pt) (2,-1) circle (3pt) (1,2) circle (3pt)  (3,2) circle (3pt) (2,-2) circle (3pt)  (2,1) circle (3pt) (2,2) circle (3pt) ;
        \node (a1) at (1,2)[above]{$b$} ; 
       \node (a1) at (3,2)[above]{$d'$} ; 
       \node (a1) at (2,-2)[below]{$c$} ;
       \node (a1) at (2,2)[above]{$a$} ;
       \node (a1) at (0,0)[left]{$u$};
       \node (a1) at (4,0)[right]{$v$};
\draw (0,0)--(4,0);
\draw (0,0)--(1,1);
\draw [dashed,->,>=Triangle](1,1) --(2,1);
\draw [dashed,<-,>=Triangle](2,1) --(3,1);
\draw [dashed,->,>=Triangle](0,0)--(2,-1);
\draw (2,1)--(2,2);
\draw (3,1)--(4,0);
\draw [dashed,<-,>=Triangle](2,-1)--(4,0);
\draw (1,1)--(1,2);
\draw (3,1)--(3,2);
\draw (2,-1)--(2,-2);
    \end{tikzpicture}\\
    \hline
    Case 2a: $d(a,b)>3, d(b,c)>3,$ and $d(a,v)=2$ & 
    \begin{tikzpicture}[scale=0.4]
       \fill[black] (0,0) circle (3pt) (2,1) circle (3pt) (2,2) circle (3pt) (4,0) circle (3pt) (1,-1) circle (3pt) (1,-2) circle (3pt) (2,-1) circle (3pt) (2,-2) circle (3pt) (3,-1) circle (3pt) (3,-2) circle (3pt) (0.5,-0.5) circle (3pt) (0.5,-2) circle (3pt);
        \node (a1) at (2,2)[above]{$a$} ; 
       \node (a1) at (1,-2)[below]{$d$} ; 
       \node (a1) at (2,-2)[below]{$c$} ;
       \node (a1) at (3,-2)[below]{$e$} ;
       \node (a1) at (0.5,-2)[below]{$b$} ;
       \node (a1) at (0,0)[left]{$u$};
       \node (a1) at (4,0)[right]{$v$};
\draw (0,0)--(4,0);
\draw [dashed,->,>=Triangle](0,0)--(2,1);
\draw [dashed,<-,>=Triangle](2,1) --(4,0);
\draw (0,0)--(1,-1);
\draw [dashed,->,>=Triangle](1,-1)--(2,-1);
\draw [dashed,<-,>=Triangle](2,-1)--(3,-1);
\draw (2,1)--(2,2);
\draw (1,-1)--(1,-2);
\draw (2,-1)--(2,-2);
\draw (3,-1)--(3,-2);
\draw (3,-1)--(4,0);
\draw (0.5,-0.5)--(0.5,-2);
    \end{tikzpicture}
    \qquad
    \begin{tikzpicture}[scale=0.4]
       \fill[black] (0,0) circle (3pt) (2,1) circle (3pt) (2,2) circle (3pt) (4,0) circle (3pt) (1,-1) circle (3pt) (1,-2) circle (3pt) (2,-1) circle (3pt) (2,-2) circle (3pt) (3,-1) circle (3pt) (3,-2) circle (3pt) (3.5,-0.5) circle (3pt) (3.5,-2) circle (3pt);
        \node (a1) at (2,2)[above]{$a$} ; 
       \node (a1) at (1,-2)[below]{$d$} ; 
       \node (a1) at (2,-2)[below]{$c$} ;
       \node (a1) at (3,-2)[below]{$e$} ;
       \node (a1) at (3.5,-2)[below]{$b$} ;
       \node (a1) at (0,0)[left]{$u$};
       \node (a1) at (4,0)[right]{$v$};
\draw (0,0)--(4,0);
\draw [dashed,->,>=Triangle](0,0)--(2,1);
\draw [dashed,<-,>=Triangle](2,1) --(4,0);
\draw (0,0)--(1,-1);
\draw [dashed,->,>=Triangle](1,-1)--(2,-1);
\draw [dashed,<-,>=Triangle](2,-1)--(3,-1);
\draw (2,1)--(2,2);
\draw (1,-1)--(1,-2);
\draw (2,-1)--(2,-2);
\draw (3,-1)--(3,-2);
\draw (3,-1)--(4,0);
\draw (3.5,-0.5)--(3.5,-2);
    \end{tikzpicture}
    \qquad
    \begin{tikzpicture}[scale=0.4]
       \fill[black] (0,0) circle (3pt) (2,2) circle (3pt) (2,3) circle (3pt) (4,0) circle (3pt) (1,-1) circle (3pt) (1,-2) circle (3pt) (2,-1) circle (3pt) (2,-2) circle (3pt) (3,-1) circle (3pt) (3,-2) circle (3pt) (2,0) circle (3pt) (2,1) circle (3pt);
        \node (a1) at (2,3)[above]{$a$} ; 
       \node (a1) at (1,-2)[below]{$d$} ; 
       \node (a1) at (2,-2)[below]{$c$} ;
       \node (a1) at (3,-2)[below]{$e$} ;
        \node (a1) at (2,1)[right]{$b$} ;
       \node (a1) at (0,0)[left]{$u$};
       \node (a1) at (4,0)[right]{$v$};
\draw (0,0)--(4,0);
\draw [dashed,->,>=Triangle](0,0)--(2,2);
\draw [dashed,<-,>=Triangle](2,2) --(4,0);
\draw (0,0)--(1,-1);
\draw [dashed,->,>=Triangle](1,-1)--(2,-1);
\draw [dashed,<-,>=Triangle](2,-1)--(3,-1);
\draw (2,2)--(2,3);
\draw (1,-1)--(1,-2);
\draw (2,-1)--(2,-2);
\draw (3,-1)--(3,-2);
\draw (3,-1)--(4,0);
\draw (2,1)--(2,0);
    \end{tikzpicture}\\
    \hline
    Case 2a: $d(a,b)>3, d(b,c)>3,$ and $d(a,v)>2$ &
     \begin{tikzpicture}[scale=0.4]
       \fill[black] (0,0) circle (3pt) (1,1) circle (3pt) (3,1) circle (3pt) (4,0) circle (3pt) (1,-1) circle (3pt) (3,-1) circle (3pt) (1,2) circle (3pt)  (3,2) circle (3pt) (1,-2) circle (3pt) (3,-2) circle (3pt) (2,1) circle (3pt) (2,2) circle (3pt) ;
        \node (a1) at (1,2)[above]{$a$} ; 
       \node (a1) at (3,2)[above]{$b$} ; 
       \node (a1) at (1,-2)[below]{$d$} ;
       \node (a1) at (3,-2)[below]{$c$} ;
       \node (a1) at (2,2)[above]{$e'$} ;
       \node (a1) at (0,0)[left]{$u$};
       \node (a1) at (4,0)[right]{$v$};
\draw (0,0)--(4,0);
\draw [dashed,->,>=Triangle](0,0)--(1,1);
\draw [dashed,<-,>=Triangle](1,1) --(2,1);
\draw (2,1) --(3,1);
\draw (0,0)--(1,-1);
\draw [dashed,->,>=Triangle](1,-1)--(3,-1);
\draw (2,1)--(2,2);
\draw (3,1)--(4,0);
\draw [dashed,<-,>=Triangle](3,-1)--(4,0);
\draw (1,1)--(1,2);
\draw (3,1)--(3,2);
\draw (1,-1)--(1,-2);
\draw (3,-1)--(3,-2);
    \end{tikzpicture}
    \qquad
     \begin{tikzpicture}[scale=0.4]
       \fill[black] (0,0) circle (3pt) (1,1) circle (3pt) (3,1) circle (3pt) (4,0) circle (3pt) (1,-1) circle (3pt) (3,-1) circle (3pt) (1,2) circle (3pt)  (3,2) circle (3pt) (1,-2) circle (3pt) (3,-2) circle (3pt) (2,-1) circle (3pt) (2,-2) circle (3pt) ;
        \node (a1) at (1,2)[above]{$a$} ; 
       \node (a1) at (3,2)[above]{$e'$} ; 
       \node (a1) at (1,-2)[below]{$b$} ;
       \node (a1) at (3,-2)[below]{$c$} ;
       \node (a1) at (2,-2)[below]{$d$} ;
       \node (a1) at (0,0)[left]{$u$};
       \node (a1) at (4,0)[right]{$v$};
\draw (0,0)--(4,0);
\draw [dashed,->,>=Triangle](0,0)--(1,1);
\draw [dashed,<-,>=Triangle](1,1) --(3,1);
\draw (3,1) --(4,0);
\draw (0,0)--(1,-1);
\draw (1,-1)--(2,-1);
\draw (2,-1)--(2,-2);
\draw [dashed,->,>=Triangle](2,-1)--(3,-1);
\draw [dashed,<-,>=Triangle](3,-1)--(4,0);
\draw (1,1)--(1,2);
\draw (3,1)--(3,2);
\draw (1,-1)--(1,-2);
\draw (3,-1)--(3,-2);
    \end{tikzpicture}
    \qquad 
    \begin{tikzpicture}[scale=0.4]
       \fill[black] (0,0) circle (3pt) (1,2) circle (3pt) (3,2) circle (3pt) (4,0) circle (3pt) (1,-1) circle (3pt) (3,-1) circle (3pt) (1,3) circle (3pt)  (3,3) circle (3pt) (1,-2) circle (3pt) (3,-2) circle (3pt) (2,0) circle (3pt) (2,1) circle (3pt);
        \node (a1) at (1,3)[above]{$a$} ; 
       \node (a1) at (3,3)[above]{$e'$} ; 
       \node (a1) at (1,-2)[below]{$d$} ;
       \node (a1) at (3,-2)[below]{$c$} ;
       \node (a1) at (2,1)[right]{$b$} ;
       \node (a1) at (0,0)[left]{$u$};
       \node (a1) at (4,0)[right]{$v$};
\draw (0,0)--(4,0);
\draw [dashed,->,>=Triangle](0,0)--(1,2);
\draw [dashed,<-,>=Triangle](1,2) --(3,2);
\draw (3,2) --(4,0);
\draw (0,0)--(1,-1);
\draw [dashed,->,>=Triangle](1,-1)--(3,-1);
\draw [dashed,<-,>=Triangle](3,-1)--(4,0);
\draw (1,2)--(1,3);
\draw (3,2)--(3,3);
\draw (1,-1)--(1,-2);
\draw (3,-1)--(3,-2);
\draw (2,0)--(2,1);
    \end{tikzpicture}\\
    \hline
    Case 2b: $d(a,u)>2,d(c,u)>2,$ and $d(a,v)=2$ &
    \begin{tikzpicture}[scale=0.4]
       \fill[black] (0,0) circle (3pt) (1,1) circle (3pt) (3,1) circle (3pt) (4,0) circle (3pt) (1,-1) circle (3pt) (3,-1) circle (3pt) (1,2) circle (3pt)  (3,2) circle (3pt) (1,-2) circle (3pt) (3,-2) circle (3pt) (2,1) circle (3pt) (2,2) circle (3pt) (2,-1) circle (3pt) (2,-2) circle (3pt) ;
        \node (a1) at (1,2)[above]{$b$} ; 
       \node (a1) at (3,2)[above]{$a$} ; 
       \node (a1) at (1,-2)[below]{$e$} ;
       \node (a1) at (3,-2)[below]{$e'$} ;
       \node (a1) at (2,-2)[below]{$c$} ;
       \node (a1) at (2,2)[above]{$d$} ;
       \node (a1) at (0,0)[left]{$u$};
       \node (a1) at (4,0)[right]{$v$};
\draw (0,0)--(4,0);
\draw (0,0)--(1,1);
\draw (1,1) --(2,1);
\draw [dashed,->,>=Triangle](2,1) --(3,1);
\draw (0,0)--(1,-1);
\draw [dashed,->,>=Triangle](1,-1)--(2,-1);
\draw (2,1)--(2,2);
\draw [dashed,<-,>=Triangle](3,1)--(4,0);
\draw [dashed,<-,>=Triangle](2,-1)--(3,-1);
\draw (1,1)--(1,2);
\draw (3,1)--(3,2);
\draw (1,-1)--(1,-2);
\draw (3,-1)--(3,-2);
\draw (3,-1)--(4,0);
\draw (2,-1)--(2,-2);
    \end{tikzpicture}
    \qquad
    \begin{tikzpicture}[scale=0.4]
       \fill[black] (0,0) circle (3pt) (1,1) circle (3pt) (1,2) circle (3pt) (4,0) circle (3pt) (1,-1) circle (3pt) (1,-2) circle (3pt) (2,-1) circle (3pt) (2,-2) circle (3pt) (3,-1) circle (3pt) (3,-2) circle (3pt) (0.5,-0.5) circle (3pt) (0.5,-2) circle (3pt) (3,1) circle (3pt) (3,2) circle (3pt);
        \node (a1) at (1,2)[above]{$d$} ;
        \node (a1) at (3,2)[above]{$a$} ;
       \node (a1) at (1,-2)[below]{$e$} ; 
       \node (a1) at (2,-2)[below]{$c$} ;
       \node (a1) at (3,-2)[below]{$e'$} ;
       \node (a1) at (0.5,-2)[below]{$b$} ;
       \node (a1) at (0,0)[left]{$u$};
       \node (a1) at (4,0)[right]{$v$};
\draw (0,0)--(4,0);
\draw (0,0)--(1,1);
\draw [dashed,->,>=Triangle](1,1)--(3,1);
\draw [dashed,<-,>=Triangle](3,1) --(4,0);
\draw (0,0)--(1,-1);
\draw [dashed,->,>=Triangle](1,-1)--(2,-1);
\draw [dashed,<-,>=Triangle](2,-1)--(3,-1);
\draw (1,1)--(1,2);
\draw (3,1)--(3,2);
\draw (1,-1)--(1,-2);
\draw (2,-1)--(2,-2);
\draw (3,-1)--(3,-2);
\draw (3,-1)--(4,0);
\draw (0.5,-0.5)--(0.5,-2);
    \end{tikzpicture}
    \qquad
    \begin{tikzpicture}[scale=0.4]
       \fill[black] (0,0) circle (3pt) (1,2) circle (3pt) (1,3) circle (3pt) (4,0) circle (3pt) (1,-1) circle (3pt) (1,-2) circle (3pt) (2,-1) circle (3pt) (2,-2) circle (3pt) (3,-1) circle (3pt) (3,-2) circle (3pt) (3.5,-0.5) circle (3pt) (3.5,-2) circle (3pt) (3,2) circle (3pt) (3,3) circle (3pt);
        \node (a1) at (1,3)[above]{$d$} ;
        \node (a1) at (3,3)[above]{$a$} ;
       \node (a1) at (1,-2)[below]{$e$} ; 
       \node (a1) at (2,-2)[below]{$c$} ;
       \node (a1) at (3,-2)[below]{$e'$} ;
       \node (a1) at (3.5,-2)[below]{$b$} ;
       \node (a1) at (0,0)[left]{$u$};
       \node (a1) at (4,0)[right]{$v$};
\draw (0,0)--(4,0);
\draw (0,0)--(1,2);
\draw [dashed,->,>=Triangle](1,2)--(3,2);
\draw [dashed,<-,>=Triangle](3,2) --(4,0);
\draw (0,0)--(1,-1);
\draw [dashed,->,>=Triangle](1,-1)--(2,-1);
\draw [dashed,<-,>=Triangle](2,-1)--(3,-1);
\draw (1,2)--(1,3);
\draw (3,2)--(3,3);
\draw (1,-1)--(1,-2);
\draw (2,-1)--(2,-2);
\draw (3,-1)--(3,-2);
\draw (3,-1)--(4,0);
\draw (3.5,-0.5)--(3.5,-2);
    \end{tikzpicture}
    \\
    \hline 
     Case 2b: $d(a,u)>2,d(c,u)>2,$ and $d(a,v)=2$ &
   \begin{tikzpicture}[scale=0.4]
       \fill[black] (0,0) circle (3pt) (1,2) circle (3pt) (1,3) circle (3pt) (4,0) circle (3pt) (1,-1) circle (3pt) (1,-2) circle (3pt) (2,-1) circle (3pt) (2,-2) circle (3pt) (3,-1) circle (3pt) (3,-2) circle (3pt) (3,2) circle (3pt) (3,3) circle (3pt) (2,0) circle (3pt) (2,1) circle (3pt);
        \node (a1) at (1,3)[above]{$d$} ;
        \node (a1) at (3,3)[above]{$a$} ;
       \node (a1) at (1,-2)[below]{$e$} ; 
       \node (a1) at (2,-2)[below]{$c$} ;
       \node (a1) at (3,-2)[below]{$e'$} ;
       \node (a1) at (2,1)[right]{$b$} ;
       \node (a1) at (0,0)[left]{$u$};
       \node (a1) at (4,0)[right]{$v$};
\draw (0,0)--(4,0);
\draw (0,0)--(1,2);
\draw [dashed,->,>=Triangle](1,2)--(3,2);
\draw [dashed,<-,>=Triangle](3,2) --(4,0);
\draw (0,0)--(1,-1);
\draw [dashed,->,>=Triangle](1,-1)--(2,-1);
\draw [dashed,<-,>=Triangle](2,-1)--(3,-1);
\draw (1,2)--(1,3);
\draw (3,2)--(3,3);
\draw (1,-1)--(1,-2);
\draw (2,-1)--(2,-2);
\draw (3,-1)--(3,-2);
\draw (3,-1)--(4,0);
\draw (2,0)--(2,1);
    \end{tikzpicture}\\
    \hline 
    Case 2b: $d(a,u)>2,d(c,u)>2,$ and $d(a,v)>2$ &
    \begin{tikzpicture}[scale=0.4]
       \fill[black] (0,0) circle (3pt) (1,1) circle (3pt) (3,1) circle (3pt) (4,0) circle (3pt) (1,-1) circle (3pt) (3,-1) circle (3pt) (1,2) circle (3pt)  (3,2) circle (3pt) (1,-2) circle (3pt) (3,-2) circle (3pt) (2,1) circle (3pt) (2,2) circle (3pt) (0.5,0.5) circle (3pt) (0.5, 2) circle (3pt) ;
       \node (a1) at (0.5,2)[above]{$b$} ;
        \node (a1) at (1,2)[above]{$d$} ; 
       \node (a1) at (3,2)[above]{$f$} ; 
       \node (a1) at (1,-2)[below]{$e$} ;
       \node (a1) at (3,-2)[below]{$c$} ;
       \node (a1) at (2,2)[above]{$a$} ;
       \node (a1) at (0,0)[left]{$u$};
       \node (a1) at (4,0)[right]{$v$};
\draw (0,0)--(4,0);
\draw (0,0)--(1,1);
\draw [dashed,->,>=Triangle](1,1) --(2,1);
\draw [dashed,<-,>=Triangle](2,1) --(3,1);
\draw (0,0)--(1,-1);
\draw [dashed,->,>=Triangle](1,-1)--(3,-1);
\draw (2,1)--(2,2);
\draw (3,1)--(4,0);
\draw [dashed,<-,>=Triangle](3,-1)--(4,0);
\draw (1,1)--(1,2);
\draw (3,1)--(3,2);
\draw (1,-1)--(1,-2);
\draw (3,-1)--(3,-2);
\draw (0.5,0.5)--(0.5,2);
    \end{tikzpicture}
    \qquad
    \begin{tikzpicture}[scale=0.4]
       \fill[black] (0,0) circle (3pt) (1,1) circle (3pt) (3,1) circle (3pt) (4,0) circle (3pt) (1,-1) circle (3pt) (3,-1) circle (3pt) (1,2) circle (3pt)  (3,2) circle (3pt) (1,-2) circle (3pt) (3,-2) circle (3pt) (2,1) circle (3pt) (2,2) circle (3pt) (0.5,0.5) circle (3pt) (0.5, 2) circle (3pt) ;
       \node (a1) at (0.5,2)[above]{$d$} ;
        \node (a1) at (1,2)[above]{$a$} ; 
       \node (a1) at (3,2)[above]{$b$} ; 
       \node (a1) at (1,-2)[below]{$e$} ;
       \node (a1) at (3,-2)[below]{$c$} ;
       \node (a1) at (2,2)[above]{$f$} ;
       \node (a1) at (0,0)[left]{$u$};
       \node (a1) at (4,0)[right]{$v$};
\draw (0,0)--(4,0);
\draw (0,0)--(0.5,0.5);
\draw [dashed,->,>=Triangle](0.5,0.5)--(1,1);
\draw [dashed,-<,>=Triangle](1,1) --(2,1);
\draw (2,1) --(3,1);
\draw (0,0)--(1,-1);
\draw [dashed,->,>=Triangle](1,-1)--(3,-1);
\draw (2,1)--(2,2);
\draw (3,1)--(4,0);
\draw [dashed,<-,>=Triangle](3,-1)--(4,0);
\draw (1,1)--(1,2);
\draw (3,1)--(3,2);
\draw (1,-1)--(1,-2);
\draw (3,-1)--(3,-2);
\draw (0.5,0.5)--(0.5,2);
    \end{tikzpicture}
    \qquad
    \begin{tikzpicture}[scale=0.4]
       \fill[black] (0,0) circle (3pt) (1,1) circle (3pt) (3,1) circle (3pt) (4,0) circle (3pt) (1,-1) circle (3pt) (3,-1) circle (3pt) (1,2) circle (3pt)  (3,2) circle (3pt) (1,-2) circle (3pt) (3,-2) circle (3pt) (2,1) circle (3pt) (2,2) circle (3pt) (2,-1) circle (3pt) (2,-2) circle (3pt) ;
        \node (a1) at (1,2)[above]{$d$} ; 
       \node (a1) at (3,2)[above]{$f$} ; 
       \node (a1) at (1,-2)[below]{$b$} ;
       \node (a1) at (2,-2)[below]{$e$} ;
       \node (a1) at (3,-2)[below]{$c$} ;
       \node (a1) at (2,2)[above]{$a$} ;
       \node (a1) at (0,0)[left]{$u$};
       \node (a1) at (4,0)[right]{$v$};
\draw (0,0)--(4,0);
\draw (0,0)--(1,1);
\draw [dashed,->,>=Triangle](1,1) --(2,1);
\draw [dashed,<-,>=Triangle](2,1) --(3,1);
\draw (0,0)--(1,-1);
\draw (1,-1)--(2,-1);
\draw (2,-1)--(2,-2);
\draw [dashed,->,>=Triangle](2,-1)--(3,-1);
\draw (2,1)--(2,2);
\draw (3,1)--(4,0);
\draw [dashed,<-,>=Triangle](3,-1)--(4,0);
\draw (1,1)--(1,2);
\draw (3,1)--(3,2);
\draw (1,-1)--(1,-2);
\draw (3,-1)--(3,-2);
    \end{tikzpicture}\\
    &\begin{tikzpicture}[scale=0.4]
       \fill[black] (0,0) circle (3pt) (1,1) circle (3pt) (3,1) circle (3pt) (4,0) circle (3pt) (1,-1) circle (3pt) (3,-1) circle (3pt) (1,2) circle (3pt)  (3,2) circle (3pt) (1,-2) circle (3pt) (3,-2) circle (3pt) (2,1) circle (3pt) (2,2) circle (3pt) (2,-1) circle (3pt) (2,-2) circle (3pt) ;
        \node (a1) at (1,2)[above]{$d$} ; 
       \node (a1) at (3,2)[above]{$f$} ; 
       \node (a1) at (1,-2)[below]{$e$} ;
       \node (a1) at (2,-2)[below]{$c$} ;
       \node (a1) at (3,-2)[below]{$b$} ;
       \node (a1) at (2,2)[above]{$a$} ;
       \node (a1) at (0,0)[left]{$u$};
       \node (a1) at (4,0)[right]{$v$};
\draw (0,0)--(4,0);
\draw (0,0)--(1,1);
\draw [dashed,->,>=Triangle](1,1) --(2,1);
\draw [dashed,<-,>=Triangle](2,1) --(3,1);
\draw (0,0)--(1,-1);
\draw [dashed,->,>=Triangle](1,-1)--(2,-1);
\draw (2,-1)--(2,-2);
\draw [dashed,<-,>=Triangle](2,-1)--(3,-1);
\draw (2,1)--(2,2);
\draw (3,1)--(4,0);
\draw (3,-1)--(4,0);
\draw (1,1)--(1,2);
\draw (3,1)--(3,2);
\draw (1,-1)--(1,-2);
\draw (3,-1)--(3,-2);
    \end{tikzpicture}
    \qquad \begin{tikzpicture}[scale=0.4]
       \fill[black] (0,0) circle (3pt) (1,2) circle (3pt) (3,2) circle (3pt) (4,0) circle (3pt) (1,-1) circle (3pt) (3,-1) circle (3pt) (1,3) circle (3pt)  (3,3) circle (3pt) (1,-2) circle (3pt) (3,-2) circle (3pt) (2,0) circle (3pt) (2,1) circle (3pt) (2,2) circle (3pt) (2,3) circle (3pt);
        \node (a1) at (1,3)[above]{$d$} ; 
        \node (a1) at (2,3)[above]{$a$} ;
       \node (a1) at (3,3)[above]{$f$} ; 
       \node (a1) at (1,-2)[below]{$e$} ;
       \node (a1) at (3,-2)[below]{$c$} ;
       \node (a1) at (2,1)[right]{$b$} ;
       \node (a1) at (0,0)[left]{$u$};
       \node (a1) at (4,0)[right]{$v$};
\draw (0,0)--(4,0);
\draw (0,0)--(1,2);
\draw (2,2)--(2,3);
\draw [dashed,->,>=Triangle](1,2)--(2,2);
\draw [dashed,<-,>=Triangle](2,2) --(3,2);
\draw (3,2) --(4,0);
\draw (0,0)--(1,-1);
\draw [dashed,->,>=Triangle](1,-1)--(3,-1);
\draw [dashed,<-,>=Triangle](3,-1)--(4,0);
\draw (1,2)--(1,3);
\draw (3,2)--(3,3);
\draw (1,-1)--(1,-2);
\draw (3,-1)--(3,-2);
\draw (2,0)--(2,1);
    \end{tikzpicture}\\
    \hline
    \caption{Illustrations for the proof of \Cref{lemma:restriction of network with two reticulation leaves}.}
    \label{table:illustration two reticulation leaves}
     \end{longtable}

\begin{longtable}{|p{5cm}|c|}
     \hline
     Case 1: the vertex $a$ is a reticulation leaf &
     \begin{tikzpicture}[scale=0.4]
       \fill[black] (0,0) circle (3pt) (1,2) circle (3pt) (3,2) circle (3pt) (4,0) circle (3pt) (2,-1) circle (3pt)  (1,3) circle (3pt)  (3,3) circle (3pt) (2,-2) circle (3pt) (2,0) circle (3pt) (2,1) circle (3pt);
       \node (a1) at (1,3)[above]{$c$} ; 
        \node (a1) at (3,3)[above]{$a$} ;
       \node (a1) at (2,1)[right]{$d$} ; 
       \node (a1) at (2,-2)[below]{$e$} ;
       \node (a1) at (0,0)[left]{$u$};
       \node (a1) at (4,0)[right]{$v$};
\draw (0,0)--(1,2);
\draw [dashed,->,>=Triangle](1,2) --(3,2);
\draw [dashed,<-,>=Triangle](3,2) --(4,0);
\draw [dashed,<-,>=Triangle](0,0)--(2,-1);
\draw (2,-1)--(4,0);
\draw [dashed,<-,>=Triangle](0,0)--(2,0);
\draw (2,0)--(4,0);
\draw (1,2)--(1,3);
\draw (3,2)--(3,3);
\draw (2,0)--(2,1);
\draw (2,-1)--(2,-2);
    \end{tikzpicture}
    \qquad
     \begin{tikzpicture}[scale=0.5]
       \fill[black] (0,0) circle (3pt) (1,1) circle (3pt) (2,1) circle (3pt)  (3,1) circle (3pt) (4,0) circle (3pt) (2,-1) circle (3pt) (1,2) circle (3pt) (2,2) circle (3pt)  (3,2) circle (3pt) (2,-2) circle (3pt);
        \node (a1) at (1,2)[above]{$c$} ; 
        \node (a1) at (2,2)[above]{$a$} ;
       \node (a1) at (3,2)[above]{$c'$} ; 
       \node (a1) at (2,-2)[below]{$d$} ;
       \node (a1) at (0,0)[left]{$u$};
       \node (a1) at (4,0)[right]{$v$};
\draw [dashed,<-,>=Triangle](0,0)--(4,0);
\draw (0,0)--(1,1);
\draw [dashed,->,>=Triangle](1,1) --(2,1);
\draw [dashed,<-,>=Triangle](2,1) --(3,1);
\draw (3,1)--(4,0);
\draw [dashed,<-,>=Triangle](0,0)--(2,-1);
\draw (2,-1)--(4,0);
\draw (1,1)--(1,2);
\draw (2,1)--(2,2);
\draw (3,1)--(3,2);
\draw (2,-1)--(2,-2);
    \end{tikzpicture}\\
    \hline
    Case 2: the vertex $a$ is not a reticulation leaf &
    \begin{tikzpicture}[scale=0.4]
       \fill[black] (0,0) circle (3pt) (1,2) circle (3pt) (3,2) circle (3pt) (4,0) circle (3pt) (2,-1) circle (3pt)  (1,3) circle (3pt)  (3,3) circle (3pt) (2,-2) circle (3pt) (2,0) circle (3pt) (2,1) circle (3pt) (2,2) circle (3pt) (2,3) circle (3pt);
       \node (a1) at (1,3)[above]{$a$} ; 
       \node (a1) at (2,3)[above]{$c$} ; 
        \node (a1) at (3,3)[above]{$b$} ;
       \node (a1) at (2,1)[right]{$d$} ; 
       \node (a1) at (2,-2)[below]{$e$} ;
       \node (a1) at (0,0)[left]{$u$};
       \node (a1) at (4,0)[right]{$v$};
\draw (0,0)--(1,2);
\draw (1,2)--(2,2);
\draw [dashed,->,>=Triangle](2,2) --(3,2);
\draw [dashed,<-,>=Triangle](3,2) --(4,0);
\draw [dashed,<-,>=Triangle](0,0)--(2,-1);
\draw (2,-1)--(4,0);
\draw [dashed,<-,>=Triangle](0,0)--(2,0);
\draw (2,0)--(4,0);
\draw (1,2)--(1,3);
\draw (3,2)--(3,3);
\draw (2,0)--(2,1);
\draw (2,-1)--(2,-2);
\draw (2,2)--(2,3);
    \end{tikzpicture}
    \qquad
    \begin{tikzpicture}[scale=0.4]
       \fill[black] (0,0) circle (3pt) (1,2) circle (3pt) (3,2) circle (3pt) (4,0) circle (3pt) (2,-1) circle (3pt)  (1,3) circle (3pt)  (3,3) circle (3pt) (2,-2) circle (3pt) (1,0) circle (3pt) (1,1) circle (3pt) (3,0) circle (3pt) (3,1) circle (3pt);
       \node (a1) at (1,3)[above]{$c$} ; 
        \node (a1) at (3,3)[above]{$b$} ;
       \node (a1) at (1,1)[right]{$d$} ; 
       \node (a1) at (3,1)[left]{$a$} ;
       \node (a1) at (2,-2)[below]{$e$} ;
       \node (a1) at (0,0)[left]{$u$};
       \node (a1) at (4,0)[right]{$v$};
\draw (0,0)--(1,2);
\draw [dashed,->,>=Triangle](1,2) --(3,2);
\draw [dashed,<-,>=Triangle](3,2) --(4,0);
\draw [dashed,<-,>=Triangle](0,0)--(2,-1);
\draw (2,-1)--(4,0);
\draw [dashed,<-,>=Triangle](0,0)--(1,0);
\draw (1,0)--(3,0);
\draw (3,0)--(4,0);
\draw (2,0)--(4,0);
\draw (1,2)--(1,3);
\draw (3,2)--(3,3);
\draw (2,-1)--(2,-2);
\draw (1,0)--(1,1);
\draw (3,0)--(3,1);
    \end{tikzpicture}
    \\
    &  \begin{tikzpicture}[scale=0.4]
       \fill[black] (0,0) circle (3pt) (1,1) circle (3pt) (3,1) circle (3pt) (4,0) circle (3pt) (2,-1) circle (3pt) (1,2) circle (3pt)  (3,2) circle (3pt) (2,-2) circle (3pt) (2,1) circle (3pt) (2,2) circle (3pt) (0.5,0.5) circle (3pt) (0.5, 2) circle (3pt) ;
       \node (a1) at (0.5,2)[above]{$a$} ;
        \node (a1) at (1,2)[above]{$c$} ; 
       \node (a1) at (3,2)[above]{$c'$} ; 
       \node (a1) at (2,-2)[below]{$d'$} ;
       \node (a1) at (2,2)[above]{$b$} ;
       \node (a1) at (0,0)[left]{$u$};
       \node (a1) at (4,0)[right]{$v$};
\draw [dashed,<-,>=Triangle](0,0)--(4,0);
\draw (0,0)--(1,1);
\draw [dashed,->,>=Triangle](1,1) --(2,1);
\draw [dashed,<-,>=Triangle](2,1) --(3,1);
\draw (2,1)--(2,2);
\draw (3,1)--(4,0);
\draw [dashed,<-,>=Triangle](0,0)--(2,-1);
\draw (1,1)--(1,2);
\draw (3,1)--(3,2);
\draw (0.5,0.5)--(0.5,2);
\draw (2,-1)--(4,0);
\draw (2,-1)--(2,-2);
    \end{tikzpicture}
    \qquad
    \begin{tikzpicture}[scale=0.4]
       \fill[black] (0,0) circle (3pt) (1,1) circle (3pt) (3,1) circle (3pt) (4,0) circle (3pt) (2,-1) circle (3pt) (1,2) circle (3pt)  (3,2) circle (3pt) (2,-2) circle (3pt) (2,1) circle (3pt) (2,2) circle (3pt) (0.5,0.5) circle (3pt) (0.5, 2) circle (3pt) ;
       \node (a1) at (0.5,2)[above]{$c$} ;
        \node (a1) at (1,2)[above]{$b$} ; 
       \node (a1) at (3,2)[above]{$a$} ; 
       \node (a1) at (2,-2)[below]{$d'$} ;
       \node (a1) at (2,2)[above]{$c'$} ;
       \node (a1) at (0,0)[left]{$u$};
       \node (a1) at (4,0)[right]{$v$};
\draw [dashed,<-,>=Triangle](0,0)--(4,0);
\draw (0,0)--(0.5,0.5);
\draw [dashed,->,>=Triangle](0.5,0.5)--(1,1);
\draw [dashed,<-,>=Triangle](1,1) --(2,1);
\draw (2,1) --(3,1);
\draw (2,1)--(2,2);
\draw (3,1)--(4,0);
\draw [dashed,<-,>=Triangle](0,0)--(2,-1);
\draw (1,1)--(1,2);
\draw (3,1)--(3,2);
\draw (0.5,0.5)--(0.5,2);
\draw (2,-1)--(4,0);
\draw (2,-1)--(2,-2);
    \end{tikzpicture}
    \qquad
     \begin{tikzpicture}[scale=0.4]
       \fill[black] (0,0) circle (3pt) (1,2) circle (3pt) (3,2) circle (3pt) (4,0) circle (3pt) (2,-1) circle (3pt)  (1,3) circle (3pt)  (3,3) circle (3pt) (2,-2) circle (3pt) (2,0) circle (3pt) (2,1) circle (3pt) (2,2) circle (3pt) (2,3) circle (3pt);
       \node (a1) at (1,3)[above]{$c$} ; 
       \node (a1) at (2,3)[above]{$b$} ; 
        \node (a1) at (3,3)[above]{$c'$} ;
       \node (a1) at (2,1)[right]{$d'$} ; 
       \node (a1) at (2,-2)[below]{$a$} ;
       \node (a1) at (0,0)[left]{$u$};
       \node (a1) at (4,0)[right]{$v$};
\draw (0,0)--(1,2);
\draw [dashed,->,>=Triangle](1,2)--(2,2);
\draw [dashed,<-,>=Triangle](2,2) --(3,2);
\draw (3,2) --(4,0);
\draw [dashed,<-,>=Triangle](0,0)--(2,-1);
\draw (2,-1)--(4,0);
\draw [dashed,<-,>=Triangle](0,0)--(2,0);
\draw (2,0)--(4,0);
\draw (1,2)--(1,3);
\draw (3,2)--(3,3);
\draw (2,0)--(2,1);
\draw (2,-1)--(2,-2);
\draw (2,2)--(2,3);
    \end{tikzpicture}\\
    \hline
     \caption{Illustrations for the proof of \Cref{lemma:restriction of network with one reticulation leaf}.}
     \label{table:illustration one reticulation leaf}
     \end{longtable}

\begin{example}\label{example:one of two assumptions is violated}
Let $N_i, 1\leq i\leq 4$ be the four networks presented in \Cref{fig:example of lemma 6.8}. By the first statement of \Cref{theorem:comparison-level-2-and-level-1}, for each pair $(N_i,N_j)$, we have $V_{N_i}\not\subseteq V_{N_j}.$ In order to distinguish each pair, we must show $V_{N_j|_S}\not\subseteq V_{N_i|_S}$ for some subset $S\subseteq \{a,b,c,d,e\}$. Let us now consider a subset $S$ of size 4.
 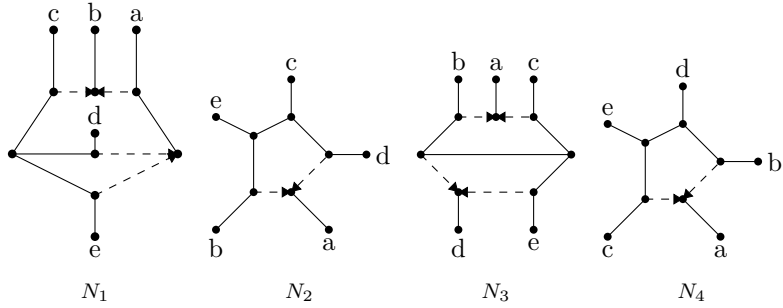
\begin{figure}[H]
    \centering
    \subfigure[$N_1$]{
    \begin{tikzpicture}[scale=0.55]
       \fill[black] (0,0) circle (3pt) (1,1.5) circle (3pt) (2,1.5) circle (3pt)  (3,1.5) circle (3pt) (4,0) circle (3pt) (2,-1) circle (3pt) (1,3) circle (3pt) (2,3) circle (3pt)  (3,3) circle (3pt) (2,-2) circle (3pt) (2,0) circle (3pt) (2,0.5) circle (3pt) ;
        \node (a1) at (1,3)[above]{c} ; 
       \node (a1) at (2,3)[above]{b} ; 
       \node (a1) at (3,3)[above]{a} ;
       \node (a1) at (2,0.5)[above]{d} ; 
       \node (a1) at (2,-2)[below]{e} ; 
       \draw (0,0)--(2,0);
\draw [dashed,->,>=Triangle](2,0)--(4,0);
\draw (0,0)--(1,1.5);
\draw [dashed,->,>=Triangle](1,1.5) --(2,1.5);
\draw [dashed,<-,>=Triangle](2,1.5) --(3,1.5);
\draw (3,1.5)--(4,0);
\draw (0,0)--(2,-1);
\draw [dashed,->,>=Triangle](2,-1)--(4,0);
\draw (1,1.5)--(1,3);
\draw (2,1.5)--(2,3);
\draw (3,1.5)--(3,3);
\draw (2,-1)--(2,-2);
\draw (2,0)--(2,0.5);
    \end{tikzpicture}}
    \subfigure[$N_2$]{
    \begin{tikzpicture}[scale=0.5]
       \fill[black] (0,0) circle (3pt) (1,0) circle (3pt) (2,1) circle (3pt)  (1,2) circle (3pt) (0,1.5) circle (3pt) (-1,-1) circle (3pt) (2,-1) circle (3pt) (3,1) circle (3pt)  (1,3) circle (3pt) (-1,2) circle (3pt) ;
        \node (a1) at (-1,-1)[below]{b} ; 
       \node (a1) at (2,-1)[below]{a} ; 
       \node (a1) at (3,1)[right]{d} ;
       \node (a1) at (1,3)[above]{c} ; 
       \node (a1) at (-1,2)[above]{e} ; 
       \draw [dashed,->,>=Triangle](0,0)--(1,0);
       \draw [dashed,<-,>=Triangle](1,0)--(2,1);
       \draw (2,1)--(1,2);
       \draw (1,2)--(0,1.5);
       \draw (0,1.5)--(0,0);
       \draw (0,0)--(-1,-1);
       \draw (1,0)--(2,-1);
       \draw (2,1)--(3,1);
       \draw (1,2)--(1,3);
       \draw (0,1.5)--(-1,2);
    \end{tikzpicture}}
    \subfigure[$N_3$]{
    \begin{tikzpicture}[scale=0.5]
       \fill[black] (0,0) circle (3pt) (1,1) circle (3pt) (2,1) circle (3pt)  (3,1) circle (3pt) (4,0) circle (3pt) (1,-1) circle (3pt) (3,-1) circle (3pt) (1,2) circle (3pt) (2,2) circle (3pt)  (3,2) circle (3pt) (3,-2) circle (3pt) (1,-2) circle (3pt);
        \node (a1) at (1,2)[above]{b} ; 
       \node (a1) at (2,2)[above]{a} ; 
       \node (a1) at (3,2)[above]{c} ;
       \node (a1) at (1,-2)[below]{d} ; 
       \node (a1) at (3,-2)[below]{e} ; 
\draw (0,0)--(4,0);
\draw (0,0)--(1,1);
\draw [dashed,->,>=Triangle](1,1) --(2,1);
\draw [dashed,<-,>=Triangle](2,1) --(3,1);
\draw (3,1)--(4,0);
\draw [dashed,->,>=Triangle](0,0)--(1,-1);
\draw [dashed,<-,>=Triangle](1,-1)--(3,-1);
\draw (3,-1)--(4,0);
\draw (1,1)--(1,2);
\draw (2,1)--(2,2);
\draw (3,1)--(3,2);
\draw (1,-1)--(1,-2);
\draw (3,-1)--(3,-2);
    \end{tikzpicture}}
    \subfigure[$N_4$]{
    \begin{tikzpicture}[scale=0.5]
       \fill[black] (0,0) circle (3pt) (1,0) circle (3pt) (2,1) circle (3pt)  (1,2) circle (3pt) (0,1.5) circle (3pt) (-1,-1) circle (3pt) (2,-1) circle (3pt) (3,1) circle (3pt)  (1,3) circle (3pt) (-1,2) circle (3pt) ;
        \node (a1) at (-1,-1)[below]{c} ; 
       \node (a1) at (2,-1)[below]{a} ; 
       \node (a1) at (3,1)[right]{b} ;
       \node (a1) at (1,3)[above]{d} ; 
       \node (a1) at (-1,2)[above]{e} ; 
       \draw [dashed,->,>=Triangle](0,0)--(1,0);
       \draw [dashed,<-,>=Triangle](1,0)--(2,1);
       \draw (2,1)--(1,2);
       \draw (1,2)--(0,1.5);
       \draw (0,1.5)--(0,0);
       \draw (0,0)--(-1,-1);
       \draw (1,0)--(2,-1);
       \draw (2,1)--(3,1);
       \draw (1,2)--(1,3);
       \draw (0,1.5)--(-1,2);
    \end{tikzpicture}}
    \caption{Two pairs $(N_1,N_2)$ and $(N_3,N_4)$ which suggest that the assumptions in the second statement of \Cref{theorem:comparison-level-2-and-level-1} are necessary.}
    \label{fig:example of lemma 6.8}
\end{figure}

For the pair $(N_1,N_2),$ $|r(N_1)|=|r(N_2)|=1$ and $r(N_1)\neq r(N_2)$ which violate the first assumption. If $S\in\{\{a,b,c,d\},\{a,b,c,e\},\{a,b,d,e\}\}$, then $N_1|_S$ is a 4-leaf simple nice strict level-2 network and $N_2|_S$ is a simple level-1 network. Thus, $V_{N_1}\not\subseteq V_{N_2}$, which is the statement already implied by the first part of \Cref{theorem:comparison-level-2-and-level-1}. If $S=\{b,c,d,e\}$, then $N_1|_S$ is not nice. If $S$ does not contain $b$, then $S=\{a,c,d,e\}$. We then obtain that $N_1|_S=N_2|_S$, which is a 4-leaf 4-cycle network. Thus, we can not justify the reversed non-containment $V_{N_2|_S}\not\subseteq V_{N_1|_S}$ using a subset $S$ of size 4.

For the pair $(N_3,N_4),$ $|r(N_3)|=2, |r(N_4)|=1$, and $r(N_4)\subseteq r(N_3)$ which violate the second assumption. If $S\in \{\{a,b,c,d\},\{a,b,d,e\}\}$, then  $N_3|_S$ is a 4-leaf strict simple nice level-2 network and $N_4|_S$ is a simple level-1 network. Thus, $V_{N_3}\not\subseteq V_{N_4}$, which is the statement already implied by the first part of \Cref{theorem:comparison-level-2-and-level-1}. If $S=\{a,c,d,e\}$, then $N_3|_S$ is not nice. If $S=\{a,b,c,e\}$, then $N_3|_S=N_4|_S$, which is a 4-leaf 4-cycle network. If $S=\{b,c,d,e\}$, then $N_3|_S$ is a 4-leaf 4-cycle network and $N_4|_S$ is a 4-leaf tree which again imply that $V_{N_3}\not\subseteq V_{N_4}$. Thus, we can not justify the reversed non-containment $V_{N_4|_S}\not\subseteq V_{N_3|_S}$ using a subset $S$ of size 4.
\end{example}

 \begin{example}\label{example: r(N1)=r(N_2)=2}
 In this example, we consider two distinct simple nice strict level-2 networks $N_1$ and $N_2$ in \Cref{fig:example r(N_1)=r(N_2)}. We can see that $r(N_1)=r(N_2)=\{a,b\}$. If $S=\{a,b,c,d\}$, then $N_1|_S$ and $N_2|_S$ are both a simple nice strict level-2 network of type 3. If $S\in \{\{a,b,c,e\},\{a,b,d,e\}\}$, then neither $N_1|_S$ nor $N_2|_S$ are nice. If $S\in\{\{a,c,d,e\},\{b,c,d,e\}\}$, then $N_1|_S=N_2|_S$, which is a 4-leaf 4-cycle network. Thus, in this case, we can not say anything about their distinguishability.
 \begin{figure}[H]
     \centering
     \subfigure[$N_1$]{
     \begin{tikzpicture}[scale=0.4]
       \fill[black] (0,0) circle (3pt) (1,2) circle (3pt) (3,2) circle (3pt) (4,0) circle (3pt) (1,-1) circle (3pt) (3,-1) circle (3pt) (1,3) circle (3pt)  (3,3) circle (3pt) (1,-2) circle (3pt) (3,-2) circle (3pt) (2,0) circle (3pt) (2,1) circle (3pt);
        \node (a1) at (1,3)[above]{$a$} ; 
       \node (a1) at (3,3)[above]{$c$} ; 
       \node (a1) at (1,-2)[below]{$d$} ;
       \node (a1) at (3,-2)[below]{$b$} ;
       \node (a1) at (2,1)[right]{$e$} ;
\draw (0,0)--(4,0);
\draw [dashed,->,>=Triangle](0,0)--(1,2);
\draw [dashed,<-,>=Triangle](1,2) --(3,2);
\draw (3,2) --(4,0);
\draw (0,0)--(1,-1);
\draw [dashed,->,>=Triangle](1,-1)--(3,-1);
\draw [dashed,<-,>=Triangle](3,-1)--(4,0);
\draw (1,2)--(1,3);
\draw (3,2)--(3,3);
\draw (1,-1)--(1,-2);
\draw (3,-1)--(3,-2);
\draw (2,0)--(2,1);
    \end{tikzpicture}}
    \subfigure[$N_2$]{
     \begin{tikzpicture}[scale=0.4]
       \fill[black] (0,0) circle (3pt) (1,2) circle (3pt) (3,2) circle (3pt) (4,0) circle (3pt) (1,-1) circle (3pt) (3,-1) circle (3pt) (1,3) circle (3pt)  (3,3) circle (3pt) (1,-2) circle (3pt) (3,-2) circle (3pt) (2,0) circle (3pt) (2,1) circle (3pt);
        \node (a1) at (1,3)[above]{$a$} ; 
       \node (a1) at (3,3)[above]{$d$} ; 
       \node (a1) at (1,-2)[below]{$c$} ;
       \node (a1) at (3,-2)[below]{$b$} ;
       \node (a1) at (2,1)[right]{$e$} ;
\draw (0,0)--(4,0);
\draw [dashed,->,>=Triangle](0,0)--(1,2);
\draw [dashed,<-,>=Triangle](1,2) --(3,2);
\draw (3,2) --(4,0);
\draw (0,0)--(1,-1);
\draw [dashed,->,>=Triangle](1,-1)--(3,-1);
\draw [dashed,<-,>=Triangle](3,-1)--(4,0);
\draw (1,2)--(1,3);
\draw (3,2)--(3,3);
\draw (1,-1)--(1,-2);
\draw (3,-1)--(3,-2);
\draw (2,0)--(2,1);
    \end{tikzpicture}}
     \caption{Two networks $N_1$ and $N_2$ with $r(N_1)=r(N_2)$ and $|r(N_i)|=2$.}
     \label{fig:example r(N_1)=r(N_2)}
 \end{figure}
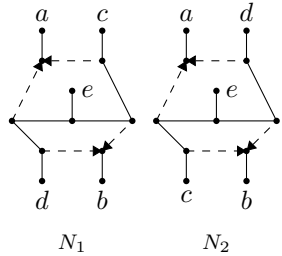

 \end{example}

\begin{example}\label{example: counterexample second statement}
In this example, we consider two networks in \Cref{fig: not N_2-RC}. It can be seen that $|r(N_1)|>|r(N_2)|$ and $r(N_2)\subseteq r(N_1)$. The pair $(N_1,N_2)$ is clearly $N_2$-RC but not $N_1$-RC.
If $S=\{a,b,c,d\}$, then $N_2|_S$ is not nice. If $S=\{a,b,d,e\}$, then $N_1|_S$ is not nice. If $S\in\{\{a,b,c,e\}, \{a,c,d,e\}\}$, then both $N_1|_S$ and $N_2|_S$ are  4-leaf simple nice strict level-2 network of types which do not allow us to conclude any varieties non-containment \Cref{lemma:containment level-2 4 leaves}. If $S=\{b,c,d,e\}$, then both $N_1|_S=N_2|_S$, which is a 4-leaf 4-cycle network.

 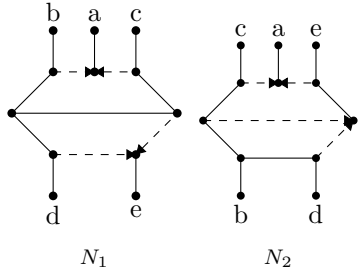
\begin{figure}[H]
     \centering
      \subfigure[$N_1$]{
    \begin{tikzpicture}[scale=0.55]
       \fill[black] (0,0) circle (3pt) (1,1) circle (3pt) (2,1) circle (3pt)  (3,1) circle (3pt) (4,0) circle (3pt) (1,-1) circle (3pt) (3,-1) circle (3pt) (1,2) circle (3pt) (2,2) circle (3pt)  (3,2) circle (3pt) (3,-2) circle (3pt) (1,-2) circle (3pt);
        \node (a1) at (1,2)[above]{b} ; 
       \node (a1) at (2,2)[above]{a} ; 
       \node (a1) at (3,2)[above]{c} ;
       \node (a1) at (1,-2)[below]{d} ; 
       \node (a1) at (3,-2)[below]{e} ; 
\draw (0,0)--(4,0);
\draw (0,0)--(1,1);
\draw [dashed,->,>=Triangle](1,1) --(2,1);
\draw [dashed,<-,>=Triangle](2,1) --(3,1);
\draw (3,1)--(4,0);
\draw (0,0)--(1,-1);
\draw [dashed,->,>=Triangle](1,-1)--(3,-1);
\draw [dashed,<-,>=Triangle](3,-1)--(4,0);
\draw (1,1)--(1,2);
\draw (2,1)--(2,2);
\draw (3,1)--(3,2);
\draw (1,-1)--(1,-2);
\draw (3,-1)--(3,-2);
    \end{tikzpicture}}
     \subfigure[$N_2$]{
    \begin{tikzpicture}[scale=0.5]
       \fill[black] (0,0) circle (3pt) (1,1) circle (3pt) (2,1) circle (3pt)  (3,1) circle (3pt) (4,0) circle (3pt) (1,-1) circle (3pt) (3,-1) circle (3pt) (1,2) circle (3pt) (2,2) circle (3pt)  (3,2) circle (3pt) (3,-2) circle (3pt) (1,-2) circle (3pt);
        \node (a1) at (1,2)[above]{c} ; 
       \node (a1) at (2,2)[above]{a} ; 
       \node (a1) at (3,2)[above]{e} ;
       \node (a1) at (1,-2)[below]{b} ; 
       \node (a1) at (3,-2)[below]{d} ; 
\draw [dashed,->,>=Triangle](0,0)--(4,0);
\draw (0,0)--(1,1);
\draw [dashed,->,>=Triangle](1,1) --(2,1);
\draw [dashed,<-,>=Triangle](2,1) --(3,1);
\draw (3,1)--(4,0);
\draw (0,0)--(1,-1);
\draw (1,-1)--(3,-1);
\draw [dashed,->,>=Triangle](3,-1)--(4,0);
\draw (1,1)--(1,2);
\draw (2,1)--(2,2);
\draw (3,1)--(3,2);
\draw (1,-1)--(1,-2);
\draw (3,-1)--(3,-2);
    \end{tikzpicture}}
     \caption{A pair $(N_1,N_2)$ that is $N_2$-RC.}
     \label{fig: not N_2-RC}
 \end{figure}
  
\end{example}
 
\begin{example}\label{example: counterexample first statement}
Let us consider two networks $N_1$ and $N_2$ presented in \Cref{fig:example r(N)=1}. It can be seen that $|r(N_1)|=|r(N_2)|=1$ and $r(N_1)= r(N_2)$. We consider five possible subsets $S\subseteq \{a,b,c,d,e\}$ of size four. If $S=\{a,b,c,d\}$, then $N_1|_S$ is a $4$-leaf 4-cycle network while  $N_2|_S$ is a $4$-leaf simple nice strict level-2 network. Thus, by \Cref{lemma:containment level-2 4 leaves}, $V_{N_2|_S}\not\subseteq V_{N_1|_S}$. If $S=\{a,c,d,e\}$, then $N_1|_S$ is a single-triangle network while $N_2|_S$ is a 4-leaf 4-cycle network. Thus, $V_{N_2|_S}\not\subseteq V_{N_1|_S}$. If $S= \{a,b,c,e\},$ then $N_1|_S$ and $N_2|_S$ are both simple strict level-2 networks of type 2. If $S\in\{ \{a,b,d,e\},\{b,c,d,e\}\},$ then $N_3|_S$ is simple level-2 of type 2 and $N_4|_S$ is simple level-2 of type 4. Again, under the JC model, $V_{N_2|_S}\not\subseteq V_{N_1|_S}$.
 \begin{figure}[H]
    \centering
    \subfigure[$N_1$]{
    \begin{tikzpicture}[scale=0.5]
       \fill[black] (0,0) circle (3pt) (1,1.5) circle (3pt) (2,1.5) circle (3pt)  (3,1.5) circle (3pt) (4,1.5) circle (3pt) (4,3) circle (3pt) (5,0) circle (3pt) (2.5,-1) circle (3pt) (1,3) circle (3pt) (2,3) circle (3pt)  (3,3) circle (3pt) (2.5,-2) circle (3pt)  ;
        \node (a1) at (1,3)[above]{a} ; 
       \node (a1) at (2,3)[above]{b} ; 
       \node (a1) at (3,3)[above]{c} ;
       \node (a1) at (4,3)[above]{d} ;
       \node (a1) at (2.5,-2)[below]{e} ; 
\draw [dashed,->,>=Triangle](0,0)--(5,0);
\draw (0,0)--(1,1.5);
\draw [dashed,->,>=Triangle](1,1.5) --(2,1.5);
\draw [dashed,<-,>=Triangle](2,1.5) --(3,1.5);
\draw (3,1.5)--(4,1.5);
\draw (4,1.5)--(5,0);
\draw (0,0)--(2.5,-1);
\draw [dashed,->,>=Triangle](2.5,-1)--(5,0);
\draw (1,1.5)--(1,3);
\draw (2,1.5)--(2,3);
\draw (3,1.5)--(3,3);
\draw (4,1.5)--(4,3);
\draw (2.5,-1)--(2.5,-2);
    \end{tikzpicture}}
    \subfigure[$N_2$]{
    \begin{tikzpicture}[scale=0.5]
       \fill[black] (0,0) circle (3pt) (1,1.5) circle (3pt) (2,1.5) circle (3pt)  (3,1.5) circle (3pt) (4,0) circle (3pt) (2,-1) circle (3pt) (1,3) circle (3pt) (2,3) circle (3pt)  (3,3) circle (3pt) (2,-2) circle (3pt) (2,0) circle (3pt) (2,0.5) circle (3pt) ;
        \node (a1) at (1,3)[above]{a} ; 
       \node (a1) at (2,3)[above]{b} ; 
       \node (a1) at (3,3)[above]{c} ;
       \node (a1) at (2,0.5)[above]{d} ; 
       \node (a1) at (2,-2)[below]{e} ; 
       \draw (0,0)--(2,0);
\draw [dashed,->,>=Triangle](2,0)--(4,0);
\draw (0,0)--(1,1.5);
\draw [dashed,->,>=Triangle](1,1.5) --(2,1.5);
\draw [dashed,<-,>=Triangle](2,1.5) --(3,1.5);
\draw (3,1.5)--(4,0);
\draw (0,0)--(2,-1);
\draw [dashed,->,>=Triangle](2,-1)--(4,0);
\draw (1,1.5)--(1,3);
\draw (2,1.5)--(2,3);
\draw (3,1.5)--(3,3);
\draw (2,-1)--(2,-2);
\draw (2,0)--(2,0.5);
    \end{tikzpicture}}
    \caption{A pair of networks $(N_1,N_2)$ which are simple nice strict level-2 networks with only one reticulation leaf.}
    \label{fig:example r(N)=1}
\end{figure}
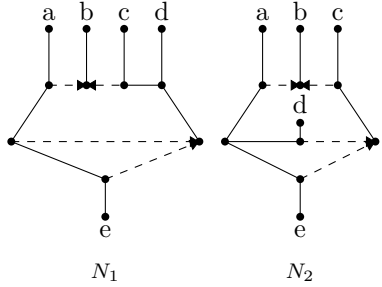
\end{example}

\begin{example}\label{example:R(N1)neq  R(N2)}
In this example, we consider two semisimple networks presented in \Cref{fig:example r=2 semisimple}. It can be seen that $R(N_1)=\{\{a\},\{d,e\}\}$ and $R(N_2)=\{\{a\},\{c\}\}.$ The pair $(N_1,N_2)$ is neither $N_1$-RBC nor $N_2$-RBC. If $S=\{\{a,b,d,e\},\{a,c,d,e\}\}$, then either $N_1|_S$ or $N_2|_S$ is not nice. So we are unable to use Fourier transform to distinguish these two subnetworks.  If $S=\{\{a,b,c,d\},\{a,b,c,e\}\}$, then again, from our results on 4-leaf networks, we are unable to distinguish the corresponding two subnetworks. If $S=\{b,c,d,e\},$ then $N_1|_S$ and $N_2|_S$ are both single-triangle network but they have the same undirected network topologies.
 \begin{figure}[H]
     \centering
      \subfigure[$N_1$]{
    \begin{tikzpicture}[scale=0.5]
       \fill[black] (0,0) circle (3pt) (1,1) circle (3pt)  (3,1) circle (3pt) (4,0) circle (3pt) (1,-1) circle (3pt) (3,-1) circle (3pt) (1,2) circle (3pt)  (3,2) circle (3pt) (2,-2) circle (3pt) (1,-2) circle (3pt) (4,-2) circle (3pt) (3,-1.5) circle (3pt);
        \node (a1) at (1,2)[above]{a} ; 
       \node (a1) at (3,2)[above]{b} ;
       \node (a1) at (1,-2)[below]{c} ; 
       \node (a1) at (2,-2)[below]{d} ; 
       \node (a1) at (4,-2)[below]{e} ; 
\draw (0,0)--(4,0);
\draw [dashed, ->,>=Triangle](0,0)--(1,1);
\draw [dashed,<-,>=Triangle](1,1) --(3,1);
\draw (3,1)--(4,0);
\draw (0,0)--(1,-1);
\draw [dashed,->,>=Triangle](1,-1)--(3,-1);
\draw [dashed,<-,>=Triangle](3,-1)--(4,0);
\draw (1,1)--(1,2);
\draw (3,1)--(3,2);
\draw (1,-1)--(1,-2);
\draw (3,-1)--(3,-1.5);
\draw (2,-2)--(3,-1.5);
\draw (4,-2)--(3,-1.5);
    \end{tikzpicture}}
     \subfigure[$N_2$]{
    \begin{tikzpicture}[scale=0.5]
       \fill[black] (0,0) circle (3pt) (1,1) circle (3pt)  (3,1) circle (3pt) (4,0) circle (3pt) (1,-1) circle (3pt) (3,-1) circle (3pt) (1,2) circle (3pt)  (3,2) circle (3pt) (2,-2) circle (3pt) (1,-2) circle (3pt) (4,-2) circle (3pt) (3,-1.5) circle (3pt);
        \node (a1) at (1,2)[above]{b} ; 
       \node (a1) at (3,2)[above]{a} ;
       \node (a1) at (1,-2)[below]{c} ; 
       \node (a1) at (2,-2)[below]{d} ; 
       \node (a1) at (4,-2)[below]{e} ; 
\draw (0,0)--(4,0);
\draw (0,0)--(1,1);
\draw [dashed,->,>=Triangle](1,1) --(3,1);
\draw [dashed, <-,>=Triangle](3,1)--(4,0);
\draw [dashed,->,>=Triangle](0,0)--(1,-1);
\draw [dashed,<-,>=Triangle](1,-1)--(3,-1);
\draw (3,-1)--(4,0);
\draw (1,1)--(1,2);
\draw (3,1)--(3,2);
\draw (1,-1)--(1,-2);
\draw (3,-1)--(3,-1.5);
\draw (2,-2)--(3,-1.5);
\draw (4,-2)--(3,-1.5);
    \end{tikzpicture}}
     \caption{A pair of networks such that $\{a\}$ is a reticulation branch in both networks.}
     \label{fig:example r=2 semisimple}
 \end{figure}
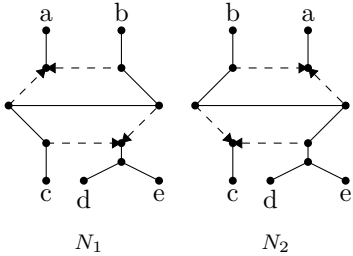
\end{example}

 \smallskip

\smallskip

\noindent \textbf{Acknowledgements.} The author would like to thank his PhD supervisor, Kaie Kubjas for extremely valuable comments and discussions and Annachiara Korchmaros for comments on the earlier version of the manuscript.

\bibliographystyle{plain}
\bibliography{references}

\smallskip

\smallskip

\noindent Author's affiliation:

\vspace{0.2cm}
\noindent Muhammad Ardiyansyah, Department of Mathematics and Systems Analysis, Aalto University,\newline \texttt{muhammad.ardiyansyah@aalto.fi}

\end{document}